\documentclass[a4paper,11pt]{article}

\usepackage[utf8]{inputenc}
\usepackage[british]{babel}
\usepackage[mono=false]{libertine}
\usepackage[T1]{fontenc}
\usepackage{amsthm}
\usepackage[top=2.5cm, bottom=2.5cm, left=2cm, right=2cm]{geometry}
\usepackage{xcolor}
\usepackage[font=small]{caption}
\usepackage{enumitem}

\newtheorem{thm}{Theorem}
\newtheorem{lemma}{Lemma} 
\newtheorem{proposition}{Proposition}

\theoremstyle{definition}
\newtheorem{remark}{Remark}

\usepackage{titlesec}
\makeatletter
\newcommand{\setappendix}{Appendix~\thesection:~~}
\newcommand{\setsection}{\thesection~~}
\titleformat{\section}{\bfseries\LARGE}{%
	\ifnum\pdfstrcmp{\@currenvir}{appendices}=0
	\setappendix
	\else
	\setsection
\fi}{0em}{}
\makeatother
\usepackage[titletoc]{appendix}
\usepackage[pdftex]{graphicx}
\graphicspath{./figures}
\usepackage{array}
\usepackage{url}
\usepackage{mathtools}
\usepackage{amssymb,amsfonts,amsmath}
\usepackage{dsfont}
\usepackage{stmaryrd}
\usepackage{bm}
\usepackage{footnote}

\newcommand*{\QEDA}{\hfill\ensuremath{\blacksquare}}

\usepackage[pdfpagemode=UseNone,bookmarksopen=false,colorlinks=true,urlcolor=blue,citecolor=magenta,citebordercolor=blue,linkcolor=blue]{hyperref}
\DeclareUnicodeCharacter{00A0}{~}

\def \({\left(}
\def \){\right)}
\def \[{\left[}
\def \]{\right]}
\def \nn{\nonumber \\}
\newcommand{\tbf}[1]{{\mathbf{#1}}}

\newcommand{\defeq}{\vcentcolon=}

\newcommand{\bY}{{\mathbf {Y}}}

\newcommand{\bW}{{\mathbf {W}}}
\newcommand{\bZ}{{\mathbf {Z}}}

\newcommand{\bw}{{\mathbf {w}}}

\newcommand{\bX}{{\mathbf {X}}}
\newcommand{\bx}{{\mathbf {x}}}

\newcommand{\by}{{\mathbf {y}}}
\newcommand{\bz}{{\mathbf {z}}}

\newcommand{\bs}{{\mathbf {s}}}

\newcommand{\bS}{{\mathbf {S}}}

\newcommand{\be}{\begin{equation}}
\newcommand{\ee}{\end{equation}}

\newcommand{\bea}{\begin{align}}
\newcommand{\eea}{\end{align}}

\DeclareMathAlphabet{\varmathbb}{U}{bbold}{m}{n}

\newcommand{\EE}{\mathbb{E}}

\usepackage{notations}
\usepackage{setspace}

\begin{document}
\title{The adaptive interpolation method: A simple scheme \\ to prove replica formulas in Bayesian inference}
\author{Jean Barbier$^{\dagger,\star}$ and Nicolas Macris$^{\dagger}$}
\date{}
\maketitle
{\let\thefootnote\relax\footnote{
\!\!\!\!\!\!\!\!\!\!$\dagger$ Laboratoire de Th\'eorie des Communications, Facult\'e Informatique et Communications, Ecole Polytechnique F\'ed\'erale de Lausanne, CH-1015, Suisse.\\
$\star$ International Center for Theoretical Physics, Strada Costiera, 11
I, 34151 Trieste, Italy.
}}
\setcounter{footnote}{0}
\begin{abstract}
In recent years important progress has been achieved towards proving the validity of the replica predictions 
for the (asymptotic) mutual information (or ``free energy'') in Bayesian inference problems. The proof techniques that have emerged appear to be quite general, 
despite they have been worked out on a case-by-case basis. Unfortunately, a common point between all these schemes is their relatively high level of technicality.
We present a new proof scheme that is quite straightforward with respect to the previous ones. 
We call it the {\it adaptive interpolation method} because it can be seen as an extension of the interpolation method developped by Guerra and Toninelli 
in the context of spin glasses, with an interpolation path that is adaptive. In order to 
illustrate our method we show how to prove the replica formula for three non-trivial inference problems. The first one is 
symmetric rank-one matrix estimation (or factorisation), which is the simplest problem considered here and the one for which the method is 
presented in full details. Then we generalize to symmetric tensor estimation and random linear estimation. We believe that the present method has a much wider range of applicability and also sheds new insights on the reasons for 
the validity of replica formulas in Bayesian inference.
\end{abstract}

{
	\singlespacing
	\hypersetup{linkcolor=black}
	\tableofcontents
}

\section{Introduction}

A very interesting development in probability theory in recent years has been the progress 
on a coherent mathematical theory \cite{Talagrand2003spin,Talagrand2011spina,Talagrand2011spinb,Panchenko2013} of 
the predictions of the replica and cavity methods \cite{mezard1990spin} in statistical physics of 
spin glasses. In this respect one of the most important tools is the invention of the {\it interpolation method}
by Guerra and Toninelli \cite{Guerra-2003,Guerra-Toninelli-2002} which eventually led Talagrand to a remarkable proof \cite{Talagrand-annals-2006}
of the Parisi formula \cite{Parisi-1980} for the free energy of the Sherrington-Kirkpatrick model \cite{Sherrington-Kirkpatrick-1975}. 

In more recent years 
the interpolation method has been fruitfully extended and adapted to problems of interest in a wide range of applications such as in coding theory, communications, signal processing and theoretical computer science, well beyond the realm 
of traditional statistical mechanics. Among these we highlight applications of the interpolation method 
to error correcting codes \cite{Montanari-codes,GKSmacris2007,Macris-bounds-codes,Kudekar-Macris-2009}, random linear 
estimation and compressive sensing \cite{MacrisKoradaAllerton2007,KoradaMacris_CDMA,barbier_allerton_RLE,barbier_ieee_replicaCS,Barbier_IMMSE_subExt}, low-rank matrix and tensor factorization \cite{korada2009exact,krzakala2016mutual,XXT}
and constraint satisfaction problems \cite{Franz-Leone-2003,Franz-Leone-Toninelli,Panchenko-Talagrand-2004,Macris-Hassani-Urbanke}. Most of these problems
are {\it inference problems} and when 
a Bayesian framework is adopted, they can be solved with a replica {\it symmetric} scheme (constraint satisfaction is not, as such at least, an inference problem and 
does not fall in this category). The replica symmetric formulas for the free energies, mutual informations and error performance measures typically predict interesting {\it first order} phase transitions, with associated 
``metastable states with infinite lifetime'', which pose interesting algorithmic challenges of great importance in practical applications as well as challenges from the analysis point of view. It has turned out that one can learn a great deal about the 
fundamental limitations for important classes of (message-passing) algorithms by studying these replica solutions 
(we refer to \cite{MezardMontanari09} for a general reference and come back to this point in the conclusion). 

In spite of their complexity, for all the inference problems cited above, complete proofs of the replica symmetric formulas have been found. These proofs usually combine Guerra-Toninelli interpolation bounds 
with some other non-trivial idea or method, namely {\it algorithmic} approaches involving so-called {\it spatially coupled} models
\cite{Giurgiu_SCproof,XXT,barbier_allerton_RLE,barbier_ieee_replicaCS}, 
{\it information theoretic} methods \cite{private,ReevesP16} or rigorous versions of the {\it cavity method} \cite{2016arXiv161103888L,2017arXiv170108010L,2017arXiv170200473M,coja-2016} 
using the Aizenman-Sims-Starr principle 
\cite{aizenman2003extended}. 
While each of these methods has its own merit and sheds interesting light,
they all lead to quite long and technically involved proofs. Besides, although each method can probably be taylored for each problem, it would clearly 
be more satisfactory to have a more or less unified approach.

In this paper we develop a new unified and self-contained interpolation method. We illustrate how it works for three different problems, 
namely rank-one symmetric matrix and tensor factorization, as well as random linear estimation and compressive sensing. Our method allows to prove at the same time matching lower {\it and} upper bounds on the free energy with much less effort than all known current proofs. All these problems 
are ``spin systems'' defined for ``dense graphs'' (complete graphs or hypergraphs). The ideas of this paper can also be adapted to error correcting codes that are 
akin to spin systems on ``sparse'' random graphs and we plan to come back to this aspect elsewhere\footnote{Since the first version of this manuscript, the method has been successfully applied to many other problems including non-symmetric matrix and tensor factorization \cite{2017arXiv170910368B}, generalized linear models and learning \cite{barbier2017phase}, models of deep neural networks \cite{2018arXiv180509785G,committee_nips}, random linear estimation with structured matrices \cite{toappear} and even problems defined by sparse graphical models such as the censored block model \cite{toappear_sparse}.}. 

Roughly speaking,
our new scheme interpolates between the original problem and the mean-field replica solution in small steps, each step involving its own set 
of trials parameters and Gaussian mean-fields in the spirit of Guerra and Toninelli (this idea of interpolating in small steps originated 
in the {\it sub-extensive interpolation method} developed by the authors in \cite{barbier_ieee_replicaCS,Barbier_IMMSE_subExt}). We are then able to choose the set of trial parameters
in various ways so that we get both upper and lower bounds that eventually match. One can interpret the set of trial parameters as a suitable ``interpolation path'' that we ``adapt''
to obtain suitable bounds, and thus we call this method the {\it adaptive interpolation method}.\footnote{In the present formulation one can also interpret the succession 
of Gaussian mean-fields in each step as a Wiener process. For this reason we initially called this new 
approach ``the stochastic interpolation method''. The interpretation in terms of a Wiener process is in fact not really needed, and here we choose a more
pedestrian path, but we believe this is an aspect of the method that may be of further interest (specially for diluted systems) 
and briefly discuss it in Appendix \ref{interpretation}.}

An important aspect of our method is the need for concentration properties of the suitable ``overlap''. 
It was already proven long ago in \cite{Pastur-Shcherbina-1991,Pastur-Shcherbina-Tirozzi-1994} that a concentration hypothesis for overlaps
implies that the replica symmetric solution is exact (an implication that was known to physicists). 
However for typical spin glass systems 
(e.g. the Sherrington-Kirkpatrick or $p$-spin spin glass) this hypothesis can only hold in some high temperature 
phase, and it is also difficult to prove. We refer to \cite{Shcherbina-1997,Pastur-Shcherbina-Tirozzi-1994} and \cite{Talagrand2003spin} for pioneering works on such proofs with the help of cavity-like methods. 
In the framework of 
Bayesian inference the situation is more favourable. The Bayes rule immediately implies a special set of identities obeyed by suitable ``correlation functions''
often known as Nishimori identities \cite{NishimoriBook01,iba1999nishimori}.
These identities then allow to deduce the concentration of overlaps from the concentration of the free energy 
in the {\it whole} phase diagram. This is also the reason why Bayesian inference problems generally
lead to replica symmetric solutions. 

The paper is organized as follows. Section \ref{stochInt} gives a pedagogic introduction to the adaptive interpolation method for one 
of the simplest, yet non-trivial problems, namely
rank-one symmetric matrix factorization. 
The replica symmetric formula for the free energy or mutual information is completely proven
in a self contained and direct way (see Theorem \ref{thm1}).
As explained in the previous paragraph, for all these problems our analysis also rests on concentration 
properties of the overlap parameters in the whole phase diagram (Lemma \ref{concentration}). 
This analysis is the subject of sections \ref{proofConc}, \ref{fluctuation-identity} and \ref{concentration-free-energy}, and can be read independently from the rest of the paper.  
We then sketch the same method for symmetric tensors (see Theorem \ref{RS_symtensor}). Section \ref{RLE_section} presents the method for a more difficult problem, namely random linear estimation. 
In particular, we provide a much simpler and transparent proof than all other existing proofs \cite{barbier_allerton_RLE,barbier_ieee_replicaCS,private,ReevesP16} of the 
replica formula (see Theorem \ref{rsformulaRLE}). 

%
\section{The adaptive interpolation method: Main ideas} \label{stochInt}
Before starting let us introduce a few notations used all along this paper: Vectorial quantities will be denoted by boldface letters, random variables by capital letters and their realizations 
by small letters. Expectations with respect to ``quenched'' variables (i.e. the variables that are fixed by the realization of the problem) are denoted $\mathbb{E}$ and those  with respect to ``annealed'' variables (i.e. the dynamical variables) are denoted by Gibbs brackets $\langle - \rangle$ possibly with appropriate subscripts. This choice follows the standards of statistical mechanics.
\subsection{Symmetric rank-one matrix estimation: Setting and main result}
\label{sec:partI}
Consider the following probabilistic rank-one matrix estimation
problem: One has access to noisy observations
$\bw=[w_{ij}]_{i,j=1}^n$ of the pair-wise product of the
components of a vector
$\bs=[s_1,\dots, s_n]^\intercal \in  \mathbb{R}^n$ with
i.i.d components distributed as $S_i\sim P_0$, $i=1,\dots,
n$ (that we simply denote $\bS\iid P_0$). A standard and natural setting is the case of additive white Gaussian noise of known variance $\Delta$, 
\begin{align}
w_{ij} =  \frac{s_i s_j}{\sqrt{n}} + z_{ij}\sqrt{\Delta}  \quad \text{for} \quad 1\leq i\leq j\leq n,\label{xx}
\end{align}
where $\tbf{z} = [z_{ij}]_{i,j=1}^n$ is a symmetric matrix with i.i.d entries $Z_{ij}  \sim  \mathcal{N}(0,1)$ for $1\leq i\leq j\leq n$. This is denoted 
$\bZ \iid {\cal N}(0,1)$. The goal is to estimate the ground truth $\bs$ from $\bw$ assuming that both $P_0$ and $\Delta$ are known and independent of $n$ (the noise is symmetric so that $w_{ij} = w_{ji}$). 

We consider a Bayesian setting and associate to the model \eqref{xx} its posterior distribution. The likelihood of the (component-wise independent) 
observation matrix $\bw$ given $\bs$ is 
\begin{align}
P(\bw|\bs) = \frac{\exp\Big\{-\frac{1}{2\Delta}\sum_{i\le j}\big(w_{ij} - \frac{s_is_j}{\sqrt{n}}\big)^2\Big\}}{(2\pi \Delta)^{\frac{n(n+1)}{2}}}.
\end{align} 
From the Bayes formula we then get the posterior distribution\footnote{We abusively use the notation $dxP_0(x)$ even though $P_0$ is not necessarily 
absolutely continuous.} for $\bx=[x_1, \ldots, x_n]^\intercal\in \mathbb{R}^n$ given the observations (it is convenient to explicitely distinguish between the ground truth signal vector $\bs$ and 
its estimate $\bx$ sampled from the posterior)
\begin{align}
P(\bx|\bw) = 
\frac{\prod_{i=1}^n P_0(x_i) P(\bw|\bx)}{\int \big\{\prod_{i=1}^n dx_i P_0(x_i)\big\}P(\bw|\bx)}.
\label{post_}
\end{align}
Replacing the observation $\bw$ by its explicit expression \eqref{xx} as a function of the signal and the noise we obtain
\begin{align}
P\Big(\bx\Big|\bw =  \frac{\bs\bs^{\intercal}}{\sqrt{n}} + \bz\sqrt{\Delta}\Big)=
\frac{\prod_{i=1}^n P_0(x_i) e^{-{\cal H}(\bx;\bs,\bz)}}{\int \big\{\prod_{i=1}^n dx_i P_0(x_i)\big\} e^{-{\cal H}(\bx;\bs,\bz)}} 
\label{post_}
\end{align}
where we call 
\begin{align}
{\cal H}(\bx;\bs,\bz) \defeq \frac{1}{\Delta}\sum_{i\le j=1}^n\Big(\frac{x_i^2x_j^2}{2n} - \frac{x_ix_js_is_j}{n} -\frac{x_ix_jz_{ij} \sqrt{\Delta}}{\sqrt{n}}\Big)
\end{align}
the {\it Hamiltonian} of the model.
In order to obtain the last form of the posterior distribution we replaced $w_{ij}$ using \eqref{xx}, 
developed the square in $P(\bw|\bx)$, and simplified the $\bx$-independent terms in the numerator and denominator.
The 
normalization factor is by definition the {\it partition function} 
\begin{align}
 {\cal Z}(\bs,\bz) \defeq \int \big\{\prod_{i=1}^n dx_i P_0(x_i)\big\} e^{-{\cal H}(\bx;\bs,\bz)}.
\end{align}
Our principal quantity of interest is the {\it average free energy} per component\footnote{For all other models considered in this paper we directly write the explicit expression of the free energy, but the 
derivation is always similar.} 
defined by 
\begin{align}\label{foriginal}
f_n\defeq-\frac1n\EE_{\bS,\bZ}[\ln {\cal Z}(\bS,\bZ)]
\end{align}
%
where $\bS \iid P_0$ and $\bZ \iid {\cal N}(0,1)$. 

Define the {\em replica symmetric (RS) potential} $f_{\rm RS}(m; \Delta)$ as
\begin{align}\label{eq:potentialfunction}
&f_{\rm RS}(m;\Delta) \defeq \frac{m^2}{4\Delta} + f_{\rm den}\big(\Sigma(m;\Delta)\big),
\end{align}
with 
\begin{align}
 \Sigma(m;\Delta) \defeq \sqrt{\frac{\Delta}{m}}.
\end{align}
Here $f_{\rm den}(\Sigma)$ is the free energy associated with a {\it scalar} Gaussian denoising model: $y = s  +  \widetilde z \, \Sigma$ where $S \sim P_0$, $\widetilde Z  \sim {\cal N}(0,1)$. The free energy $f_{\rm den}(\Sigma)$ is minus the average logarithm of the normalization of the posterior distribution $P(x|s  +  \widetilde z\,\Sigma)\propto \exp(-\Sigma^{-2}(x^2/2 -xs -x\widetilde z\,\Sigma))P_0(x)$:
\begin{align}
f_{\rm den}(\Sigma)\defeq-\mathbb{E}_{S,\widetilde Z}\Big[\ln\int dxP_0(x) e^{-\frac{1}{\Sigma^2}\big(\frac{x^2}{2} - xS- x\widetilde Z\,\Sigma \big)}\Big]. \label{MFf}
\end{align}

Our first theorem illustrating the adaptive interpolation method is 
\begin{thm}[RS formula for symmetric rank-one matrix estimation]\label{thm1}
Fix $\Delta > 0$. For any $P_0$ with bounded support, the asymptotic free energy of the symmetric rank-one matrix estimation model \eqref{xx} verifies
\begin{align}\label{rsresult1}
\lim_{n\to \infty} f_n = \min_{m\ge 0}f_{\rm RS}(m;\Delta).	
\end{align}
\end{thm}

\begin{proof}
 The theorem follows from Proposition \ref{UpperBound} in section \ref{secUpper} and Proposition \ref{LowerBound} in section \ref{seclower}.
\end{proof}

The bounded support property hypothesis for $P_0$ is not really a requisite of the adaptive interpolation method, but simply makes the necessary concentration proofs for the free energy simpler. There is no condition on the size 
of the support, and it is presumably possible to take a support equal to the whole real line by a limiting process applied to \eqref{rsresult1}, as long 
as the first four moments of $P_0$ are finite.

Formulas such as \eqref{rsresult1}, where a complicated statistical model is related to a scalar (and thus analyzable) statistical model are at the root of the mean-field theory in statistical mechanics. A possible intuition behind this formula (and all formulas of the same type in this article) is as follows: The estimation problem \eqref{xx} is effectively ``replaced'' by a decoupled estimation model $\by = \bs  +  \widetilde{\bz} \, \Sigma(m;\Delta)$ where the noise variance is perfectly tuned through the minimization problem \eqref{rsresult1} in order to faithfully ``summarize'' the complex interactions among variables in the original model; $\Sigma(m;\Delta)$ thus plays the role of a ``mean-field''. See e.g. \cite{mezard1990spin,mezard2009information} for more details on the mean-field theory and its applications.

This theorem has already been obtained recently in \cite{XXT,2016arXiv161103888L} (with varying hypothesis on $P_0$) 
by the more elaborate methods mentionned in the introduction.
In the next paragraphs we introduce the adaptive interpolation method through a pedagogical and new proof of this theorem. 

\begin{remark}[Free energy, mutual information and algorithms]
In Bayesian inference the average free energy is related to the 
{\it mutual information} $I(\bS;\bW)$ between the observation and the unknown vector 
(which is formally expressed as a difference of Shannon entropies: $I(\bS;\bW) =  H(\bW) - H(\bW |\bS)$). 
For model \eqref{xx}, a straightforward computation shows that when $P_0$ has bounded first four moments
\begin{align}
\frac{I(\bS;\bW)}{n} = f_n + \frac{\mathbb{E}[S^2]^2}{4\Delta} + {\cal O}(n^{-1}), \label{fMI}
\end{align}
where $S \sim P_0$. The $n \to \infty$ limit of the mutual 
information (or equivalently of the average free energy) is an interesting object to compute because it allows to locate 
the {\it phase transition(s)} occuring in the inference problem, which corresponds to its non-analyticity point(s) as a function of $\Delta$. 
This phase transition threshold usually separates a 
low-noise regime where inference is information 
theoretically possible from a high-noise regime where inference is impossible.
In this high-noise regime the observation simply does not carry enough information for reconstructing the signal. Furthermore, remarkably, the 
replica formula for the mutual information (or average free energy) also allows to determine an algorithmic noise threshold, below the phase transition threshold,
which separates the information theoretic possible phase in two regions: An ``easy'' phase where there exist low complexity message-passing algorithms for optimal inference and a ``hard'' phase where message-passing algorithms yield suboptimal inference. For further information and rigorous results 
on these issues for model \eqref{xx} we refer to \cite{XXT}. A few more pointers to the literature are given in the conclusion. 
\end{remark}
\begin{remark}[Channel universality]\label{rmk:univers}
The Gaussian noise setting \eqref{xx} is actually sufficient to completely characterize the generic model where the entries of $\tbf{w}$ are observed through a noisy element-wise
(possibly non-linear) output probabilistic channel $P_{\rm out}(w_{ij}|s_is_j/\sqrt{n})$. This is made possible by a theorem of 
channel universality \cite{krzakala2016mutual} (conjectured in \cite{lesieur2015mmse} and already proven for community detection in
\cite{deshpande2015asymptotic}). Roughly speaking this theorem states that given an output channel $P_{\rm out}(w|y)$, such that
at $y=0$ the function $y\mapsto \ln P_{\rm out}(w|y)$ is three times differentiable, 
with bounded second and
third derivatives, then the mutual information satisfies
\begin{align}
I(\tbf{S};\tbf{W} ) = I (\tbf{S}; \tbf{S}\tbf{S}^{\intercal}/\sqrt{n} +\tbf{Z} \sqrt{\Delta}) +\mathcal{O}(\sqrt{n} ),
\end{align}
where $\Delta$ is the inverse Fisher information (at $y = 0$) of the output channel: 
\begin{align*}
\Delta^{-1} \defeq \int dw P_{\rm out}(w|0)(\partial_y \ln P_{\rm out}(w|y)\vert_{y=0})^2.	
\end{align*}
Informally, this means that we only have to compute the mutual information for a Gaussian channel to take care of a wide range of problems, which can be expressed in terms of their Fisher information.
\end{remark}

\subsection{The $(k,t)$--interpolating model}\label{seckinterp}
\begin{figure}[!t]
\centering
\includegraphics[width=1\textwidth]{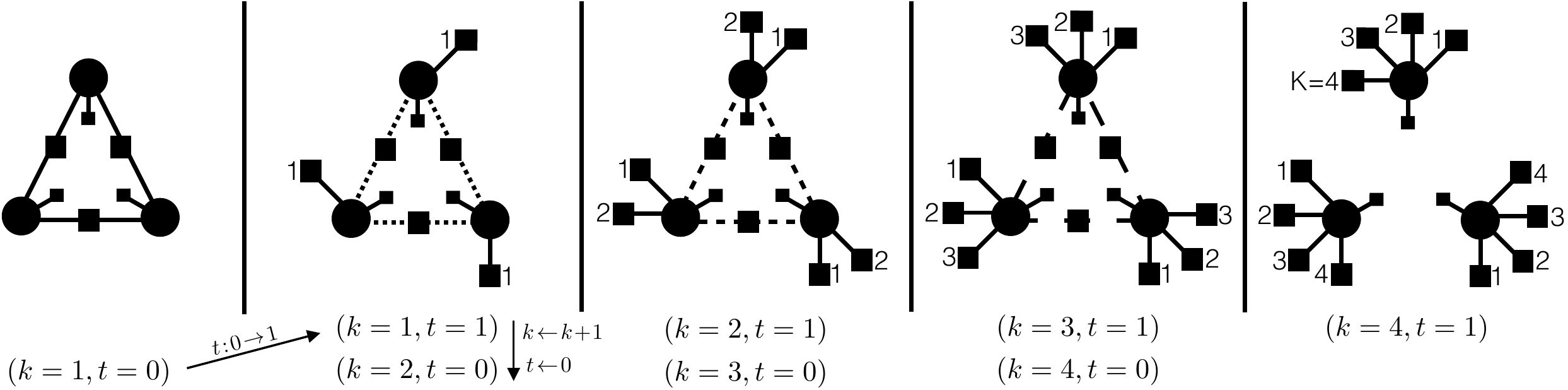}
\vspace*{-18pt}
\caption{Factor graph representation of 
the $(k,t)$--interpolating model $P_{k,t;\epsilon=0}(\bx|\boldsymbol{\theta})$ for $k = 1,\ldots,K = 4$. The adaptive interpolation is parametrized by two distinct ``time'' parameters: A discrete one $k = 1,\ldots,K$ that controls the interpolation at a {\it global} level. 
Then for a fixed $k$ we define a continuous $t \in [0,1]$ that controls the interpolation at a {\it local} level. The adaptive interpolation iteratively ``constructs'' the mean-field (decoupled) model, corresponding to $(k = K,t = 1)$, by starting from the original model $(k = 1,t = 0)$ and then incrementally reducing the interaction strength of the edges associated with the original model, while compensating by adding mean-field decoupled factors to the graph (the small factors acting independently on each nodes that represent the factorized prior $P_0$ remain unchanged). This works as follows. At a fixed discrete step $k$, letting $t$ increase from $0$ to $1$ continuously decreases  the strength of all the interactions of the original model by a factor $K^{-1}$, 
while continuously adding the $k$-th Gaussian mean-field factors (one equivalent factor per node). This corresponds to the {\it local} interpolation. Then $k$ is increased by one unit, $t$ is set to $0$ and the local interpolation process is then repeated. At the end of the adaptive interpolation, the set of all interactions of the original model have been replaced by $K$ Gaussian mean-fields per node. An important matching property is that the $(k,t = 1)$ and $(k + 1, t = 0)$ models are statistically equivalent.}\label{fig:Kinterp}
\end{figure}
Let $\bz^{(k)} = [z_{ij}^{(k)}]_{i,j=1}^n, 
\widetilde{\bz}^{(k)} = [\widetilde{z_{i}}^{(k)}]_{i=1}^n$, $Z_{ij}^{(k)} = Z_{ji}^{(k)} \sim   {\cal N}(0,1)$, 
$\widetilde Z_{i}^{(k)}  \sim   {\cal N}(0,1)$ for $k = 1,\ldots,K$ be Gaussian noise symmetric matrices and vectors. 
It is important to keep in mind that these are indexed both by the 
vertex indices $i, j$ and the discrete global interpolation parameter $k$.  

The {\it $(k,t)$--interpolating Hamiltonian} is
\begin{align}\label{intH}
{\cal H}_{k,t}(\bx)\defeq &\sum_{k'=k+1}^K h\Big(\bx, \bs,\bz^{(k')},K\Delta\Big)+\sum_{k'=1}^{k-1} h_{\rm mf}\Big(\bx, \bs,\widetilde{\bz}^{(k')},\frac{K\Delta}{m_{k'}}\Big)\nn
&\qquad + h\Big(\bx, \bs,\bz^{(k)},\frac{K\Delta}{1-t}\Big)+h_{\rm mf}\Big(\bx, \bs,\widetilde{\bz}^{(k)},\frac{K\Delta}{t\,m_{k}}\Big),
\end{align}
where the trial parameters $\{m_{k}\}_{k=1}^K$ are to be fixed later
(these will be chosen  ${\cal O}(1)$ with respect to (w.r.t) $n$ and can be interpreted as signal-to-noise ratios), $t \in [0,1]$
the continuous local interpolation parameter, and
%
\begin{align}
h(\bx, \bs,\bz,\sigma^2) &\defeq\frac{1}{\sigma^2}\sum_{i\le j=1}^n \Big(\frac{x_i^2x_j^2}{2n} - \frac{x_ix_js_is_j}{n}- \frac{\sigma x_ix_jz_{ij}}{\sqrt{n}}\Big), \label{h}\\
h_{\rm mf}(\bx, \bs,\widetilde{\bz},\sigma^2) &\defeq\frac{1}{\sigma^2}\sum_{i=1}^n \Big(\frac{x_i^2}{2} - x_is_i- \sigma x_i\widetilde z_{i}\Big).
\label{hmf}
\end{align}
Here the subscript ``mf'' stands for ``mean-field''. 

A possible interpretation of the scheme is the following. The $(k,t)$--interpolating model corresponds to the following 
inference model. One has access to the following sets of noisy observations about the signal $\bs$ where each noise realization is independent:
\begin{align}
\Big\{&\bw^{(k')}=\frac{\bs \bs^\intercal}{\sqrt{n}} + \bz^{(k')}\sqrt{K\Delta}\Big\}_{k'=k+1}^K, \label{10_}\\
\Big\{&\by^{(k')}=\bs + \widetilde\bz^{(k')} \sqrt{\frac{K\Delta}{m_{k'}}}\Big\}_{k'=1}^{k-1}, \label{11_}\\
&\bw^{(k)}=\frac{\bs \bs^\intercal}{\sqrt{n}} + \bz^{(k)}\sqrt{\frac{K\Delta}{1-t}},\label{12_}\\
&\by^{(k)}=\bs + \widetilde\bz^{(k)} \sqrt{\frac{K\Delta}{t\,m_{k}}}. \label{13_}
\end{align}
The first and third sets of observations correspond to similar inference channels as the original model \eqref{xx} but with 
a much higher noise variance proportional to $K$. These correspond to the 
first and third terms, respectively, of the $(k,t)$--interpolating Hamiltonian \eqref{intH}. The second and fourth sets 
instead correspond to decoupled Gaussian denoising models, with associated ``mean-field'' second and fourth terms in \eqref{intH}. The noise 
variances are proportional to $K$ because the total number of observations is $K$ and we 
want the total signal-to-noise ratio to be ${\cal O}(1)$. At fixed $k$, letting $t$ increase
from $0$ to $1$ increases by one unit the number of decoupled observations \eqref{11_} by continuously 
adding the observation \eqref{13_}: Its signal-to-noise ratio that 
vanishes at $t = 0$ (which is equivalent to not having access to this observation) becomes 
finite and equal to the signal-to-noise ratio of the individual observations in the set \eqref{11_} at $t = 1$. 
Simultaneously it reduces by one the number of observations of the form \eqref{10_} by ``removing'' the 
observation \eqref{12_}: its signal-to-noise ratio, which is 
finite at $t = 0$, vanishes at $t = 1$. From \eqref{10_}--\eqref{13_} it is clear 
that the $(k,t = 1)$ and $(k + 1,t = 0)$--interpolating models are 
statistically equivalent. A complementary and more graphical illustration of the 
interpolation scheme is found on Figure~\ref{fig:Kinterp}.

In order to use an important concentration lemma later on, we will need a slightly more general Hamiltonian, and consider the following perturbed version of \eqref{intH}:
\begin{align}
{\cal H}_{k,t;\epsilon}(\bx;\boldsymbol{\theta}) \defeq {\cal H}_{k,t}(\bx) 
+ \epsilon\sum_{i=1}^n\Big(\frac{x_i^2}{2} -x_is_i -\frac{x_i\widehat z_i}{\sqrt{\epsilon}}\Big),  
\label{perturb}
\end{align}
with i.i.d $\widehat Z_i \sim {\cal N}(0,1)$ and $\boldsymbol{\Theta} \defeq \{\bS, \{\bZ^{(k)}, \widetilde{\bZ}^{(k)}\}_{k=1}^K, \widehat \bZ\}$ is the collection of all {\it quenched} random variables. It should be kept in mind that the signal-to-noise ratio $\epsilon$ of this additional Gaussian ``side-channel'' $\by=\bs\sqrt{\epsilon} + \widehat \bz$ will tend to $0$ at the end of the proof. Therefore we always consider
$\epsilon \in  [0,1]$.

The $(k,t)$--interpolating model has an associated posterior distribution, {\it Gibbs expectation} $\langle-\rangle_{k,t;\epsilon}$ and {\it $(k,t)$--interpolating free energy} $f_{k,t;\epsilon}$:
\begin{align}
P_{k,t;\epsilon}(\bx|\boldsymbol{\theta}) &\defeq \frac{\prod_{i=1}^n P_0(x_i) e^{-{\cal H}_{k,t;\epsilon}(\bx;\boldsymbol{\theta})}}
{\int\bigl\{\prod_{i=1}^n dx_i P_0(x_i)\bigr\}e^{-{\cal H}_{k,t;\epsilon}(\bx;\boldsymbol{\theta})}}, \label{post}\\
\langle A(\bX)\rangle_{k,t;\epsilon} &\defeq \int d\bx \, A(\bx) P_{k,t;\epsilon}(\bx|\boldsymbol{\theta}), \label{Gibbs}\\
f_{k,t;\epsilon}&\defeq -\frac{1}{n}\mathbb{E}_{\boldsymbol{\Theta}}\Big[\ln\int\big\{\prod_{i=1}^n dx_i P_0(x_i)\big\}e^{-{\cal H}_{k,t;\epsilon}(\bx;\boldsymbol{\Theta})}\Big]. \label{intf}
\end{align}
In the following, we simply denote $\mathbb{E}_{\boldsymbol{\Theta}}$ by $\mathbb{E}$.

\begin{lemma}[Linking the perturbed and plain free energies]\label{thermolimit}
 Let $P_0$ have finite second moment. Then for the initial and final systems
 \begin{align}
  \vert f_{1, 0; \epsilon} - f_{1, 0; 0}\vert \leq \frac{\epsilon}{2} \mathbb{E}[S^2]\,, \quad \text{and}\quad  \vert f_{K, 1; \epsilon} - f_{K, 1; 0}\vert \leq \frac{\epsilon}{2} \mathbb{E}[S^2]\,.
 \end{align}
 \end{lemma}
 
A short and generic proof is found in Appendix \ref{appendix-exchg-lim}. 
This statement shows in particular that if the thermodynamic limit $n\to +\infty$ exists, then it can be exchanged with the limit $\epsilon\to 0_+$ (as long as $P_0$ has bounded second moment). 
We stress that the existence of the thermodynamic limit is {\it not} directly used in our subsequent analysis, but rather, 
follows as a consequence. 

%
\subsection{The initial and final models}\label{initfinalmodels}
Let us compute the $(k,t)$--interpolating free energy $f_{1,0;0}$ associated with the initial $(k = 1,t = 0)$ model. Using \eqref{intH} and \eqref{h}, 
\begin{align}\label{initialmodel}
{\cal H}_{1,0;0}(\bx;\boldsymbol{\theta}) = &\sum_{k=1}^K h\Big(\bx, \bs,\bz^{(k)},K\Delta\Big) = \frac{1}{\Delta}\sum_{i\le j=1}^n\Big(\frac{x_i^2x_j^2}{2n} - \frac{x_ix_js_is_j}{n}- \frac{x_ix_j}{\sqrt{n}} \sqrt{\Delta} \sum_{k=1}^K\frac{z_{ij}^{(k)}}{\sqrt{K}}\Big).
\end{align}
As the $Z_{ij}^{(k)}$, $1 \le i \le j \le n$, are i.i.d ${\cal N}(0,1)$ random variables, they possess the {\it stability property}, namely $Z_{ij} \defeq \sum_{k=1}^K Z_{ij}^{(k)}/\sqrt{K}$ are i.i.d ${\cal N}(0,1)$ random variables as well (and symmetric). Let $\bz = [z_{ij}]_{i,j=1}^n$. Using this we obtain 
\begin{align}
f_{1,0;0}  = -\frac{1}{n}\mathbb{E}_{\bS,\bZ}\Big[\ln\int\big\{\prod_{i=1}^n dx_i P_0(x_i)\big\} e^{-{\cal H}_{1,0;0}(\bx;\bS,\bZ)}\Big]
\end{align} 
which is actually the free energy \eqref{foriginal} of the original model. We thus have:
\begin{align}
f_{1,0;0}=f_n. \label{f10isf}
\end{align}
Let us now consider the free energy $f_{K,1;0}$ of the final model. Using \eqref{intH} and \eqref{hmf} we get
\begin{align}
{\cal H}_{K,1;0}(\bx;\boldsymbol{\theta}) = &\sum_{k=1}^K h_{\rm mf}\Big(\bx, \bs,\widetilde{\bz}^{(k)},\frac{K\Delta}{m_{k}}\Big) = \sum_{k=1}^K\frac{m_{k}}{K\Delta}\sum_{i=1}^n\Big(\frac{x_i^2}{2} - x_is_i- x_i\widetilde{z}_i^{(k)} \sqrt{\frac{K\Delta}{m_{k}}} \Big). \label{HK1_before}
\end{align}
Define 
\begin{align}
 m_{\rm mf}^{(K)} \defeq  \frac1K\sum_{k=1}^K m_{k}.
\end{align}
Simple algebra leads to
\begin{align}
{\cal H}_{K,1;0}(\bx;\boldsymbol{\theta}) = \frac{m_{\rm mf}^{(K)}}{\Delta}\sum_{i=1}^n\Big(\frac{x_i^2}{2} - x_is_i- x_i \sqrt{\frac{\Delta}{m_{\rm mf}^{(K)}}}\sum_{k=1}^K \widetilde{z}_i^{(k)}\sqrt{\frac{m_{k}}{Km_{\rm mf}^{(K)} }}  \Big). \label{HK1_after}
\end{align}
We now proceed as previously using again the stability property of the Gaussian noise variables. Since $\widetilde Z_{i}^{(k)}$ are i.i.d ${\cal N}(0,1)$, 
then $\widetilde Z_i \defeq \sum_{k=1}^K \widetilde{Z}_i^{(k)} \sqrt{m_{k}/(Km_{\rm mf}^{(K)})}  \sim  {\cal N}(0,1)$ and 
are i.i.d. Let $\widetilde \bz = [\widetilde z_{i}]_{i=1}^n$. Using \eqref{intf} we find that 
$f_{K,1;0}$ can also be expressed as
\begin{align}
f_{K,1;0}&=-\frac{1}{n}\mathbb{E}_{\bS,\widetilde\bZ}\Big[\ln\int\big\{\prod_{i=1}^n dx_i P_0(x_i)\big\} \exp\Big\{-\frac{m_{\rm mf}^{(K)}}{\Delta}\sum_{i=1}^n\big(\frac{x_i^2}{2} - x_iS_i- x_i\widetilde Z_{i}\sqrt{\frac{\Delta}{m_{\rm mf}^{(K)}}}  \big)\Big\}\Big]
\nn
&
=-\mathbb{E}_{S,\widetilde Z}\Big[\ln\int dx P_0(x) \exp\Big\{-\frac{m_{\rm mf}^{(K)}}{\Delta}\big(\frac{x^2}{2} - xS- x\widetilde Z\sqrt{\frac{\Delta}{m_{\rm mf}^{(K)}}}  \big)\Big\}\Big]. \label{identifyMF}
\end{align}
Expression \eqref{identifyMF} is nothing else than the free energy \eqref{MFf} associated with the following scalar denoising model: $y = s  +  \widetilde z \, \Sigma(m_{\rm mf}^{(K)};\Delta)$, which leads to
\begin{align}
f_{K,1;0}=f_{\rm den}\big(\Sigma(m_{\rm mf}^{(K)};\Delta)\big). \label{fK1isfden}
\end{align}
%
%
\subsection{Free energy change along the adaptive interpolation path} \label{secFchange}
By construction of \eqref{intH} we have the following coherency property (see Figure \ref{fig:Kinterp}): The $(k,t = 1)$ and $(k + 1, t = 0)$ models are equivalent (the Hamiltonian \eqref{intH} is invariant under this change) and thus $f_{k,1;\epsilon} = f_{k+1,0;\epsilon}$ for any $k$. This implies that the $(k,t)$--interpolating free energy \eqref{intf} verifies 
\begin{align}\label{telescopf}
f_{1,0;\epsilon} = f_{K,1;\epsilon} +\sum_{k=1}^K (f_{k,0;\epsilon}-f_{k,1;\epsilon}) = f_{K,1;\epsilon}- \sum_{k=1}^K \int_{0}^1 dt \frac{d f_{k,t;\epsilon}}{dt}.
\end{align}
Let us evaluate $d f_{k,t;\epsilon}/dt$. Define the {\it overlap} $q_{\bx,\bs}  \defeq  n^{-1}\sum_{i=1}^n x_is_i$. Starting 
from \eqref{intf}, lenghty but simple algebra (see sec.~\ref{proofdH} for the details) shows that as long as $P_0$ has bounded first four moments,
\begin{align}
\frac{d f_{k,t;\epsilon}}{dt} =\frac{1}{4\Delta K}\mathbb{E}[\langle q_{\bX,\bS}^2- 2m_{k}q_{\bX,\bS}^{} \rangle_{k,t;\epsilon}]+ {\cal O}((nK)^{-1}). \label{dHdt}
\end{align}
This, with \eqref{telescopf} and \eqref{fK1isfden} yields
\begin{align}
f_{1,0;\epsilon} &= f_{K,1;\epsilon} -\frac{1}{4\Delta K}\sum_{k=1}^K \int_{0}^1 dt\, \mathbb{E}[\langle q_{\bX,\bS}^2 - 2m_{k}q_{\bX,\bS}^{} \rangle_{k,t;\epsilon}]+ {\cal O}(n^{-1})\nonumber \\
&=(f_{K,1;\epsilon} - f_{K,1;0}) +f_{\rm den}\big(\Sigma(m_{\rm mf}^{(K)};\Delta)\big) \nn
&\qquad- \frac{1}{4\Delta}\Big\{ -\frac{1}{K}\sum_{k=1}^K m_{k}^2 + \frac{1}{K}\sum_{k=1}^K \int_{0}^1 dt\, \mathbb{E}[\langle (q_{\bX,\bS}^{} - m_{k})^2 \rangle_{k,t;\epsilon}]\Big\}+ {\cal O}(n^{-1})\nonumber\\
&=(f_{K,1;\epsilon} - f_{K,1;0})+f_{\rm RS}(m_{\rm mf}^{(K)};\Delta) + \frac{V_K(\{m_{k}\}_{k=1}^K)}{4\Delta} \nn
&\qquad - \frac{1}{4\Delta K}\sum_{k=1}^K \int_{0}^1 dt\, \mathbb{E}[\langle (q_{\bX,\bS}^{} - m_{k})^2 \rangle_{k,t;\epsilon}]+ {\cal O}(n^{-1}), 
\label{fundam}
\end{align}
where in the last equality we used \eqref{eq:potentialfunction} and introduced the non-negative variance
\begin{align}
V_K(\{m_{k}\}_{k=1}^K) \defeq \frac{1}{K}\sum_{k=1}^K m_{k}^2- \Big(\frac{1}{K}\sum_{k=1}^K m_{k}\Big)^2.
\end{align}
The fundamental sum rule \eqref{fundam} can now be used to prove the replica symmetric formula.
\subsection{Upper bound} \label{secUpper}
From \eqref{fundam} we recover the upper bound usually obtained by the classical method of 
Guerra and Toninelli \cite{guerra2002thermodynamic} and applied in \cite{krzakala2016mutual} to symmetric rank-one matrix estimation
(but see also \cite{korada2009exact} which already fully proved the replica formula in the binary case). 
Choose $m_{k} = {\rm argmin}_{m\ge 0} f_{\rm RS}(m;\Delta)$ for all $k = 1,\ldots,K$. This implies $m_{\rm mf}^{(K)}  = {\rm argmin}_{m\ge 0} f_{\rm RS}(m;\Delta)$ as well as $V_K(\{m_{k}\}) = 0$. Thus since
 the integrand in \eqref{fundam} is non-negative we get the bound 
\begin{align}
 f_{1,0;\epsilon} \le (f_{K,1;\epsilon} - f_{K,1;0}) + \min_{m\ge 0}f_{\rm RS}(m;\Delta) + {\cal O}(n^{-1}). \label{24}
\end{align}
Now we apply this inequality to a sequence $\epsilon_n\to 0$ as $n\to +\infty$. 
From Lemma~\ref{thermolimit} and \eqref{f10isf} we obtain the upper bound:
\begin{proposition}[Upper bound] \label{UpperBound}
Fix $\Delta > 0$. For any $P_0$ with bounded first four moments, 
\begin{align}
\limsup_{n\to \infty}f_n\le \min_{m\ge 0}f_{\rm RS}(m;\Delta).	
\end{align}
\end{proposition}
\subsection{Lower bound} \label{seclower}
The converse bound is generally the one requiring extra technical tools, such as the use of {\it spatial coupling} \cite{6284230,Giurgiu_SCproof,XXT,barbier_allerton_RLE,barbier_ieee_replicaCS} or the Aizenman-Sims-Starr scheme, see \cite{aizenman2003extended,2017arXiv170200473M,2017arXiv170108010L,2016arXiv161103888L}. 
Thanks to the adaptive interpolation method the proof is quite straightforward.
As in all of the existing methods, we need a concentration lemma which takes the following form in the present context (see sec.~\ref{proofConc} for the proof).
\begin{lemma}[Overlap concentration]\label{concentration}
Let $P_0$ have bounded support. For any sequences $K_n\to +\infty$, $0< a_n <b_n<1$, and any choice of the trial parameters 
$m_k: \epsilon\mapsto m_k^{(n)}(\epsilon)$, $k=1,\cdots,K$ differentiable, bounded, non-decreasing with respect to $\epsilon$, we have
\begin{align}\label{concentration-overlap}
\int_{a_n}^{b_n} d\epsilon \, \frac{1}{K_n}\sum_{k=1}^{K_n}\int_0^1 dt\, \mathbb{E}\big[ \big\langle (q_{\bX,\bS}^{}
- \mathbb{E}[\langle q_{\bX,\bS}^{}\rangle_{k,t;\epsilon}])^2\big\rangle_{k,t;\epsilon}\big] 
\leq \frac{C}{a_n^2 n^{\alpha}},
\end{align}
for any $0<\alpha <1/4$ and some constant $C > 0$ independent of $n$, $K_n$, $b_n$ and the set of trial parameters
($C$ depends on the second moment and the support of $P_0$). 
\end{lemma}

\begin{remark}
In applications of this lemma the sequence $a_n$ tends to zero as slowly as we wish. In practice we will set later on $b_n=2a_n$ and take $a_n\to 0$ slowly enough so that $a_n^{-3} n^{-\alpha} \to 0$. In particular the r.h.s of \eqref{concentration-overlap} tends to zero.
\end{remark}

For sequences $K_n$, $0<a_n <b_n<1$, and $\{m_k=m_k^{(n)}(\epsilon)\}_{k=1}^K$ as in Lemma \ref{concentration}, \eqref{fundam} becomes
\begin{align}\label{fundam_afterConc}
\int_{a_n}^{b_n} d\epsilon\, f_{1,0;\epsilon}&=\int_{a_n}^{b_n} d\epsilon\, (f_{K_n,1;\epsilon} - f_{K_n,1;0}) + \int_{a_n}^{b_n} d\epsilon\,\biggl\{f_{\rm RS}(m_{\rm mf}^{(K_n)};\Delta) 
+ \frac{V_{K_n}(\{m_{k}\}_{k=1}^{K_n})}{4\Delta}\biggr\}
\nn
&\qquad- \frac{1}{4\Delta}\int_{a_n}^{b_n} d\epsilon\,\frac{1}{K_n}\sum_{k=1}^{K_n} \int_{0}^1 dt\, \big(\mathbb{E}[\langle q_{\bX,\bS}^{}\rangle_{k,t;\epsilon}] - m_{k}\big)^2  + \mathcal{O}(a_n^{-2}n^{-\alpha})
\end{align}
where ${\cal O}(a_n^{-2}n^{-\alpha})$ is uniform in the choice of $K_n$, $b_n$ and trial parameters.
%
At this point we need another important Lemma (see Appendix~\ref{appendix-proof-of-lemma3} for the proof) which is made possible by construction of the adaptive interpolation method.
\begin{lemma}[Weak $t$-dependence at fixed $k$]\label{tinvar}
Fix $K$, $\epsilon$ and $\{m_k\}_{k=1}^K$. 
For $P_0$ with bounded first four moments and any $k \in \{1,\ldots,K\}$ and $t \in [0,1]$,
\begin{align} 
\big|\mathbb{E}[\langle q_{\bX,\bS}^{}\rangle_{k,t;\epsilon}] - \mathbb{E}[\langle q_{\bX,\bS}^{}\rangle_{k,0;\epsilon}]\big| = {\cal O}\Big(\frac{n}{K}\Big). \label{noverK}
\end{align}
uniformly in $\epsilon$ and $\{m_k\}_{k=1}^K$.
This result also applies when fixed $K$ is replaced by $K_n$ and $\{m_k= m_k^{(n)}(\epsilon)\}$.
\end{lemma}

Using this lemma for a sequence $K_n = \Omega(n^b)$ with $b>1$ large enough, say $b = 2$, \eqref{fundam_afterConc} takes the following convenient form:
%
\begin{align}\label{fundam_afterConc_aftertinvar}
\int_{a_n}^{b_n} d\epsilon\,f_{1,0;\epsilon}= &\, \int_{a_n}^{b_n} d\epsilon\,(f_{K_n,1;\epsilon} - f_{K_n,1;0}) +\int_{a_n}^{b_n} d\epsilon\,\biggl\{f_{\rm RS}(m_{\rm mf}^{(K_n)};\Delta) 
+ \frac{V_{K_n}(\{m_{k}\}_{k=1}^{K_n})}{4\Delta}\biggr\}
\nonumber \\ &
 \qquad- \frac{1}{4\Delta}\int_{a_n}^{b_n} d\epsilon\,\frac{1}{K_n}\sum_{k=1}^{K_n}
\big(\mathbb{E}[\langle q_{\bX,\bS}^{}\rangle_{k,0;\epsilon}] - m_{k}\big)^2  + \mathcal{O}(a_n^{-2}n^{-\alpha}).
\end{align}
We now use the last crucial lemma which is a fundamental property of the adaptive interpolation.
\begin{lemma}[Choice for the trial parameters] \label{freedom}
For a given $n$ one can {\it freely select} differentiable and non-decreasing trial parameters $\{m_k = m_k^{(n)}(\epsilon)\}_{k=1}^K$ as
\begin{align}\label{choicex}
m_{k}= \mathbb{E}[\langle q_{\bX,\bS}^{}\rangle_{k,0;\epsilon}], \quad k=1, \cdots, K_n.
\end{align}
\end{lemma}
\begin{proof}
This is authorized by construction of the adaptive interpolation method. 
Indeed, the $(k = 1,t = 0)$--interpolating model (see the 
Hamiltonian ${\cal H}_{1,0;\epsilon}(\bx;\boldsymbol{\theta})$ in \eqref{intH}) is 
independent of $\{m_{k}\}_{k=1}^{K_n}$. Thus we can freely 
set $m_1=m_1^{(n)}(\epsilon) = \mathbb{E}[\langle q_{\bX,\bS}^{}\rangle_{1,0;\epsilon}]$. 
Once $m_1$ is fixed to this value $m_1^{(n)}(\epsilon)$, we go to the next step and 
set $m_2=m_2^{(n)}(\epsilon) = \mathbb{E}[\langle q_{\bX,\bS}^{}\rangle_{2,0;\epsilon}]$, which again is 
possible due to the fact that the Hamiltonian ${\cal H}_{2,0;\epsilon}(\bx;\boldsymbol{\theta})$ and the 
Gibbs average $\langle -\rangle_{2,0;\epsilon}$ as well depend {\it only} on $m_1$ which has {\it already been fixed}. And so forth: As seen from Fig.~\ref{fig:Kinterp}, 
the Gibbs average $\langle -\rangle_{k,0;\epsilon}$ depends 
 only on $\{m_{k'}\}_{k'=1}^{k-1}$ which were already fixed in the previous steps so that the choice \eqref{choicex} is valid. 
 Note that $\mathbb{E}[\langle q_{\bX,\bS}^{}\rangle_{k,0;\epsilon}] \ge 0$ which is important as the $m_{k}^{(n)}(\epsilon)$'s play the role of signal-to-noise ratios, 
 and thus must be positive. Moreover the maps $\epsilon\mapsto m_k^{(n)}(\epsilon)$ are differentiable  and non-decreasing. They are obviously differentiable since we work with $n$ finite. 
 To see that they are non-decreasing we look at their derivative. 
 It is easy to see from the construction of the Gibbs bracket 
 that $\mathbb{E}[\langle q_{\bX,\bS}^{}\rangle_{k,0;\epsilon}]$ is a function $\epsilon\mapsto G_k^{(n)}(\epsilon + \frac{1}{K}\sum_{l=1}^{k-1} m_l)$ so that 
 $$
 \frac{d}{d\epsilon}m_k^{(n)}(\epsilon) = G_k^{(n)\prime}\biggl(\epsilon + \frac{1}{K}\sum_{l=1}^{k-1} m_l^{(n)}(\epsilon)\biggr) 
 \biggl(1 + \frac{1}{K}\sum_{l=1}^{k-1} \frac{d}{d\epsilon}m_l^{(n)}(\epsilon)\biggr)\,.
 $$
 Now, $G_k^{(n)\prime}$ is an expected variance, and is therefore positive, as can be directly shown from a direct 
 calculation (see equations \eqref{average-second-derivative} and \eqref{alternative} in  section \ref{proofConc}). This implies by induction that 
 $\frac{d}{d\epsilon}m_k^{(n)}(\epsilon) \geq 0$. 
\end{proof}

With this particular choice of trial parameters $\{m_k=m_k^{(n)}(\epsilon)\}_{k=1}^{K_n}$ the sum over $k = 1,\ldots, K_n$ in \eqref{fundam_afterConc_aftertinvar} is set to zero: 
The interpolation path has been {\it adapted} (thus the name of the method). Since $V_{K_n}$ is non-negative, 
\eqref{fundam_afterConc_aftertinvar} directly implies the following lower bound:
\begin{align}\label{fundam_afterConc_aftertinvar_2}
\int_{a_n}^{b_n} d\epsilon\,f_{1,0;\epsilon}&=\int_{a_n}^{b_n} d\epsilon\,\biggl\{(f_{K_n,1;\epsilon} - f_{K_n,1;0}) + f_{\rm RS}\Big(K_n^{-1}\sum_{k=1}^{K_n} m_k^{(n)} ;\Delta\Big) 
+ \frac{V_{K_n}(\{m_k^{(n)}\}_{k=1}^{K_n})}{4\Delta}\biggr\} +  \mathcal{O}(a_n^{-2}n^{-\alpha})\nn
&\ge \int_{a_n}^{b_n} d\epsilon\,(f_{K_n,1;\epsilon} - f_{K_n,1;0})+ (b_n - a_n)\min_{m\ge 0}f_{\rm RS}(m ;\Delta) 
+  \mathcal{O}(a_n^{-2}n^{-\alpha}).
\end{align}
Finally, setting $b_n = 2 a_n$ and taking $a_n \to 0$ slowly enough as $n\to +\infty$ so that $a_n^{-3} n^{-\alpha}\to 0$, using Lemma~\ref{thermolimit} and \eqref{f10isf} and the mean value theorem, we deduce
\begin{proposition}[Lower bound]\label{LowerBound}
Fix $\Delta > 0$. For any $P_0$ with bounded support,
\begin{align}\label{llbb}
\liminf_{n\to \infty}f_n\ge \min_{m\ge 0}f_{\rm RS}(m;\Delta).
\end{align}
\end{proposition}
\begin{remark}[The overlap must concentrate]
Note that it is not obvious that one can find $\{m_{k}^{(n)}(\epsilon)\}_{k=1}^{K_n}$ which directly cancel the integrals in the fundamental identity \eqref{fundam} without 
using the overlap concentration of Lemma~\ref{concentration}.
Overlap concentration is a fundamental requirement of the above proof. This agrees with the statistical physics assumption that a 
necessary condition for the validity of the replica {\it symmetric} method is precisely the overlap concentration \cite{mezard1990spin}.
\end{remark}

In Appendix \ref{appendix-alternative-route} we present an alternative useful, albeit not completely rigorous, argument to obtain the lower bound.
\subsection{Proof of the fundamental sum rule} \label{proofdH}
In this paragraph we derive the formula \eqref{dHdt}.
We will need a simple but fundamental identity\footnote{This identity has been abusively called ``Nishimori identity'' in the statistical physics literature. One should however note that
it is a simple consequence of Bayes formula (see e.g appendix B of \cite{barbier_ieee_replicaCS}). The 
``true'' Nishimori identity \cite{nishimori01} concerns models with one extra feature, namely a 
gauge symmetry which allows to eliminate the input signal, and the expectation over $\bS$ in \eqref{nishibasic} 
can therefore be dropped (see e.g. \cite{korada2009exact}).} which is a straightforward consequence of the Bayes law. Let $\bX$, $\bX'$ be two i.i.d ``replicas'' drawn 
according to the product distribution $P_{k,t;\epsilon}(\bx|\boldsymbol{\theta}) P_{k,t;\epsilon}(\bx'|\boldsymbol{\theta})$. Recall 
the notation $\boldsymbol{\theta} \defeq \{\bs, \{\bz^{(k)}, \widetilde{\bz}^{(k)}\}_{k=1}^K, \widehat \bz\}$ for the quenched variables and $\mathbb{E} = 
\mathbb{E}_{\boldsymbol{\Theta}}$ for the expectation with respect to these.
Then for any function $g$ which {\it does not} depend on the Gaussian noise random variables,
\begin{align}\label{nishibasic}
\EE[\langle g(\bX, \bS) \rangle_{k,t;\epsilon}] = \EE[\langle g(\bX, \bX')\rangle_{k,t;\epsilon}].
\end{align}
We give a proof of this identity in Appendix \ref{appendix-Nishimori-Bayes} for completeness. 

Let us now compute $d f_{k,t;\epsilon}/dt$. Starting from \eqref{intH},
\eqref{perturb}, \eqref{intf} one obtains 
\begin{align}
\frac{d f_{k,t;\epsilon}}{dt}&=\frac{1}{n}\mathbb{E}\Big[\Big\langle\frac{d {\cal H}_{k,t;\epsilon}(\bX;\boldsymbol{\Theta})}{dt}\Big\rangle_{k,t;\epsilon}\Big]\nn
&= \frac{1}{n}\mathbb{E}\Big[\Big\langle \frac{d }{dt}h_{\rm mf}\Big(\bx, \bs,\widetilde{\bz}^{(k)},\frac{K\Delta}{t\,m_{k}} \Big)+ \frac{d }{dt}h\Big(\bx, \bs,{\bz}^{(k)},\frac{K\Delta}{1-t}\Big) \Big\rangle_{k,t;\epsilon}\Big]\nonumber\\
&= \frac{1}{nK\Delta}\mathbb{E}\Big[\Big\langle m_k\sum_{i=1}^n\Big(\frac{X_i^2}{2} - X_iS_i-\frac{X_i\widetilde Z_i^{(k)}}{2}\sqrt{\frac{K\Delta}{t\,m_k}} \Big)\nn
&\qquad- \sum_{i\le j=1}^n\Big(\frac{X_i^2X_j^2}{2n} - \frac{X_iX_jS_iS_j}{n}-\frac{X_iX_j Z_{ij}^{(k)}}{2\sqrt{n}}\sqrt{\frac{K\Delta}{1-t}} \Big) \Big\rangle_{k,t;\epsilon}\Big].\nonumber
\end{align}
Now we integrate by part the Gaussian noise using the elementary formula $\EE_{Z}[Zf(Z)] = \EE_Z[f'(Z)]$ where $f'$ is the derivative of $f$. This leads to
\begin{align}
\frac{d f_{k,t;\epsilon}}{dt}&= \frac{1}{nK\Delta}\mathbb{E}\Big[\Big\langle m_k\sum_{i=1}^n\Big(\frac{X_iX_i'}{2}- X_iS_i\Big)- \sum_{i\le j=1}^n\Big(\frac{X_iX_j X_i'X_j'}{2n}- \frac{X_iX_jS_iS_j}{n} \Big) \Big\rangle_{k,t;\epsilon}\Big],
\end{align}
where $\bX$, $\bX'$ are the two i.i.d replicas drawn according to \eqref{post}. An application of identity \eqref{nishibasic} then leads to
\begin{align}
\frac{d f_{k,t;\epsilon}}{dt}&= \frac{1}{2K\Delta}\mathbb{E}\Big[\Big\langle \frac{1}{n^2} \sum_{i\le j=1}^nX_iX_jS_iS_j -\frac{m_k}{n}\sum_{i=1}^n X_iS_i  \Big\rangle_{k,t;\epsilon}\Big] \nonumber\\
&= \frac{1}{2K\Delta}\mathbb{E}\Big[\Big\langle  \frac{1}{2n^2} \sum_{i,j=1}^n X_iX_jS_iS_j +\frac{1}{2n^2} \sum_{i=1}^n X_i^2S_i^2 -\frac{m_k}{n}\sum_{i=1}^n X_iS_i \Big\rangle_{k,t;\epsilon}\Big]. \label{last}
\end{align}
The Cauchy-Schwarz inequality and \eqref{nishibasic} imply that 
$\mathbb{E}[\langle n^{-2}\sum_{i=1}^n X_i^2S_i^2\rangle_{k,t;\epsilon}]  = {\cal O}(n^{-1})$ as long as $P_0$ has bounded fourth moment.
Indeed, by Cauchy-Schwarz
\begin{align}
 \mathbb{E}\Big[\Big\langle n^{-1} \sum_{i=1}^n X_i^2S_i^2\Big\rangle_{k,t;\epsilon}\Big] \leq 
 \Big(\mathbb{E}\Big[\Big\langle n^{-1}\sum_{i=1}^n X_i^4\Big\rangle_{k,t;\epsilon}\Big]\Big)^{1/2}\Big(\mathbb{E}\Big[n^{-1}\sum_{i=1}^n S_i^4\Big]\Big)^{1/2} 
\end{align}
and by \eqref{nishibasic} we have $\mathbb{E}[\langle X_i^4\rangle_{k,t;\epsilon}] = \mathbb{E}[ S_i^4]$ for $i=1,\dots, n$, thus we get
\begin{align}
\mathbb{E}\Big[\Big\langle n^{-1} \sum_{i=1}^n X_i^2S_i^2\Big\rangle_{k,t;\epsilon}\Big] \leq \mathbb{E}[S^4].
\end{align}
Finally, expressing the two other terms in \eqref{last} uisng the overlap $q_{\bx,\bs} = n^{-1}\sum_{i=1}^n x_is_i$ we find \eqref{dHdt}.
\section{Application to rank-one symmetric tensor estimation}\label{tensor}
The present method can be extended to cover rank-one symmetric tensor estimation, which amounts to treat the $p$-spin model on the Nishimori line. For binary spins the Guerra-Toninelli bound was proven in \cite{korada2009exact} for {\it any} value of $p$, the replica symmetric formula was proved in the whole phase diagram for $p=2$, and also in a restricted region away from the first order 
phase transition for $p \geq  3$. A complete proof for $p=2$ and general spins (that can thus be real) was achieved using the spatial coupling technique in \cite{XXT} and in \cite{2017arXiv170108010L} by a rigorous version of the cavity method. The case $p\geq 2$ and general spins has been treated using again the cavity method in \cite{2016arXiv161103888L}.

\subsection{Symmetric rank-one tensor estimation: Setting and main result}
\label{sym-est-subsec}
%
The symmetric tensor problem is very close to the matrix case presented in full details in sec.~\ref{stochInt} so we only sketch the main steps. 
The observed symmetric tensor $\bw \in \mathbb{R}^{n_1\times n_2\times\ldots \times n_p}$ is obtained through the following estimation model:
\begin{align}
w_{i_1i_2\ldots i_p} = \sqrt{\frac{(p-1)!}{n^{p-1}}} s_{i_1}s_{i_2}\ldots s_{i_p} + z_{i_1i_2\ldots i_p}\sqrt{\Delta}  \quad \text{for} \quad  1\le i_1\leq i_2\le \ldots \le i_p\le n, \label{symtensormodel}
\end{align}
where $\bs \in \mathbb{R}^n$ with i.i.d components distributed according to a known prior 
$P_0$, $\bZ \in \mathbb{R}^{n_1\times n_2\times\ldots \times n_p}$ is a symmetric Gaussian noise 
tensor with i.i.d (up to the symmetry constraint) ${\cal N}(0,1)$ entries. 
We note that, like in the case of symmetric matrix estimation of sec.~\ref{sec:partI}, the channel universality property (see remark \ref{rmk:univers}) is valid in 
the present setting. This means that by 
covering the case of additive white Gaussian noise \eqref{symtensormodel}, we actually treat a wide range of 
(component-wise) inference channels $$P_{\rm out}(w_{i_1i_2\ldots i_p}|\sqrt{(p-1)!n^{1-p}} s_{i_1}s_{i_2}\ldots s_{i_p}).$$ 
We refer to \cite{lesieur2015mmse,2017arXiv170100858L,2017arXiv170108010L} for more details on this point.
The free energy of the model is 
\begin{align}
f_n \defeq &-\frac{1}{n}\EE_{\bS,\bZ}\Big[ \ln\int \big\{\prod_{i=1}^n dx_i P_0(x_i)\big\}e^{-\mathcal{H}(\bx; \bS, \bZ)}\Big]\label{foriginal_symtensor}
\end{align}
where the Hamiltonian $\mathcal{H}(\bx; \bs, \bz)$ is 
\begin{align}
\frac{1}{\Delta}\sum_{i_1\le i_2\le \ldots \le i_p}\Big(\frac{(p-1)!}{2n^{p-1}} x_{i_1}^2\ldots x_{i_p}^2 - \frac{(p-1)!}{n^{p-1}} x_{i_1} s_{i_1}\ldots x_{i_p} s_{i_p} - \sqrt{\frac{\Delta(p-1)!}{n^{p-1}}} z_{i_1i_2\ldots i_p} x_{i_1}\ldots x_{i_p} \Big).
\end{align}
For a $P_0$ with bounded first four moments the free energy is related to the mutual information $I(\bS;\bW)$ through 
\begin{align}
\frac{I(\bS;\bW)}{n}=f_n+\frac{\EE[S^2]^p}{2p\Delta} + {\cal O}(n^{-1}).
\end{align}

We define the replica symmetric potential for symmetric tensor estimation as
\begin{align}
f_{\rm RS}(m;\Delta) \defeq \frac{(p-1)m^p}{2p\Delta} + f_{{\rm den}}\big(\Sigma(m;\Delta)\big) \label{RSpot_symtensor}
\end{align}
where $\Sigma(m;\Delta)^2 \defeq \Delta/m^{p-1}$ and $f_{{\rm den}}(\Sigma)$ is given by \eqref{MFf}. Next we prove the RS formula.
\begin{thm}[RS formula for symmetric rank-one tensor estimation] \label{RS_symtensor}
Fix $\Delta > 0$. For any $P_0$ with bounded support, the asymptotic free energy of the symmetric tensor estimation model \eqref{symtensormodel} verifies
\begin{align}\label{RS-sym-tensor-case}
\lim_{n\to \infty}f_n= \min_{m\ge 0}f_{\rm RS}(m;\Delta).	
\end{align}
\end{thm}
Again, we note that the bounded support property of $P_0$ is only needed for concentration proofs and does not impose any upper limit on the size of the support.
We believe this can be removed by a limiting process as long as $P_0$ has bounded first four moments. 
\subsection{Sketch of proof of the replica symmetric formula}
\label{sec:partII_symtensor}
We prove Theorem \ref{RS_symtensor}. Since this proof is similar to the one of Theorem \ref{thm1} for the matrix case, we only give the main ideas. 
The starting point is the introduction of a
(perturbed) $(k,t)$--interpolating Hamiltonian:
\begin{align}\label{intH_tensor}
{\cal H}_{k,t;\epsilon}(\bx;\boldsymbol{\theta})&\defeq \sum_{k'=k+1}^K h\Big(\bx, \bs,\bz^{(k')},K\Delta\Big)+ \sum_{k'=1}^{k-1}  h_{\rm mf}\Big(\bx, \bs,\widetilde{\bz}^{(k')},K\,\Sigma(m_k;\Delta)^2\Big) \nonumber\\
&\quad+ h\Big(\bx, \bs,\bz^{(k)},\frac{K\Delta}{1-t}\Big)+ h_{\rm mf}\Big(\bx, \bs,\widetilde{\bz}^{(k)},\frac{K\,\Sigma(m_k;\Delta)^2}{t} \Big) + \epsilon \sum_{i=1}^{n}\Big(\frac{x_i^2}{2} - x_i s_{i} -\frac{x_{i}\widehat z_{i}}{\sqrt{\epsilon}}\Big),
\end{align}
where the trial parameters $\{m_{k}\}_{k=1}^{K}$ are to be fixed later and
\begin{align}
h(\bx, \bs,\bz,\sigma^2) &\defeq\frac{1}{\sigma^2}\sum_{i_1\le i_2\le \ldots \le i_p} \Big(\frac{(p - 1)!}{2n^{p-1}} x_{i_1}^2\ldots x_{i_p}^2 - \frac{(p - 1)!}{n^{p-1}} x_{i_1} s_{i_1}\ldots x_{i_p} s_{i_p} \nn
&\qquad\qquad- \sigma\sqrt{\frac{(p - 1)!}{n^{p-1}}}z_{i_1i_2\ldots i_p} x_{i_1}\ldots x_{i_p} \Big), \\
h_{\rm mf}(\bx, \bs,\widetilde{\bz},\sigma^2) &\defeq \frac{1}{\sigma^2}\sum_{i=1}^{n}\Big(\frac{x_{i}^2}{2} - x_{i} s_{i}- \sigma \widetilde z_{i} x_{i} \Big).
\label{hmf_tensor}
\end{align}
The associated $(k,t)$--interpolating model, Gibbs expectation and $(k,t)$--interpolating free energy are defined respectively by \eqref{post}, \eqref{Gibbs} and \eqref{intf}. 
Using the stability property of the Gaussian noise variables, one can check that the intial and final $(k,t)$--interpolating models are such that 
\begin{align}
f_{1,0;0}&=f_n, \label{114_symtensor}	\\
f_{K,1;0}&=f_{\rm den}\big(\Sigma_{\rm mf}(\{m_k\}_{k=1}^K;\Delta)\big), \label{115_symtensor}
\end{align}
where
\begin{align}
\Sigma_{\rm mf}(\{m_k\}_{k=1}^K;\Delta)^{-2} \defeq \frac{1}{K}\sum_{k=1}^K \Sigma(m_k;\Delta)^{-2} = \frac{1}{\Delta K}\sum_{k=1}^K m_k^{p-1}	.
\end{align}
By a trivial generalization of the calculations of sec.~\ref{proofdH}, we obtain the variation of the $(k,t)$--interpolating free energy:
\begin{align}
\frac{d f_{k,t;\epsilon}}{dt} = \frac{1}{2p\Delta K}\EE[\langle q_{\bX,\bS}^{p} - p\, m_k^{p-1} q_{\bX,\bS}^{} \rangle_{k,t;\epsilon}] + {\cal O}((nK)^{-1}), \label{61_symtensor}
\end{align}
where the overlap is again $q_{\bx,\bs}^{} \defeq n^{-1}\sum_{i}^{n} x_{i}s_{i}$. This result holds as long as $P_0$ has finite first four moments.
Proceeding similarly to section \ref{secFchange} we get the sum rule 
\begin{align}
 f_{1,0;\epsilon} & = (f_{K,1;\epsilon} - f_{K,1;0}) + f_{\rm den}(\Sigma_{\rm mf}(\{m_k\}_{k=1}^K;\Delta)) + \frac{p-1}{2p\Delta K}\sum_{k=1}^K m_k^p
 \nonumber\\
&\qquad- \frac{1}{2p\Delta K}\sum_{k=1}^K \int_{0}^1 dt\Big(\EE[\langle q_{\bX,\bS}^{p}\rangle_{k,t;\epsilon}] - p\, m_k^{p-1} 
\EE[\langle q_{\bX,\bS}^{}\rangle_{k,t;\epsilon}] + (p-1)m_k^p\Big) + \mathcal{O}(n^{-1})
\nonumber \\
&
=
(f_{K,1;\epsilon} - f_{K,1;0}) + f_{\rm den}(\Sigma_{\rm mf}(\{m_k\}_{k=1}^K;\Delta)) 
+ \frac{p-1}{2p\Delta^{-\frac{1}{p-1}}}\Big(\Sigma_{\rm mf}(\{m_k\}_{k=1}^K;\Delta))^{-2}\Big)^{\frac{p}{p-1}}
\nonumber \\ & \qquad
+ \frac{p-1}{2p\Delta}V_{K,p}(\{m_k\}_{k=1}^K)
 \nonumber\\
&\qquad- \frac{1}{2p\Delta K}\sum_{k=1}^K \int_{0}^1 dt\Big(\EE[\langle q_{\bX,\bS}^{p}\rangle_{k,t;\epsilon}] - p\, m_k^{p-1} 
\EE[\langle q_{\bX,\bS}^{}\rangle_{k,t;\epsilon}] + (p-1)m_k^p\Big) + \mathcal{O}(n^{-1})
\label{fund_symtensor}
%
\end{align}
where 
\begin{align}
 V_{K,p}(\{m_k\}_{k=1}^K) \defeq \frac{1}{K}\sum_{k=1}^K m_k^p - \Big(\frac{1}{K}\sum_{k=1}^K m_k^{p-1}\Big)^{\frac{p}{p-1}}.
\end{align}
Note that $V_{K,p}$ is non-negative by Jensen's inequality applied to the convex function $x\in \mathbb{R}_+ \mapsto x^{\frac{p}{p-1}}\in \mathbb{R}_+$ (here the trial parameters are all non-negative). 
For $p=2$ the identity \eqref{fund_symtensor} reduces to \eqref{fundam}.

The proof of Theorem \ref{RS_symtensor} proceeds from this fundamental sum rule in much the same way as in the case $p=2$ of sections \ref{secUpper} and \ref{seclower}. 
Here we only give a brief summary of the arguments insisting only on the essential differences. We start with the upper bound.

For the case of {\it even} $p$ it is straightforward to derive an upper bound. One chooses all trial parameters as $m_k = m_*= {\rm argmin}_{m\geq 0} f_{\rm RS}(m; \Delta)$, $k=1,\cdots, K$, which yields 
$V_{K,p}(\{m_k\}_{k=1}^K) =0$. From there one can use the classic argument of Guerra-Toninelli: By convexity of $x\in \mathbb{R} \mapsto x^p\in \mathbb{R}_+$ for even $p$ we see that 
$q_{\bX,\bS}^{p} - p\, m_*^{p-1} q_{\bX,\bS} + (p-1)m_*^p \geq 0$, which implies the upper bound analogous to \eqref{24}. Since Lemma \ref{thermolimit} holds verbatim here, by taking a sequence 
$\epsilon_n\to 0_+$, $n\to +\infty$, we deduce as before the upper bound $\limsup_{n\to +\infty} f_n \leq \min_{m\geq 0}f_{\rm RS}(m;\Delta)$ (when $P_0$ has finite first four moments). 

For the case of {\it odd} $p$ we cannot immediately apply the convexity argument. We first need to apply a concentration result. As it will become clear from its proof, Lemma \ref{concentration} is 
generic and the same statement applies to the
present tensor setting.\footnote{Here we use Lemma \ref{concentration} but a weaker form of concentration is enough for this argument, 
namely it suffices to control the following type of ``thermal'' fluctuation 
$\mathbb{E}[\langle q_{\bX, \bS}^2\rangle_{k,t,\epsilon} - \langle q_{\bX, \bS}^{}\rangle_{k,t,\epsilon}^2]$. Moreover it is not necessary to allow for an $\epsilon$-dependence in $m_k$'s.} 
Therefore for any sequence $K_n\to +\infty$, $0<a_n <b_n<1$, $b_n\to 0$, and 
trial parameters $\{m_k= m_k^{(n)}(\epsilon)\}_{k=1}^K$ which are non-decreasing functions of $\epsilon$
we have
\begin{align}
 & \int_{a_n}^{b_n}d\epsilon\, f_{1,0;\epsilon}  = \int_{a_n}^{b_n} d\epsilon\, (f_{K_n,1;\epsilon} - f_{K_n,1;0}) 
 + \int_{a_n}^{b_n} d\epsilon\,\biggl\{f_{\rm den}(\Sigma_{\rm mf}(\{m_k\}_{k=1}^{K_n};\Delta)) 
\nonumber \\ & 
 + \frac{p-1}{2p\Delta^{-\frac{1}{p-1}}}\Big(\Sigma_{\rm mf}(\{m_k\}_{k=1}^{K_n};\Delta))^{-2}\Big)^{\frac{p}{p-1}}
+ \frac{p-1}{2p\Delta}V_{K_n,p}(\{m_k\}_{k=1}^{K_n})\biggr\}
 \nonumber\\
&- \frac{1}{2p\Delta}\int_{a_n}^{b_n} d\epsilon\,\frac{1}{K_n}\sum_{k=1}^{K_n} \int_{0}^1 dt\Big(\EE[\langle q_{\bX,\bS}^{}\rangle_{k,t;\epsilon}]^{p} - p\, m_k^{p-1} 
\EE[\langle q_{\bX,\bS}^{}\rangle_{k,t;\epsilon}] + (p-1)m_k^p\Big) + \mathcal{O}(a_n^{-2}n^{-\alpha}),
\label{fund_symtensor2}
%
\end{align}
for $0<\alpha <1/4$ and ${\cal O}(a_n^{-2}n^{-\alpha})$ uniform in $K_n$, $b_n$, and trial parameters. Furthermore by the Nishimori 
identity \eqref{nishibasic} we see that $\EE[\langle q_{\bX,\bS}^{}\rangle_{k,t;\epsilon}] \geq 0$,
so convexity of $x\in \mathbb{R}_+ \mapsto x^p\in \mathbb{R}_+$ shows that the term under the integral is positive, 
which allows to deduce the upper bound as above (of course this argument works for any $p$ even or odd). 

let us finally briefly discuss the lower bound. Lemma \ref{tinvar} and its proof hold for the tensor setting as well which means that in \eqref{fund_symtensor2} we can replace 
$\EE[\langle q_{\bX,\bS}^{}\rangle_{k,t;\epsilon}]$ by $\EE[\langle q_{\bX,\bS}^{}\rangle_{k,0;\epsilon}]$. This then allows to choose (adapt) the sequence of trial parameters as
$\{m_k=m_k^{(n)}(\epsilon)\}$ where 
$m_k = \EE[\langle q^{}_{\bX,\bS}\rangle_{k,0;\epsilon}]$, $k=1, \cdots, K_n$ just as in Lemma \ref{freedom}. Thus we obtain 
\begin{align}
  \int_{a_n}^{b_n}d\epsilon\,f_{1,0;\epsilon} &= \int_{a_n}^{b_n}d\epsilon\,(f_{K_n,1;\epsilon} - f_{K_n,1;0}) 
 \nonumber \\ &
 \qquad+ 
 \int_{a_n}^{b_n}d\epsilon\,\biggl\{f_{\rm den}(\Sigma_{\rm mf}(\{m_k\}_{k=1}^{K_n};\Delta)) 
+ \frac{p-1}{2p\Delta^{-\frac{1}{p-1}}}\Big(\Sigma_{\rm mf}(\{m_k\}_{k=1}^{K_n};\Delta))^{-2}\Big)^{\frac{p}{p-1}}
 \nonumber \\ & 
 \qquad+ \frac{p-1}{2p\Delta} V_{K_n,p}(\{m_k^{(n)}\}_{k=1}^{K_n})\biggr\} + \mathcal{O}(a_n^{-2}n^{-\alpha})
 \nonumber \\ &
 \geq 
 \int_{a_n}^{b_n}d\epsilon\,(f_{K_n,1;\epsilon} - f_{K_n,1;0}) + (b_n-a_n)\min_{\Sigma\geq 0}\Big\{f_{\rm den}(\Sigma) 
+ \frac{p-1}{2p\Delta^{-\frac{1}{p-1}}}\big(\Sigma^{-2}\big)^{\frac{p}{p-1}}\Big\} + \mathcal{O}(a_n^{-2}n^{-\alpha})
\nonumber \\ &
= 
\int_{a_n}^{b_n}d\epsilon\,(f_{K_n,1;\epsilon} - f_{K_n,1;0}) + (b_n-a_n)\min_{m\geq 0} f_{\rm RS}(m;\Delta) + \mathcal{O}(a_n^{-2}n^{-\alpha})
\end{align}
where we used the non-negativity of $V_{K, p}$ to get the inequality and the change of variable $\Sigma^2 = \Delta/m^{p-1}$ to get the last line. The usual limiting argument, taking
$b_n = 2 a_n$ with $a_n\to 0$ slowly enough so that $a_n^{-3}n^{-\alpha} \to 0$ when $n\to +\infty$,
implies  $\liminf_{n\to +\infty} f_n \geq \min_{m\geq 0} f_{\rm RS}(m;\Delta)$. We have proven this lower bound under the assumption of boundedness of the support of $P_0$ (used in the proof 
of overlap concentration)

Combining the upper and lower bounds yields Theorem \ref{RS_symtensor}. 
As a final note we remark that the alternative route to the lower bound proposed in Appendix \ref{appendix-alternative-route} for the matrix case also holds essentially unchanged here. 

\section{Application to Gaussian random linear estimation}\label{RLE_section}
\subsection{Gaussian random linear estimation: Setting and result}
In Gaussian random linear estimation (RLE) one is interested in reconstructing a signal $\bs = [s_i]_{i=1}^n \in  \mathbb{R}^n$ from few noisy
measurements $\by =  [y_\mu]_{\mu=1}^m \in  \mathbb{R}^m$ obtained from the projection of $\bs$ by a random Gaussian \emph{measurement matrix} $\bm{\phi} = [\phi_{\mu i}]_{\mu,i=1}^{m, n} \in  \mathbb{R}^{m \times n}$ with i.i.d entries $\phi_{\mu i} \sim \mathcal{N}(0,1/n)$. The \emph{measurement rate} is $\alpha \defeq  m/n$. We consider i.i.d additive white Gaussian noise of known variance $\Delta$. Let
the standardized noise components be $Z_\mu  \sim  \mathcal{N}(0,1)$, $\mu = 1,\ldots, m$. Then the measurement model is
\be \label{eq:CSmodel}
\by = \bm{\phi}{\bs} + \bz \sqrt{\Delta}, \qquad \text{or}\qquad y_{\mu} = \sum_{i=1}^n\phi_{\mu i} s_i + z_\mu \sqrt{\Delta} \quad \text{for}\quad 1\leq \mu\leq m.
\ee
The signal has i.i.d components distributed according to a discrete prior $P_0(s_i)  =  \sum_{b=1}^B p_b \delta(s_i  -  a_b)$ with a finite number $B$ of terms and $\max_{b} \vert a_{b}\vert \leq s_{\rm max}$. Note that the more general case where the signal has i.i.d vectorial components, as considered in \cite{barbier_allerton_RLE,barbier_ieee_replicaCS}, can be tackled with our proof technique exactly in the same way but we consider the scalar case for the sake of notational simplicity. 

The free energy of the RLE model \eqref{eq:CSmodel} (which is also equal to the mutual information per component $I(\bS;\bY)/n$ between the noisy observation and the signal) is defined as
\begin{align} \label{eq:true_mutual_info}
f_n &\defeq -\frac{1}{n} \EE_{\bS,\bZ,\bm{\Phi}}\Big[\ln\int \big\{\prod_{i=1}^ndx_i P_0(x_i)\big\}
\exp\Big\{-\frac{1}{\Delta}\sum_{\mu=1}^m\Big(\frac{1}{2}[\bm{\Phi}(\bx - \bS)]_\mu^2 -[\bm{\Phi}(\bx - \bS)]_\mu Z_\mu\sqrt{\Delta}\Big)\Big\}\Big],
\end{align}
where  $[\bm{\phi}(\bx - \bs)]_\mu \defeq  \sum_{i=1}^n\bm{\phi}_{\mu i} (x_i - s_i)$. Let 
\begin{align}
\Sigma(E;\Delta)^{-2} &\defeq \frac{\alpha}{\Delta+ E},\label{eq:defSigma2}\\
\psi(E;\Delta) &\defeq \frac{\alpha}{2}\Big(\ln\Big(1+ \frac{E}{\Delta}\Big)- \frac{E}{\Delta+E}\Big). \label{psi}
\end{align}
Define the following RS potential:
\begin{align}
f_{\rm RS}(E;\Delta) \defeq \psi(E;\Delta) + i_{\rm den}\big(\Sigma(E;\Delta)\big), 
\label{eq:rs_mutual_info}
\end{align}
where $i_{\rm den}(\Sigma)=I(S;S+\widetilde Z\,\Sigma)$ is the mutual information of a scalar Gaussian denoising model $y = s + \widetilde z\,\Sigma$ with $S \sim P_0$, $\widetilde Z \sim {\cal N}(0,1)$, and $\Sigma^{-2}$ an effective signal to noise ratio:
\begin{align}
i_{\rm den}(\Sigma) \defeq -\mathbb{E}_{S,\widetilde Z}\Big[\ln\int dx P_0(x) e^{-\frac{1}{\Sigma^2}\big(\frac{(x - S)^2}{2} - (x - S)\widetilde Z\,\Sigma \big)}\Big].
\end{align}
We will prove the RS formula (already proven in \cite{barbier_allerton_RLE,barbier_ieee_replicaCS,private,ReevesP16}):
\begin{thm}[RS formula for Gaussian RLE]\label{rsformulaRLE}
Fix $\Delta > 0$. For any discrete $P_0$, the asymptotic free energy of the RLE model \eqref{eq:CSmodel} verifies
\begin{align}
\lim_{n\to \infty} f_n = \min_{E\ge0} f_{\rm RS}(E; \Delta).
\end{align}
\end{thm}
\begin{proof}
The result follows from Propositions \ref{UpperBound_RLE} and \ref{LowerBound_RLE} in sec. \ref{rsform}.
\end{proof}
\subsection{Proof of the RS formula}\label{rsform}
Let ${\bz}^{(k)} = [z_{\mu}^{(k)}]_{\mu=1}^m, \widetilde{\bz}^{(k)} = [\widetilde{z_{i}}^{(k)}]_{i=1}^n$ 
and $\widehat{\bz} = [\widehat{z_{i}}]_{i=1}^n$ all with i.i.d ${\cal N}(0,1)$ entries for $k = 1,\ldots,K$. 
Define $\Sigma_k \defeq   \Sigma(E_k;\Delta)$ where the trial parameters $\{E_k\}_{k=1}^K$ are fixed later on. 
The (perturbed) $(k,t)$--interpolating Hamiltonian for the present problem is
\begin{align}\label{intH_RLE}
{\cal H}_{k,t;\epsilon}(\bx;\boldsymbol{\theta})\defeq &\sum_{k'=k+1}^K h\Big(\bx, \bs,\bz^{(k')},\bm{\phi},K\Delta\Big)+\sum_{k'=1}^{k-1} h_{\rm mf}\Big(\bx, \bs,\widetilde{\bz}^{(k')},K\,\Sigma_{k'}^2\Big)\nonumber\\
&\qquad + h\Big(\bx, \bs,\bz^{(k)},\bm{\phi},\frac{K}{\gamma_k(t)}\Big)+h_{\rm mf}\Big(\bx, \bs,\widetilde{\bz}^{(k)},\frac{K}{\lambda_k(t)}\Big) + \epsilon\sum_{i=1}^n\Big(\frac{x_i^2}{2} -x_is_i -\frac{x_i\widehat z_i}{\sqrt{\epsilon}}\Big).
\end{align}
Again, the last term is a small perturbation needed to use an important concentration 
result (here equation \eqref{MMSErelation}). Here $\boldsymbol{\theta} \defeq \{\bs, \{\bz^{(k)}, \widetilde{\bz}^{(k)}\}_{k=1}^K, \widehat \bz, \bm{\phi}\}$, $k \in \{1,\ldots,K\}$, $t \in [0,1]$ and 
\begin{align}
h(\bx, \bs,\bz,\bm{\phi},\sigma^2) &\defeq\frac{1}{\sigma^2}\sum_{\mu=1}^m \Big(\frac{[\bm{\phi}\bar \bx]_\mu^2}{2} - \sigma[\bm{\phi}\bar \bx]_\mu z_\mu\Big), \label{h_rle}\\
h_{\rm mf}(\bx, \bs,\widetilde{\bz},\sigma^2) &\defeq\frac{1}{\sigma^2}\sum_{i=1}^n \Big(\frac{\bar x_i^2}{2} - \sigma\bar x_i \widetilde z_{i}\Big)\label{hmf_rle}
\end{align}
where $\bar\bx = \bx - \bs$, $\bar x_i = x_i - s_i$.
Moreover the ``signal-to-noise functions'' $\{\gamma_k(t),\lambda_k(t)\}_{k=1}^K$ verify
\begin{align}
\gamma_k(0)&= \Delta^{-1}, \qquad \gamma_{k}(1)=0, \label{112}\\
\lambda_k(0)&= 0, \qquad \hspace{0.5cm}\lambda_{k}(1)= \Sigma_{k}^{-2},\label{113}
\end{align}
as well as the following constraint (see \cite{barbier_ieee_replicaCS} for an interpretation of this formula)
\begin{align}
\frac{\alpha}{\gamma_k(t)^{-1}+E_k} + \lambda_k(t) = \Sigma_k^{-2} \quad \text{and thus} \quad \frac{d\lambda_k(t)}{dt} 
= - \frac{d\gamma_k(t)}{dt}\frac{\alpha}{(1+\gamma_k(t)E_k)^{2}}.	 
\label{contraint}
\end{align}
We also require $\gamma_k(t)$ to be strictly decreasing with $t$. The associated $(k,t)$--interpolating model, Gibbs expectation and $(k,t)$--interpolating free energy are defined respectively by \eqref{post}, \eqref{Gibbs} and \eqref{intf} with the Hamiltonian \eqref{intH_RLE}. Note that Lemma~\ref{thermolimit} remains valid for the present model (with the same proof).

Similarly as in sec.~\ref{initfinalmodels}, and using again the stability property of the Gaussian random noise variables, it is easy to verify that the initial and final $(k,t)$--interpolating models correspond to the RLE and denoising models respectively, that is
\begin{align}
f_{1,0;0}&=f_n, \label{114bis}\\
f_{K,1;0}&=i_{\rm den}\bigl(\Sigma_{\rm mf}(\{E_k\}_{k=1}^K;\Delta) \bigr),
\label{115}
\end{align}
where
\begin{align}
 \Sigma_{\rm mf}^{-2}(\{E_k\}_{k=1}^K;\Delta) \defeq \frac{1}{K}\sum_{k=1}^K \Sigma_k^{-2}\, .
\end{align}
As before we use the identity \eqref{telescopf} and compute the free energy change along the adaptive interpolation. Straightforward differentiation leads to (with $\bar \bX = \bX - \bS$)
\begin{align}
\frac{d f_{k,t;\epsilon}}{dt} &= \frac{1}{K}\big({\cal A}_{k,t;\epsilon} + {\cal B}_{k,t;\epsilon}\big), 
\label{eq:ab}\\
{\cal A}_{k,t;\epsilon} &\defeq 
\frac{d\gamma_k(t)}{dt}\frac{1}{2n}\sum_{\mu=1}^m \mathbb{E}\Big[\Big\langle [\bm{\Phi}\bar \bX]_\mu^2 - 
\sqrt{\frac{K}{\gamma_k(t)}}[\bm{\Phi}\bar \bX]_\mu Z_{\mu}^{(k)} \Big\rangle_{k,t;\epsilon}\Big],
\label{eq:a}\\ 
{\cal B}_{k,t;\epsilon} &\defeq \frac{d\lambda_k(t)}{dt} \frac{1}{2n}\sum_{i=1}^n 
\mathbb{E}\Big[\Big\langle \bar X_i^2 - \sqrt{\frac{K}{\lambda_k(t)}}\bar X_i \widetilde Z_i^{(k)}\Big\rangle_{k,t;\epsilon}\Big],
\label{eq:b}
\end{align}
where as before $\EE$ denotes the average w.r.t to all quenched random variables $\boldsymbol{\theta}$
and $\langle - \rangle_{k,t;\epsilon}$ the Gibbs average with Hamiltonian \eqref{intH_RLE}. The two quantities \eqref{eq:a} and \eqref{eq:b} 
can be simplified using 
Gaussian integration by parts. For example, integrating by parts w.r.t $Z_\mu^{(k)}$,
\begin{align}
\sqrt{\frac{K}{\gamma_k(t)}}\mathbb{E}[\langle [\bm{\Phi}\bar \bX]_\mu \rangle_{k,t;\epsilon}Z_\mu^{(k)}] &= \mathbb{E}[\langle [\bm{\Phi}\bar \bX]_\mu^2 \rangle_{k,t;\epsilon} - \langle [\bm{\Phi}\bar \bX]_\mu \rangle_{k,t;\epsilon}^2]. \label{eq:ippNoise}
\end{align}
It allows to simplify $\mathcal{A}_{k,t;\epsilon}$ as follows, 
\begin{align}
{\cal A}_{k,t;\epsilon} &= \frac{d\gamma_k(t)}{dt} 
\frac{1}{2n}\sum_{\mu=1}^m \mathbb{E}[\langle [\bm{\Phi}\bar \bX]_\mu \rangle_{k,t;\epsilon}^2] 
= \frac{d\gamma_k(t)}{dt} \frac{\alpha}{2m}\mathbb{E}[\|\bm{\Phi}(\langle\bX\rangle_{k,t;\epsilon} - \bS)\|^2] 
= \frac{d\gamma_k(t)}{dt} \frac{\alpha}{2} {\rm ymmse}_{k,t;\epsilon}, 
\label{eq:was71}
\end{align}
where we recognized the ``measurement minimum mean-square-error'' 
\begin{align}
{\rm ymmse}_{k,t;\epsilon}\defeq \frac{1}{m}\mathbb{E}[\|\bm{\Phi}(\langle\bX\rangle_{k,t;\epsilon}  -  \bS)\|^2].
\end{align}
For $\mathcal{B}_{k,t;\epsilon}$ we proceed similarly with an integration by parts w.r.t $\widetilde Z_i^{(k)}$, and find
\begin{align}
{\cal B}_{k,t;\epsilon} &=\frac{d\lambda_k(t)}{dt} \frac{1}{2n}\sum_{i=1}^n \mathbb{E}[\langle \bar X_i \rangle_{k,t;\epsilon}] = \frac{d\lambda_k(t)}{dt}  \frac{1}{2n}\mathbb{E}[\|\langle \bX \rangle_{k,t;\epsilon} - \bS\|^2] \nn
&= - \frac{d\gamma_k(t)}{dt}\frac{1}{(1+\gamma_k(t)E_k)^{2}}\frac{\alpha}{2} {\rm mmse}_{k,t;\epsilon}, 
\label{eq:was71_2}
\end{align}
using \eqref{contraint} for the last equality, and the minimum mean-square-error (MMSE) defined as
\begin{align}
 {\rm mmse}_{k,t;\epsilon} \defeq \frac{1}{n}\mathbb{E}[\|\langle \bX \rangle_{k,t;\epsilon} - \bS\|^2].
\end{align}

The free energy can be shown to concentrate by generalizing the computations of Appendix E in \cite{barbier_ieee_replicaCS} taking into account that the noise variables $\{Z_\mu^{(k)}, \widetilde Z_i^{(k)}\}$ are indexed by the discrete interpolation parameter 
(the techniques of \cite{barbier_ieee_replicaCS} use a discrete $P_0$ with bounded support for the free energy concentration). Since the free energy at fixed 
quenched random variables realization concentrates, both sec.~VIII of \cite{barbier_ieee_replicaCS} or sec.~\ref{proofConc} of the 
present paper apply here (these are perfectly equivalent analyses and only require the identity 
\eqref{nishibasic} and the free energy concentration to be valid). Thus the overlap $q_{\bx,\bs} \defeq n^{-1}\sum_i x_i s_i$ concentrates too. As a consequence an analog of Lemma~4.6 in \cite{barbier_ieee_replicaCS} can be shown here: Fix a discrete $P_0$ with bounded support. For any sequence 
$K_n\to +\infty$, and $0<a_n<b_n<1$ (that tend to zero slowly enough in the application), and 
trial parameters $\{E_k = E_k^{(n)}(\epsilon)\}_{k=1}^{K_n}$ which are differentiable, bounded and non-increasing in $\epsilon$, we have
\begin{align}
\int_{a_n}^{b_n} d\epsilon\, \frac{1}{K_n}\sum_{k=1}^{K_n}\int_0^1 dt\frac{d\gamma_k(t)}{dt}\biggl\{{\rm ymmse}_{k,t;\epsilon} - \frac{{\rm mmse}_{k,t;\epsilon}}{1+\gamma_k(t){\rm mmse}_{k,t;\epsilon}}\biggr\} = \mathcal{O}(a_n^{-2} n^{-\alpha}), \label{MMSErelation}
\end{align}	 
for some $0<\alpha<1$ and $C>0$.  

Now combining \eqref{telescopf}, \eqref{114bis}, \eqref{115}, \eqref{eq:ab} and \eqref{eq:was71}, \eqref{eq:was71_2}, 
together with \eqref{MMSErelation}, we obtain
\begin{align}
\int_{a_n}^{b_n} d\epsilon\, f_{1,0;\epsilon} &=\int_{a_n}^{b_n} d\epsilon\,\biggl\{(f_{K_n,1;\epsilon} - f_{K_n,1;0}) +i_{\rm den}\big(\Sigma_{\rm mf}(\{E_k\}_{k=1}^{K_n};\Delta)\big) 
\nonumber\\ &
- \frac{\alpha}{2K_n}\sum_{k=1}^{K_n} \int_{0}^1 dt\frac{d\gamma_k(t)}{dt}\Big(\frac{{\rm mmse}_{k,t;\epsilon}}{1+\gamma_k(t){\rm mmse}_{k,t;\epsilon}} - \frac{{\rm mmse}_{k,t;\epsilon}}{(1+\gamma_k(t)E_k)^{2}}\Big)\biggr\} + \mathcal{O}(a_n^{-2} n^{-\alpha}) \label{124}.
\end{align}
We need the following useful identity which can easily be checked using \eqref{psi}, \eqref{112}, \eqref{113}, \eqref{contraint}:
\begin{align}\label{intermediate-psi-equ}
\psi(E_k;\Delta) = \frac{\alpha}{2}\int_0^1 dt\, \frac{d\gamma_k(t)}{dt} \Big(\frac{E_k}{(1+\gamma_k(t) E_k)^2} - \frac{E_k}{1+\gamma_k(t) E_k}\Big).
\end{align}
Let us define
\begin{align}
\widetilde f_{\rm RS}(\{E_k\}_{k=1}^K;\Delta) \defeq i_{\rm den}\big(\Sigma_{\rm mf}(\{E_k\}_{k=1}^K;\Delta)\big) + \frac{1}{K}\sum_{k=1}^K \psi(E_k;\Delta).
\label{ftilde}
\end{align}
With the help of \eqref{intermediate-psi-equ} and \eqref{ftilde} the identity \eqref{124} becomes
\begin{align}
\int_{a_n}^{b_n} d\epsilon\, 
 f_{1,0;\epsilon} &= \int_{a_n}^{b_n} d\epsilon\,\biggl\{(f_{K_n,1;\epsilon} - f_{K_n,1;0}) +\widetilde f_{\rm RS}(\{E_k\}_{k=1}^{K_n};\Delta)
-\frac{\alpha}{2K_n}\sum_{k=1}^{K_n} \int_{0}^1 dt\frac{d\gamma_k(t)}{dt}
\nonumber \\ &
\times
\Big(\frac{{\rm mmse}_{k,t;\epsilon}}{1+\gamma_k(t){\rm mmse}_{k,t;\epsilon}} 
- \frac{{\rm mmse}_{k,t;\epsilon}}{(1+\gamma_k(t)E_k)^{2}}+\frac{E_k}{(1+\gamma_k(t) E_k)^2} - \frac{E_k}{1+\gamma_k(t) E_k}\Big)
\biggr\} + \mathcal{O}(a_n^{-2} n^{-\alpha})
\nonumber \\ & 
=
\int_{a_n}^{b_n} d\epsilon\,\biggl\{(f_{K_n,1;\epsilon} - f_{K_n,1;0}) + \widetilde f_{\rm RS}(\{E_k\}_{k=1}^{K_n};\Delta) 
\nonumber \\ &  \quad
+\frac{\alpha}{2K_n}\sum_{k=1}^{K_n} \int_{0}^1 dt\frac{d\gamma_k(t)}{dt} \frac{\gamma_k(t)(E_k - {\rm mmse}_{k,t;\epsilon})^2}{(1+\gamma_k(t) E_k)^2(1+\gamma_k(t){\rm mmse}_{k,t;\epsilon})}
\biggr\}
+ \mathcal{O}(a_n^{-2} n^{-\alpha}). \label{127}
\end{align}
This is the fundamental sum rule which forms the basis for the proof of Theorem~\ref{rsformulaRLE}. 

We start with the upper bound. As in sec.~\ref{secUpper} we choose $E_{k} = E_* \defeq {\rm argmin}_{E\ge 0} f_{\rm RS}(E;\Delta)$ for all $k = 1,\ldots,K_n$ (here $\epsilon$-independent) which implies that $\Sigma_{\rm mf}(\{E_k = E_*\}_{k=1}^{K_n};\Delta)  = \Sigma(E_*;\Delta)$ and thus, as seen from \eqref{ftilde}, $\widetilde f_{\rm RS}(\{E_k = E_*\}_{k=1}^{K_n};\Delta) = \min_{E\ge 0}f_{\rm RS}(E;\Delta)$. Thus since the integrand in \eqref{127} is non-positive (recall that $d\gamma_k(t)/dt \leq 0$) and using the arguments similar to sec.~\ref{stochInt} in order to take the $n\to +\infty$ limit, we get:
\begin{proposition}[Upper bound] \label{UpperBound_RLE}
Fix $\Delta > 0$. For $P_0$ discrete and with bounded support:
\begin{align}
\limsup_{n\to\infty} f_n \le \min_{E\ge 0}f_{\rm RS}(E;\Delta).
\end{align}	
\end{proposition}

Let us now prove the lower bound. This bound required the use of spatial coupling 
in \cite{barbier_allerton_RLE,barbier_ieee_replicaCS} or 
``conditional central limit theorems'' in \cite{private,ReevesP16}. Here we derive the bound in a direct and much simpler manner
following the same steps as in sec.~\ref{seclower}. We first need the following identity: For any discrete $P_0$ with bounded support, any $k \in \{1,\ldots,K\}$ and $\epsilon, t \in [0,1]$,
\begin{align} 
|{\rm mmse}_{k,t;\epsilon} -  {\rm mmse}_{k,0;\epsilon}| = {\cal O}\Big(\frac{n}{K}\Big). 
\end{align}
Its proof is very similar to the one of Lemma \ref{tinvar}. Using this identity with $K_n =  \Omega(n^b)$, $b > 2$, in \eqref{127} and 
constructing $E_k^{(n)}(\epsilon) = {\rm mmse}_{k,0;\epsilon}$ (which is indeed non-increasing with $\epsilon$ being a MMSE) for all $k = 1,\ldots,K_n$ 
(by the same arguments than those in the proof of Lemma~\ref{freedom}), we reach
\begin{align}
\int_{a_n}^{b_n} d\epsilon\, f_{1,0;\epsilon} = \int_{a_n}^{b_n}d\epsilon\, \biggl\{(f_{K_n,1;\epsilon} - f_{K_n,1;0}) + \widetilde f_{\rm RS}(\{E_k^{(n)}\}_{k=1}^{K_n};\Delta)\biggr\}  +  \mathcal{O}(a_n^{-2} n^{-\alpha}).
\label{127_}
\end{align}
Recall $\Sigma^{-2} \defeq \alpha/ (E+ \Delta)$ and thus $E= \alpha/\Sigma^{-2}-\Delta$. 
For given $\Delta$ we set $\widetilde \psi_{\Delta}(\Sigma^{-2}) \defeq \psi( \alpha/\Sigma^{-2}-\Delta;\Delta)$ 
and note that $\widetilde\psi_{\Delta}(\cdot)$ is a convex function. Thus from \eqref{ftilde}
\begin{align}
\widetilde f_{\rm RS}(\{E_k\}_{k=1}^{K_n};\Delta) & = i_{\rm den}\big(\Sigma_{\rm mf}(\{E_k\}_{k=1}^{K_n};\Delta)\big) + \frac{1}{K_n}\sum_{k=1}^{K_n} \widetilde\psi_{\Delta}(\Sigma_k^{-2})
\nonumber \\ &
\geq 
i_{\rm den}\big(\Sigma_{\rm mf}(\{E_k\}_{k=1}^{K_n};\Delta)\big) 
+ \widetilde\psi_\Delta\big(\Sigma_{\rm mf}^{-2}(\{E_k\}_{k=1}^{K_n};\Delta)\big)
\nonumber \\ &
\geq 
\min_{\Sigma \geq 0} \Big(i_{\rm den}(\Sigma) + \widetilde\psi_\Delta(\Sigma^{-2})\Big)
\nonumber \\ &
= \min_{E\geq 0} f_{\rm RS}(E; \Delta).
\end{align}
Thus \eqref{127_} becomes 
\begin{align}
\int_{a_n}^{b_n} d\epsilon\, f_{1,0;\epsilon} \ge \int_{a_n}^{b_n} d\epsilon\, (f_{K_n,1;\epsilon} - f_{K_n,1;0}) + (b_n-a_n)\min_{E \ge 0} f_{\rm RS}(E;\Delta) + \mathcal{O}(a_n^{-2} n^{-\alpha}). \label{finalstuff}
\end{align}
Taking $b_n=2a_n$, $a_n\to 0$, such that $a_n^{-3} n^{-\alpha} \to 0$ as $n \to +\infty$ we obtain 
(recall Lemma~\ref{thermolimit})
%
\begin{proposition}
[Lower bound] \label{LowerBound_RLE}
Fix $\Delta > 0$. For any dicrete $P_0$ with bounded support, 
\begin{align}
\liminf_{n\to\infty} f \ge \min_{E\ge 0}f_{\rm RS}(E;\Delta).
\end{align}	
\end{proposition}
\section{Concentration of overlaps}\label{proofConc}

The main goal of this section is the proof of Lemma \ref{concentration}.
The proof strategy outlined here is very general and it will appear to the reader that it applies to essentially any inference 
problem for which the identity \eqref{nishibasic} is valid and as long as the free energy can be shown to concentrate. In the framework of inference problems 
such proofs go back to \cite{GKSmacris2007,KoradaMacris_CDMA,korada2009exact} for 
binary signals (in coding, CDMA and the gauge symmetric p-spin model) and have been extended more recently in random linear estimation for 
arbitrary signal distributions \cite{barbier_ieee_replicaCS}. The results and exposition given here slightly generalize and streamlines the one of the previous works. 

From now on the trial parameters are chosen of the form $\{m_k= m_k^{(n)}(\epsilon)\}_{k=1}^K$.
It will be convenient to adopt the notation 
$\tilde\epsilon = \epsilon + (K\Delta)^{-1}(\sum_{l=1}^{k-1} m_l^{(n)}(\epsilon) + t m_k^{(n)}(\epsilon))$. Here $\tilde\epsilon$ depends on $k,t$ but we do not write this 
dependence explicitly as it does not play a role
(we work at fixed $k,t$ in the rest of this section). Let 
\begin{align}\label{L-def}
\mathcal{L} \defeq \frac{1}{n}\sum_{i=1}^n\Big(\frac{x_i^2}{2} - x_i s_i - \frac{x_i \widehat z_i}{2\sqrt{\tilde\epsilon}} \Big).
\end{align}
We will show 
that Lemma \ref{concentration} is a direct consequence of the following:
\begin{proposition}[Concentration of $\mathcal{L}$ on $\mathbb{E}\langle\mathcal{L}\rangle_{k,t;\epsilon}$ ] \label{L-concentration}
Let $P_0$ with finite second moment and bounded support in $[-M, M]$. For any choice of trial parameters $\{m_k= m_k^{(n)}(\epsilon)\}_{k=1}^K$ that are non-decreasing bounded and 
differentiable functions of $\epsilon\in ]0,1[$, and any sequences $0<a_n<b_n<1$,
we have
\begin{align}
\int_{a_n}^{b_n} d\epsilon\,
\mathbb{E}\big[ \big\langle (\mathcal{L} - \mathbb{E}[\langle \mathcal{L}\rangle_{k, t;\epsilon}])^2\big\rangle_{k, t;\epsilon}\big] \leq \frac{C}{a_n^2 \, n^{\alpha}}
\end{align}
for any $0< \alpha < 1/4$ 
with $C>0$ a constant uniform in $k,t$ and the trial parameters and depending only on the second moment of $P_0$ and $M$.
\end{proposition}

The proof of this proposition is broken in two parts. Notice that 
\begin{align}
\mathbb{E}\big[ \big\langle (\mathcal{L} - \mathbb{E}[\langle \mathcal{L}\rangle_{k,t;\epsilon}])^2\big\rangle_{k,t;\epsilon}\big]
& = 
\mathbb{E}\big[ \big\langle (\mathcal{L} - \langle \mathcal{L}\rangle_{k,t;\epsilon})^2\big\rangle_{k,t;\epsilon}\big] 
+ \mathbb{E}\big[(\langle \mathcal{L}\rangle_{k,t;\epsilon} - \mathbb{E}[\langle \mathcal{L}\rangle_{k,t;\epsilon}])^2\big].
\end{align}
Thus it suffices to prove the two following lemmas. The first lemma expresses concentration w.r.t the posterior distribution (or ``thermal fluctuations'') and is an elementary consequence of concavity properties of the free energy.
\begin{lemma}[Concentration of $\mathcal{L}$ on $\langle \mathcal{L}\rangle_{k,t;\epsilon}$ ]\label{thermal-fluctuations}
Let $P_0$ with finite second moment.
For any choice of trial parameters $\{m_k= m_k^{(n)}(\epsilon)\}_{k=1}^K$ that are non-decreasing bounded and differentiable functions of $\epsilon\in ]0,1[$, and any sequences $0<a_n<b_n<1$,
\begin{align}\label{integral-form}
 \int_{a_n}^{b_n} d\epsilon\, 
 \mathbb{E}\big[ \big\langle (\mathcal{L} - \langle \mathcal{L}\rangle_{k,t;\epsilon})^2 \big\rangle_{k,t;\epsilon} \big] 
 \leq \frac{\mathbb{E}[S^2]}{n} \Big(1 + \frac{\vert \ln a_n\vert}{4}\Big) \,.
\end{align}
\end{lemma}

The second lemma expresses the concentration of the Gibbs average w.r.t the realizations of quenched disorder variables.
\begin{lemma}[Concentration of $\langle\mathcal{L}\rangle_{k,t;\epsilon}$ on $\mathbb{E}\langle \mathcal{L}\rangle_{k,t;\epsilon}$ ]\label{disorder-fluctuations}
Let $P_0$ with finite second moment and bounded support in $[-M, M]$. For any choice of trial parameters $\{m_k= m_k^{(n)}(\epsilon)\}_{k=1}^K$ that are non-decreasing bounded and 
differentiable functions of $\epsilon\in ]0,1[$, and any sequences $0<a_n<b_n<1$,

\begin{align}\label{integral-form}
 \int_{a_n}^{b_n} d\epsilon\, 
 \mathbb{E}\big[ (\langle \mathcal{L}\rangle_{k,t;\epsilon} - \mathbb{E}[\langle \mathcal{L}\rangle_{k,t;\epsilon}])^2 \big] \leq \frac{C}{a_n^2n^{\frac{1}{4} - \frac{\eta}{2}} }
\end{align}
for any  $0 < \eta  < 1/2$ and 
where $C>0$ depends only on the second moment of $P_0$ and $M$. In particular $C$ is independent of $k,t$ and the trial parameters.
\end{lemma}

\begin{remark}
Thanks to the identity \eqref{82} below, that we will show in section \ref{fluctuation-identity}, the statements of Proposition \ref{L-concentration} and Lemmas \ref{thermal-fluctuations} 
and \ref{disorder-fluctuations} hold if we replace $\mathcal{L}$ by the overlap $q_{\bx,\bs}$. 
\end{remark}

The proof of this last lemma is based on 
an important but generic result concerning the concentration of the $(k,t)$--interpolating free energy for a 
single realization of quenched variables. Let 
\begin{align}\label{free-non-averaged}
 F_{k, t; \epsilon}(\boldsymbol{\theta}) \defeq -\frac{1}{n} \ln \int \big\{\prod_{i=1}^n dx_i P_0(x_i)\big\} e^{-\mathcal{H}_{k, t;\epsilon}(\bx; \boldsymbol{\theta})}\, .
\end{align}
Recall that $f_{k, t; \epsilon} = \mathbb{E}[F_{k, t; \epsilon}(\boldsymbol{\Theta})]$.
\begin{proposition}[Concentration of the $(k,t)$--interpolating free energy]\label{conc-free}
Let $P_0$ with bounded support in $[-M, M]$. One can find $c  > 0$ which depends only on $M$ and $\Delta$
such that for all $k = 1, \ldots, K$, $t \in  [0,1]$ and $\epsilon \in  [0, 1]$,
\begin{align}
 \mathbb{P}\big[ \vert F_{k, t; \epsilon}(\boldsymbol{\Theta}) - f_{k, t; \epsilon} \vert > u  \big] \leq e^{- c n u^2}
\end{align}
where $u>0$. Explicit expressions for $c$ can be derived from \eqref{explicit} in sec.~\ref{concentration-free-energy}.
\end{proposition}

This proposition is proved in sec.~\ref{concentration-free-energy}. In the rest of this section 
we prove Lemmas \ref{concentration}, \ref{thermal-fluctuations} and \ref{disorder-fluctuations}. 
The parameters $k$ and $t$ stay fixed and do not play any role, but it is important to be careful about the $\epsilon$ dependence.

\subsubsection*{Proof of Lemma \ref{concentration}}
The proof is based on the remarkable identity (here $S\sim P_0$)
\begin{align}
\mathbb{E}\big[\big\langle (\mathcal{L} - \mathbb{E}[\langle \mathcal{L}\rangle_{k,t;\epsilon}])^2\big\rangle_{k,t;\epsilon}\big]
= \,& \frac{1}{4}\big(\mathbb{E}[\langle q_{\bX,\bS}^2\rangle_{k,t;\epsilon}] - \mathbb{E}[\langle q_{\bX,\bS}^{}\rangle_{k,t;\epsilon}]^2\big)
 + \frac{1}{2}\big(\mathbb{E}[\langle q_{\bX,\bS}^2\rangle_{k,t;\epsilon}] -   \mathbb{E}[\langle q_{\bX,\bS}^{}\rangle_{k,t;\epsilon}^2]\big)
 \nonumber \\ &
\qquad + \frac{1}{4n \tilde\epsilon} \mathbb{E}[S^2]\,. \label{82}
\end{align}
Its derivation is found in sec. \ref{fluctuation-identity} and involves lengthy algebra using identity \eqref{nishibasic} and integrations by parts w.r.t the Gaussian noise. This formula implies
\begin{align}
\mathbb{E}\big[\big\langle (q_{\bX,\bS}^{} -   \mathbb{E}[\langle q_{\bX,\bS}^{}\rangle_{k,t;\epsilon}])^2\big\rangle_{k,t;\epsilon}\big]
\leq 
4\mathbb{E}\big[\big\langle (\mathcal{L} - \mathbb{E}[\langle \mathcal{L}\rangle_{k,t;\epsilon}])^2\big\rangle_{k,t;\epsilon}\big]
\label{nice-inequ}
\end{align}
and using Fubini's theorem
\begin{align}
\int_{a_n}^{b_n} d\epsilon \frac{1}{K_n}\sum_{k=1}^{K_n}\int_0^1 dt\, 
\mathbb{E}\big[\big\langle (q_{\bX,\bS}^{} -   & \mathbb{E}[\langle q_{\bX,\bS}^{}\rangle_{k,t;\epsilon}])^2\big\rangle_{k,t;\epsilon}\big]
\nonumber \\ &
\leq 
\frac{4}{K_n}\sum_{k=1}^{K_n}\int_0^1 dt 
\int_{a_n}^{b_n} d\epsilon\, \mathbb{E}\big[\big\langle (\mathcal{L} - \mathbb{E}[\langle \mathcal{L}\rangle_{k,t;\epsilon}])^2\big\rangle_{k,t;\epsilon}\big]\,.
\label{nice-inequ2}
\end{align}
Then applying Proposition \ref{L-concentration}
we obtain (since the bounds are uniform in $k,t$)
\begin{align}
 \int_{a_n}^{b_n} d\epsilon
\frac{1}{K_n}\sum_{k=1}^{K_n}\int_0^1dt\, 
\mathbb{E}\big[\big\langle (q_{\bX,\bS}^{} - \mathbb{E}[\langle q_{\bX,\bS}^{}\rangle_{k,t;\epsilon}])^2\big\rangle_{k,t;\epsilon}\big]
\leq \frac{4C}{a_n^2\,n^{\alpha}}\,.
\end{align}
so that \eqref{concentration-overlap} is verified for any $0<\alpha<1/4$.
\QEDA
\\ 

We now turn to the proof of Lemmas \ref{thermal-fluctuations} and \ref{disorder-fluctuations}. The main ingredient is a set of formulas for the first two derivatives of the free energy 
w.r.t $\tilde\epsilon$.  
For any given realisation of the quenched disorder we have the equalities (here $Z_i\sim\mathcal{N}(0,1)$ i.i.d)
\begin{align}
 \frac{dF_{k, t; \epsilon}(\boldsymbol{\theta})}{d\tilde\epsilon}  & = \langle \mathcal{L} \rangle_{k,t;\epsilon}\,,\label{first-derivative}\\
 \frac{1}{n}\frac{d^2F_{k, t; \epsilon}(\boldsymbol{\theta})}{d\tilde{\epsilon}^2}  &= - (\langle \mathcal{L}^2 \rangle_{k,t;\epsilon} - \langle \mathcal{L} \rangle_{k,t;\epsilon}^2) +
 \frac{1}{4 n^2\tilde\epsilon^{3/2}}\sum_{i=1}^n  \langle X_i\rangle_{k,t;\epsilon} z_i\label{second-derivative}\,.
\end{align}
Averaging \eqref{first-derivative} and \eqref{second-derivative} and using a Gaussian integration by parts w.r.t $z_i$ and the identity
$\mathbb{E}[\langle X_i\rangle_{k,t;\epsilon} S_i]  =  \mathbb{E}[\langle X_i\rangle_{k,t;\epsilon}^2]$ 
(again a special case of \eqref{nishibasic}), we find (see Appendix \ref{appendix-exchg-lim})
\begin{align}
 \frac{df_{k, t;\epsilon}}{d\tilde\epsilon} 
 &= \mathbb{E}[\langle \mathcal{L} \rangle_{k,t;\epsilon}] 
 = - \frac{1}{2n} \sum_{i=1}^n\mathbb{E}[\langle X_i\rangle_{k,t;\epsilon}^2] \,,\label{first-derivative-average}\\
 \frac{1}{n}\frac{d^2f_{k, t;\epsilon}}{d\tilde{\epsilon}^2} &= - \mathbb{E}[\langle \mathcal{L}^2 \rangle_{k,t;\epsilon} - \langle \mathcal{L} \rangle_{k,t;\epsilon}^2]
 +\frac{1}{4n^2\tilde\epsilon} \sum_{i=1}^n \mathbb{E}[\langle X_i^2\rangle_{k,t;\epsilon} - \langle X_i\rangle_{k,t;\epsilon}^2]\label{average-second-derivative}\,.
\end{align}
There is another useful formula for $d^2f_{k, t;\epsilon}/d{\tilde \epsilon}^2$ that can be worked 
out directly (see sec. \ref{fluctuation-identity}) by differentiating the second expression in \eqref{first-derivative-average} instead of the first:
\begin{align}
\frac{1}{n}\frac{d^2f_{k, t;\epsilon}}{d\tilde{\epsilon}^2}  & 
= \frac{1}{2n}\sum_{i=1}^n \mathbb{E}[2\langle X_i\rangle_{k,t;\epsilon} \langle X_i \mathcal{L}\rangle_{k,t;\epsilon} 
- 2 \langle X_i \rangle_{k,t;\epsilon}^2 \langle \mathcal{L}\rangle_{k,t;\epsilon}] 
\nonumber \\ &
= -\frac{1}{2n^2}\sum_{i, j =1}^n \mathbb{E}[(\langle X_i X_j\rangle_{k,t;\epsilon} - \langle X_i\rangle_{k,t;\epsilon}\langle X_j\rangle_{k,t;\epsilon})^2]\,.
\label{alternative}
\end{align}
This formula clearly shows that $f_{k, t;\epsilon}$ is a concave function of $\tilde\epsilon$. 

\subsubsection*{Proof of Lemma \ref{thermal-fluctuations}}
From \eqref{average-second-derivative} we have 
\begin{align}
\mathbb{E}\big[\big\langle (\mathcal{L} - \langle \mathcal{L} \rangle_{k,t;\epsilon})^2\big\rangle_{k,t;\epsilon}\big]
& = 
-\frac{1}{n}\frac{d^2f_{k, t;\epsilon}}{d\tilde{\epsilon}^2}
+\frac{1}{4n^2\tilde\epsilon} \sum_{i=1}^n \mathbb{E}[\langle X_i^2\rangle_{k,t;\epsilon} - \langle X_i\rangle_{k,t;\epsilon}^2] 
\nonumber \\ &
\leq 
-\frac{1}{n}\frac{d^2f_{k, t;\epsilon}}{d\tilde{\epsilon}^2} +\frac{\mathbb{E}[S^2]}{4n\epsilon} ,
\end{align}
where we used $\tilde{\epsilon}\geq \epsilon$ and $\mathbb{E}[\langle X_i^2\rangle_{k,t;\epsilon}] = \mathbb{E}[S^2]$ (an application of \eqref{nishibasic}). 
We perform an integration of this inequality over $\epsilon\in [a_n, b_n]$. Note that the map 
$\epsilon\in [a_n, b_n]\mapsto \tilde\epsilon\in [\tilde\epsilon(a_n), \tilde\epsilon(b_n)]$ is differentiable and the inverse map is well defined and also differentiable
since we have assumed that the trial parameters are differentiable and non decreasing. Obviously the Jacobian $J = d\tilde\epsilon/d\epsilon \geq 1$ since the trial parameters are non-decreasing.
Integrating over $\epsilon\in [a_n, b_n]$ and performing the change of variables $\epsilon\mapsto \tilde\epsilon$, and using $J\geq 1$, we obtain 
\begin{align}
\int_{a_n}^{b_n} d\epsilon\, \mathbb{E}\big[\big\langle (\mathcal{L} - \langle \mathcal{L} \rangle_{k,t;\epsilon})^2\big\rangle_{k,t;\epsilon}\big]
& \leq 
- \frac{1}{n}\int_{a_n}^{b_n} d\epsilon \,\frac{d^2f_{k, t;\epsilon}}{d\tilde{\epsilon}^2} + \frac{\mathbb{E}[S^2]}{4n}\int_{a_n}^{b_n} \,\frac{d\epsilon}{\epsilon} 
\nonumber \\ &
=
- \frac{1}{n}\int_{\tilde\epsilon(a_n)}^{\tilde\epsilon(b_n)} \frac{d\tilde\epsilon}{J} \,\frac{d^2f_{k, t;\epsilon}}{d\tilde{\epsilon}^2} 
+ \frac{\mathbb{E}[S^2]}{4n}\int_{a_n}^{b_n} \frac{d\epsilon}{\epsilon}
\nonumber \\ &
\leq 
- \frac{1}{n}\int_{\tilde\epsilon(a_n)}^{\tilde\epsilon(b_n)} d\tilde\epsilon \,\frac{d^2f_{k, t;\epsilon}}{d\tilde{\epsilon}^2} 
+ \frac{\mathbb{E}[S^2]}{4n}\int_{a_n}^{b_n} \,\frac{d\epsilon}{\epsilon} 
\nonumber \\ &
\leq
\Big(\frac{1}{n}\frac{df_{k, t; \epsilon}}{d\tilde\epsilon}\Big\vert_{\tilde\epsilon(a_n)} 
- \frac{1}{n}\frac{df_{k, t; \epsilon}}{d\tilde\epsilon}\Big\vert_{\tilde\epsilon(b_n)}\Big) + \frac{\mathbb{E}[S^2]}{4n}(\ln b_n-\ln a_n)\,.
\label{93}
\end{align}
From \eqref{first-derivative-average} combined with the convexity of the square and an application of the Nishimori identity, we see that the first term is certainly smaller in absolute value than $\frac{1}{n}\E[S^2]$. 
The second term is smaller than $\frac{\mathbb{E}[S^2]}{4n}\vert\ln a_n\vert$.
This concludes the proof of Lemma \ref{thermal-fluctuations}.
\QEDA

\subsubsection*{Proof of Lemma \ref{disorder-fluctuations}}
In what follows we view $F_{k, t;\epsilon}(\boldsymbol{\theta})$ and $f_{k, t;\epsilon}$ as functions of 
$\tilde\epsilon$. Recall that $P_0$ has bounded support in $[-M, M]$. Define the two functions of $\tilde\epsilon$
\begin{align}\label{new-free}
 \widetilde F(\tilde\epsilon) \defeq F_{k, t;\epsilon}(\boldsymbol{\theta}) +\frac{\sqrt{\tilde\epsilon}}{n} 
 \sum_{i=1}^n M\vert \widehat z_i\vert,
 \quad 
 \widetilde f(\tilde\epsilon) \defeq f_{k, t;\epsilon} + \frac{\sqrt{\tilde\epsilon}}{n} \sum_{i=1}^n M\, \mathbb{E}[\vert \widehat z_i\vert].
\end{align}
Because of 
\eqref{second-derivative}  we see that the second derivative of 
$\widetilde F(\tilde\epsilon)$ is negative, so this is a {\it concave} function of $\tilde\epsilon$
(without this extra term $F_{k,t;\epsilon}$ is not necessarily concave, although $f_{k,t;\epsilon}$ is concave). 
Note also that $\widetilde f(\tilde\epsilon)$ is concave.
Concavity implies for any $\delta>0$
\begin{align}
\frac{d \widetilde F(\tilde\epsilon)}{d\tilde\epsilon} - \frac{d\widetilde f(\tilde\epsilon)}{d\tilde\epsilon} \le\,& \frac{\widetilde F(\tilde\epsilon)-\widetilde F(\tilde\epsilon-\delta)}{\delta}
- \frac{d\widetilde f(\tilde\epsilon)}{d\tilde\epsilon} \nn
\le \,&\frac{\widetilde F(\tilde\epsilon) - \widetilde f(\tilde\epsilon)}{\delta} 
- \frac{\widetilde F(\tilde\epsilon-\delta) - \widetilde f(\tilde\epsilon-\delta)}{\delta} + \frac{d\widetilde f(\tilde\epsilon-\delta)}{d\tilde\epsilon}
- \frac{d\widetilde f(\tilde\epsilon)}{d\tilde\epsilon},  
\label{eq:firstBound_df}\\
\frac{d \widetilde F(\tilde\epsilon)}{d\tilde\epsilon} - \frac{d\widetilde f(\tilde\epsilon)}{d\tilde\epsilon} \ge \,&\frac{\widetilde F(\tilde\epsilon+\delta) - 
\widetilde f(\tilde\epsilon+\delta)}{\delta} - \frac{\widetilde F(\tilde\epsilon) - \widetilde f(\tilde\epsilon)}{\delta} 
+ \frac{d\widetilde f(\tilde\epsilon +\delta)}{d\tilde\epsilon} - \frac{d\widetilde f(\tilde\epsilon)}{d\tilde\epsilon}. 
\label{eq:secondBound_df}
\end{align}
The difference between the derivatives appearing on the r.h.s of these inequalities cannot be considered small because at a first order transition 
point the derivatives have jump discontinuities. Set
\begin{align}
-C^-(\tilde\epsilon)&\defeq \frac{d\widetilde f(\tilde\epsilon+\delta)}{d\tilde\epsilon} - \frac{d\widetilde f(\tilde\epsilon)}{d\tilde\epsilon} \le 0, 
\qquad C^+(\tilde\epsilon)\defeq\frac{d\widetilde f(\tilde\epsilon -\delta)}{d\tilde\epsilon} - \frac{d\widetilde f(\tilde\epsilon)}{d\tilde\epsilon} \ge 0, 
\label{175}
\end{align}
where the signs of these quantities follow from concavity of $\widetilde f(\tilde\epsilon)$. From \eqref{eq:firstBound_df}, \eqref{eq:secondBound_df} and \eqref{175} 
we get
\begin{align}\label{long-inequ}
\frac{\widetilde F(\tilde\epsilon+\delta) - \widetilde f(\tilde\epsilon+\delta)}{\delta} - \frac{\widetilde F(\tilde\epsilon) - \widetilde f(\tilde\epsilon)}{\delta} 
- C^-(\tilde\epsilon) &\le \frac{d \widetilde F(\tilde\epsilon)}{d\tilde\epsilon} - \frac{d\widetilde f(\tilde\epsilon)}{d\tilde\epsilon} \nn
&\le \frac{\widetilde F(\tilde\epsilon) 
 -  \widetilde f(\tilde\epsilon)}{\delta} - \frac{\widetilde F(\tilde\epsilon-\delta) - \widetilde f(\tilde\epsilon-\delta)}{\delta} + C^+(\tilde\epsilon)\,.
\end{align}
Now we will cast this inequality in a more usable form. From \eqref{new-free}
\begin{align}\label{fdiff}
 \widetilde F(\tilde\epsilon) - \widetilde f(\tilde\epsilon) = F_{k, t;\epsilon}(\boldsymbol{\theta}) - f_{k, t; \epsilon} + \sqrt{\tilde\epsilon} M A 
\end{align}
with 
\begin{align}
 A = \frac{1}{n}\sum_{i=1}^n \big(\vert \widehat z_i\vert -\mathbb{E}[\vert \widehat z_i\vert]\big)
\end{align} 
and from \eqref{first-derivative}, \eqref{first-derivative-average},
\begin{align}\label{derdiff}
 \frac{d \widetilde F(\tilde\epsilon)}{d\tilde\epsilon} - \frac{d\widetilde f(\tilde\epsilon)}{d\tilde\epsilon} = 
 \langle \mathcal{L}\rangle_{k,t;\epsilon} - \mathbb{E}[\langle \mathcal{L}\rangle_{k,t;\epsilon}] + \frac{M}{2\sqrt{\tilde\epsilon}} A.
\end{align}
From \eqref{fdiff}, \eqref{derdiff} it is easy to show that \eqref{long-inequ} implies
\begin{align}\label{usable-inequ}
&\big\vert \langle \mathcal{L}\rangle_{k,t;\epsilon} - \mathbb{E}[\langle \mathcal{L}\rangle_{k,t;\epsilon}]\big\vert\nn
&\qquad\qquad\leq 
\delta^{-1} \sum_{u\in \{\tilde\epsilon -\delta, \tilde\epsilon, \tilde\epsilon+\delta\}}
 \big(\vert F_{k, t;u}(\boldsymbol{\theta}) - f_{k, t;u} \vert + M\vert A \vert \sqrt{u} \big)
  + C^+(\tilde\epsilon) + C^-(\tilde\epsilon) + \frac{M}{2\sqrt{\tilde\epsilon}} \vert A\vert.
\end{align}
At this point we use Proposition \ref{conc-free}. A standard argument given at the end of this proof shows that this proposition implies 
\begin{align}\label{implies}
 \mathbb{E}[ ( F_{k, t; \epsilon}(\boldsymbol{\Theta}) - f_{k, t; \epsilon})^2 ]  = \mathcal{O}(n^{-1 + \eta})
\end{align}
for any $0 <  \eta  <  1$. Squaring, then taking the expectation of \eqref{usable-inequ} 
and using $\mathbb{E}[A^2]  =  \mathcal{O}(n^{-1})$, $\tilde\epsilon\geq \epsilon$, and $(\sum_{i=1}^pv_i)^2 \le p\sum_{i=1}^pv_i^2$,
\begin{align}\label{intermediate}
 \frac{1}{9}\mathbb{E}\big[\big(\langle \mathcal{L}\rangle_{k,t;\epsilon} - \mathbb{E}[\langle \mathcal{L}\rangle_{k,t;\epsilon}]\big)^2\big]
 \leq &\, 
 \delta^{-2} \mathcal{O}(n^{-1 + \eta}) +3 \delta^{-2}M^2 (\tilde\epsilon +\delta)\mathcal{O}(n^{-1}) 
 \nonumber \\ &\quad
 + C^+(\tilde\epsilon)^2 + C^-(\tilde\epsilon)^2
 + \frac{M^2}{4\epsilon} \mathcal{O}(n^{-1})\,.
\end{align}
We now take $\epsilon\in[a_n, b_n]$ and $0<\delta <a_n$.
Using the change of variables $\epsilon\mapsto \tilde\epsilon(\epsilon)$, that the Jacobian $J = d\tilde\epsilon/d\epsilon \geq 1$, 
$|d\widetilde f(\tilde\epsilon)/d\tilde\epsilon|  \leq  (\mathbb{E}[S^2]  +  M/\sqrt{\tilde\epsilon})/2$ from \eqref{first-derivative-average} 
and \eqref{new-free}, $C^\pm(\tilde\epsilon)  \geq  0$ from \eqref{175}, and the mean value theorem
\begin{align}
 \int_{a_n}^{b_n} d\epsilon\, \big(C^+(&\tilde\epsilon)^2 + C^-(\tilde\epsilon)^2\big)
 = \int_{\tilde\epsilon(a_n)}^{\tilde\epsilon(b_n)} \frac{d\tilde\epsilon}{J} \, \big(C^+(\tilde\epsilon)^2 + C^-(\tilde\epsilon)^2\big)
 \nonumber \\ &
 \leq 
 \Big(\mathbb{E}[S^2] + \frac{M}{\sqrt{\tilde\epsilon(a_n)}}\Big)
 \int_{\tilde\epsilon(a_n)}^{\tilde\epsilon(b_n)} d\tilde\epsilon\, \big(C^+(\tilde\epsilon) + C^-(\tilde\epsilon)\big)
 \nonumber \\ &
=   
\Big(\mathbb{E}[S^2] + \frac{M}{\sqrt{\tilde\epsilon(a_n)}}\Big) \Big[\Big(\widetilde f(\tilde\epsilon(b_n)-\delta) - \widetilde f(\tilde\epsilon(b_n)+\delta)\Big) + \Big(\widetilde f(\tilde\epsilon(a_n)+\delta) - \widetilde f(\tilde\epsilon(a_n)-\delta)\Big)\Big]
\nonumber \\ &
 \leq 
 2\delta \Big(\mathbb{E}[S^2] + \frac{M}{\sqrt{\tilde\epsilon(a_n)-\delta}}\Big)^2
 \nonumber \\ &
 \leq 
 2\delta \Big(\mathbb{E}[S^2] + \frac{M}{\sqrt{a_n-\delta}}\Big)^2\,.
\end{align}
Thus, integrating \eqref{intermediate} over $\epsilon\in [a_n, b_n]$ yields with $0<\delta < a_n$
\begin{align*}
 \frac{1}{9}\int_{a_n}^{b_n} d\epsilon\, 
 \mathbb{E}\big[\big(\langle \mathcal{L}&\rangle_{k,t;\epsilon} - \mathbb{E}[\langle \mathcal{L}\rangle_{k,t;\epsilon}]\big)^2\big]\nn
 \leq \,&\delta^{-2} \mathcal{O}(n^{-1 + \eta}) +3 \delta^{-2}M^2 (B +\delta)\mathcal{O}(n^{-1}) + \frac{M^2}{4}\vert \ln a_n\vert\mathcal{O}(n^{-1}) 
 + 2\delta \Big(\mathbb{E}[S^2] + \frac{M}{\sqrt{a_n-\delta}}\Big)^2
\end{align*}
where $B\ge\tilde \epsilon$, because $\tilde \epsilon$ is bounded by assumption of the boundedness of the $m_k$'s and $\epsilon \le 1$.
Finally we choose $\delta  =  a_n n^{-\frac{1}{4}  +  \frac{\eta}{2}}$, $0 < \eta < 1/2$, and obtain for $n$ large enough 
(and $a$ fixed positive small)
\begin{align}\label{conc-L}
 \int_{a_n}^{b_n} d\epsilon\, 
 \mathbb{E}\big[\big( \langle \mathcal{L}\rangle_{k,t;\epsilon} - \mathbb{E}[\langle \mathcal{L}\rangle_{k,t;\epsilon}]\big)^2\big]
 \leq
 C a_n^{-2}n^{-\frac{1}{4} + \frac{\eta}{2}}
\end{align}
for some constant $C>0$ depending only on $M$ and $\mathbb{E}[S^2]$. 

It remains to justify \eqref{implies}. By the Cauchy-Schwarz inequality and Proposition \ref{conc-free} we have 
\begin{align}
 \mathbb{E}\bigl[ ( F_{k, t; \epsilon}(\boldsymbol{\Theta}) - f_{k, t; \epsilon})^2 \bigr]  &= 
 \mathbb{E}\bigl[ ( F_{k, t; \epsilon}(\boldsymbol{\Theta}) - f_{k, t; \epsilon})^2 \mathds{1}(\vert F_{k, t; \epsilon}(\boldsymbol{\Theta}) - f_{k, t; \epsilon}\vert \leq u) \bigr] 
 \nonumber \\ & \qquad+ 
 \mathbb{E}\bigl[ ( F_{k, t; \epsilon}(\boldsymbol{\Theta}) - f_{k, t; \epsilon})^2 \mathds{1}(\vert F_{k, t; \epsilon}(\boldsymbol{\Theta}) - f_{k, t; \epsilon}\vert > u) \bigr] 
 \nonumber \\ &
 \leq 
 u^2 + \sqrt{\mathbb{E}\big[( F_{k, t; \epsilon}(\boldsymbol{\Theta}) - f_{k, t; \epsilon})^4\big]} \sqrt{\mathbb{E}\big[\mathds{1}(\vert F_{k, t; \epsilon}(\boldsymbol{\Theta}) - f_{k, t; \epsilon}\vert > u) \big] }
 \nonumber \\ &
 \leq 
 u^2 + \sqrt{\mathbb{E}\big[ ( F_{k, t; \epsilon}(\boldsymbol{\Theta}) - f_{k, t; \epsilon})^4\big]} e^{-cn u^2/2}.
\end{align}
If we can show that the moments of the (random) free energy $F_{k, t; \epsilon}(\boldsymbol{\Theta})$ are bounded uniformly in 
$n$, then the choice $u = n^{-1/2 +\eta}$ for any $0 < \eta  <  1/2$ allows to conclude the proof.
Let us briefly show how the moments are estimated. By the Jensen's inequality
\begin{align}
F_{k, t;\epsilon}(\boldsymbol{\theta}) \leq \frac{1}{n} \int \big\{\prod_{i=1}^n dx_i P_0(x_i)\big\} \mathcal{H}_{k, t;\epsilon}(\bx;\boldsymbol{\theta}).
\end{align}
The expectation over $\bX$ is computed from \eqref{intH} and one finds a polynomial in $\{s_i, \{z_{ij}^{(k)}, \widetilde z_i^{(k)}\}_{k=1}^K, \widehat z_i\}_{i=1}^n$ which all have bounded moments. On the other hand from \eqref{h}, \eqref{hmf} by completing the squares we have
\begin{align}
h(\bx, \bs,\bz,\sigma^2) & \geq -\frac{1}{2\sigma^2}\sum_{i\leq j =1}^n\Big(\frac{s_is_j}{\sqrt n} + z_{ij}\sigma\Big)^2,
\\
h_{\rm mf}(\bx, \bs,\widetilde{\bz},\sigma^2) & \geq -\frac{1}{2\sigma^2}\sum_{i=1}^n(s_i + \widetilde z_{i})^2,
\end{align} 
and find that $\mathcal{H}_{k,t;\epsilon}(\bx;\boldsymbol{\theta})$ is lower bounded by a polynomial in $\{s_i, \{z_{ij}^{(k)}, \widetilde z_i^{(k)}\}_{k=1}^K, \widehat z_i\}_{i=1}^n$. This is also the case for $F_{k, t;\epsilon}(\boldsymbol{\theta})$. With these upper and lower bounds on $F_{k, t;\epsilon}(\boldsymbol{\theta})$ it is easy to show that for any integer $p$
\begin{align}
\mathbb{E}[\vert F_{k, t;\epsilon}(\boldsymbol{\theta})\vert^p] \leq C_p
\end{align}
where $C_p$ is independent of $n$ and depends only on $\Delta$ and moments of $P_0$. \QEDA

\section{A fluctuation identity}\label{fluctuation-identity}
The purpose of this appendix is to prove the identity \eqref{82} relating the various fluctuations. This identity is quite powerful
and holds in quite some generality and in particular for the three applications presented in this paper. To alleviate the notation we denote 
$\langle-\rangle_{k,t;\epsilon}$ simply by $\langle - \rangle$. It actually follows from the exact formula
\begin{align}
\mathbb{E}\big[ \big\langle (\mathcal{L} - \mathbb{E}[\langle \mathcal{L}\rangle])^2\big\rangle\big]
= \, &\frac{1}{4n^2}\sum_{i, j=1}^n \big\{\mathbb{E}[\langle X_i X_j\rangle^2] - \mathbb{E}[\langle X_i\rangle^2] \mathbb{E}[\langle X_j\rangle^2]\big\}
\nonumber \\ 
&\quad+ \frac{1}{2n^2}\sum_{i, j=1}^n \big\{\mathbb{E}[\langle X_i X_j\rangle^2] - \mathbb{E}[\langle X_i X_j\rangle \langle X_i\rangle
\langle X_j\rangle]\big\} +\frac{1}{4n^2 \tilde\epsilon} \sum_{i=1}^n \mathbb{E}[\langle X_i^2\rangle]
\label{exact}
\end{align}
that we derive next. But before doing so, let us show how \eqref{exact} implies \eqref{82}. First we note that by \eqref{nishibasic} 
the last subdominant sum equals $\mathbb{E}[S^2]/4 n\tilde\epsilon={\cal O}(1/n)$. We then express the first two terms in terms of the overlap $q_{\bx,\bs}$. From \eqref{nishibasic} we have $\mathbb{E}[\langle X_i X_j\rangle^2]  =  \mathbb{E}[S_iS_j\langle X_i X_j\rangle]$ and therefore
\begin{align}\label{nish1}
\frac{1}{n^2}\sum_{i, j=1}^n \mathbb{E}[\langle X_i X_j\rangle^2] = \frac{1}{n^2}\sum_{i, j=1}^n\mathbb{E}[S_i S_j\langle X_i X_j\rangle] 
= \mathbb{E}[\langle q_{\bX,\bS}^2\rangle].
\end{align}
Similarly $\mathbb{E}[\langle X_i\rangle^2] =  \mathbb{E}[S_i \langle X_i \rangle]$, so 
\begin{align}\label{nish2}
\frac{1}{n^2}\sum_{i, j=1}^n \mathbb{E}[\langle X_i\rangle^2] \mathbb{E}[\langle X_j\rangle^2] 
= \mathbb{E}[\langle q_{\bX,\bS}^{}\rangle]^2,
\end{align}
and $\mathbb{E}[\langle X_i X_j\rangle \langle X_i\rangle
\langle X_j\rangle] = \mathbb{E}[S_i S_j \langle X_i\rangle
\langle X_j\rangle]$ which implies
\begin{align}\label{nish3}
\frac{1}{n^2}\sum_{i,j=1}^n\mathbb{E}[\langle X_i X_j\rangle \langle X_i\rangle
\langle X_j\rangle] = \mathbb{E}[\langle q_{\bX,\bS}^{}\rangle^2].
\end{align}
Replacing the three last identities in \eqref{exact} leads to \eqref{82}.

We now summarise the main steps leading to the formula \eqref{exact}, using the identity \eqref{nishibasic} and integrations by parts w.r.t the Gaussian noise. 
This formula follows by summing the two following identities
\begin{align}
\mathbb{E}[\langle \mathcal{L}^2 \rangle]  -  \mathbb{E}[\langle \mathcal{L} \rangle^2]
 &
  =   \frac{1}{2n^2} \sum_{i, j =1}^n  \big\{\mathbb{E}[\langle X_i X_j\rangle^2]  - 2\mathbb{E}[\langle X_iX_j\rangle \langle X_i\rangle\langle X_j\rangle]  
 + \mathbb{E}[\langle X_i\rangle^2\langle X_j\rangle^2]\big\}
 \nonumber \\ &
  \qquad+ 
  \frac{1}{4n^2\tilde\epsilon}  \sum_{i=1}^n \mathbb{E}[\langle X_i^2\rangle  -  \langle X_i\rangle^2],
  \label{thermal}
\\
\mathbb{E}[\langle \mathcal{L} \rangle^2]  -  \mathbb{E}[\langle \mathcal{L} \rangle]^2
 & 
  =  \frac{1}{4n^2} \sum_{i,j=1}^n \big\{ \mathbb{E}[\langle X_i X_j\rangle^2]  -  \mathbb{E}[\langle X_i\rangle^2] \mathbb{E}[\langle X_j\rangle^2]\big\}
  \nonumber \\ &
  \qquad+ 
 \frac{1}{2n^2} \sum_{i,j=1}^n \big\{ \mathbb{E}[\langle X_i\rangle \langle X_j\rangle \langle X_i X_j\rangle]
 -  \mathbb{E}[\langle X_i\rangle^2\langle X_j\rangle^2]\big\} + \frac{1}{4n^2\tilde\epsilon}\sum_{i=1}^n \mathbb{E}[\langle X_i\rangle^2].
\label{disorder} 
\end{align}
We first derive the second identity which requires somewhat longer calculations.

\subsubsection*{Derivation of \eqref{disorder}}

First we compute $\mathbb{E}[\langle \mathcal{L}\rangle]^2$.
From \eqref{L-def} we have
\begin{align}
\mathbb{E}[\langle \mathcal{L}\rangle]
= 
\frac{1}{n}\sum_{i=1}^n \Big\{ \frac{1}{2}\mathbb{E}[\langle X_i^2\rangle] - \mathbb{E}[\langle X_i\rangle S_i] -
\frac{1}{2\sqrt{\tilde\epsilon}} \mathbb{E}[\langle X_i\rangle\widehat Z_i]\Big\}.
\end{align}
From \eqref{nishibasic} we have $\mathbb{E}[\langle X_i\rangle S_i] = \mathbb{E}[\langle X_i\rangle^2]$ and by an integration by parts
\begin{align}\label{first-integration-by-parts}
\frac{1}{\sqrt{\tilde\epsilon}}\mathbb{E}[\langle X_i\rangle\widehat Z_i] = \frac{1}{\sqrt{{\tilde \epsilon}}}\mathbb{E}\Big[\frac{\partial}{\partial\widehat Z_i} \langle X_i\rangle\Big] =
\mathbb{E}[\langle X_i^2\rangle - \langle X_i\rangle^2] .
\end{align}
Thus we find
\begin{align}
 \mathbb{E}[\langle \mathcal{L}\rangle] = - \frac{1}{2n} \sum_{i=1}^n\mathbb{E}[\langle X_i\rangle^2],\qquad 
 \label{bothsides}
\end{align}
which is formula \eqref{first-derivative-average}. Squaring, we have
\begin{align}\label{easy}
 \mathbb{E}[\langle \mathcal{L}\rangle]^2 = \frac{1}{4n^2} \sum_{i, j=1}^n\mathbb{E}[\langle X_i\rangle^2]\mathbb{E}[\langle X_j\rangle^2]\,.
\end{align}
Now we compute $\mathbb{E}[\langle \mathcal{L}\rangle^2]$.
From \eqref{L-def} we have 
\begin{align}
\langle \mathcal{L}\rangle^2 = 
\frac{1}{n^2}\sum_{i,j=1}^n &\Big\{
\frac{1}{4}\langle X_i^2\rangle \langle X_j^2\rangle -\frac{1}{2}\langle X_i^2\rangle \langle X_j\rangle s_j - \frac{1}{4\sqrt{\tilde\epsilon}} \langle X_i^2\rangle \langle X_j\rangle \widehat z_j
\nonumber \\ &
- \frac{1}{2}\langle X_i\rangle s_i \langle X_j^2\rangle + \langle X_i\rangle s_i \langle X_j\rangle s_j + \frac{1}{2\sqrt{\tilde\epsilon}}
\langle X_i\rangle s_i \langle X_j\rangle \widehat z_j
\nonumber \\ &
- \frac{1}{4\sqrt{\tilde\epsilon}} \langle X_i\rangle \langle X_j^2\rangle \widehat z_i + \frac{1}{2\sqrt{\tilde\epsilon}} \langle X_i\rangle \langle X_j\rangle s_j \widehat z_i + \frac{1}{4\tilde \epsilon} \langle X_i\rangle \langle X_j\rangle \widehat z_i\widehat z_j
\Big\}.
\label{start}
\end{align}
Taking the expectation and using \eqref{nishibasic} (for the terms that do {\it not} contain explicit $z$-factors) we find
\begin{align}
\mathbb{E}[\langle \mathcal{L}\rangle^2] =
\frac{1}{n^2}\sum_{i,j=1}^n &\Big\{\frac{1}{4}\mathbb{E}[\langle X_i^2\rangle \langle X_j^2\rangle] -\frac{1}{2}\mathbb{E}[\langle X_i^2\rangle \langle X_j\rangle^2] - \frac{1}{4\sqrt{\tilde \epsilon}} \mathbb{E}[\langle X_i^2\rangle \langle X_j\rangle \widehat Z_j]
\nonumber \\ &
- \frac{1}{2}\mathbb{E}[\langle X_i\rangle^2 \langle X_j^2\rangle] + \mathbb{E}[\langle X_i X_j\rangle \langle X_i\rangle\langle X_j\rangle] + \frac{1}{2\sqrt{\tilde \epsilon}}
\mathbb{E}[\langle X_i\rangle S_i \langle X_j\rangle \widehat Z_j]
\nonumber \\ &
- \frac{1}{4\sqrt{\tilde \epsilon}} \mathbb{E}[\langle X_i\rangle \langle X_j^2\rangle \widehat Z_i] + \frac{1}{2\sqrt{\tilde \epsilon}} \mathbb{E}[\langle X_i\rangle \langle X_j\rangle S_j \widehat Z_i] + 
\frac{1}{4{\tilde \epsilon}} \mathbb{E}[\langle X_i\rangle \langle X_j\rangle \widehat Z_i\widehat Z_j]
\Big\}.
\label{startbis}
\end{align}
In order to simplify this expression we now integrate by parts all terms that contain explicit $Z$-factors:
\begin{align}
\frac{1}{4{\tilde \epsilon}} \mathbb{E}[\langle X_i\rangle \langle X_j\rangle \widehat Z_i\widehat Z_j]
 & =  \frac{1}{4{\tilde \epsilon}} \mathbb{E}\Big[\frac{\partial}{\partial \widehat Z_j}(\langle X_i\rangle \langle X_j\rangle \widehat Z_i)\Big]
\nonumber \\ &
= \frac{1}{4\sqrt {\tilde \epsilon}} \Big(\mathbb{E}[\langle X_i^2\rangle \langle X_j\rangle \widehat Z_j] - 
2 \mathbb{E}[\langle X_i\rangle^2 \langle X_j\rangle \widehat Z_j]  +  \mathbb{E}[\langle X_i\rangle \langle X_jX_i\rangle \widehat Z_j]\Big)
+ \frac{1}{4{\tilde \epsilon}} \mathbb{E}[\langle X_i\rangle^2] \delta_{ij}
\nonumber \\ &
=
\frac{1}{4}\mathbb{E}[\langle X_i^2 X_j\rangle\langle X_j\rangle] - \frac{1}{4}\mathbb{E}[\langle X_i^2\rangle\langle X_j\rangle^2]
+\frac{1}{4}\mathbb{E}[\langle X_i^2\rangle\langle X_j^2\rangle] - \frac{1}{4}\mathbb{E}[\langle X_i^2\rangle\langle X_j\rangle^2]
\nonumber \\ &
\qquad-\mathbb{E}[\langle X_i\rangle\langle X_j\rangle \langle X_iX_j\rangle] + \mathbb{E}[\langle X_i\rangle^2\langle X_j\rangle^2] -\frac{1}{2} \mathbb{E}\langle X_i\rangle^2 \langle X_j^2\rangle] 
+\frac{1}{2} \mathbb{E}[\langle X_i\rangle^2 \langle X_j\rangle^2] 
\nonumber \\ &
\qquad+
\frac{1}{4}\mathbb{E}[\langle X_i X_j\rangle^2] - \frac{1}{4} \mathbb{E}[\langle X_i\rangle \langle X_j\rangle \langle X_iX_j\rangle] +\frac{1}{4} \mathbb{E}[\langle X_i\rangle \langle X_i X_j^2\rangle] 
\nonumber \\ &
\qquad- \frac{1}{4}\mathbb{E}[\langle X_i\rangle \langle X_j\rangle \langle X_i X_j\rangle]+ \frac{1}{4{\tilde \epsilon}} \mathbb{E}[\langle X_i\rangle^2] \delta_{ij} 
\nonumber \\ &
= 
\frac{1}{4}\mathbb{E}[\langle X_i X_j\rangle^2] +
\frac{1}{4}\mathbb{E}[\langle X_i^2 X_j\rangle\langle X_j\rangle] +\frac{1}{4} \mathbb{E}[\langle X_i\rangle \langle X_i X_j^2\rangle]   
+\frac{1}{4}\mathbb{E}[\langle X_i^2\rangle\langle X_j^2\rangle] 
\nonumber \\ &
\qquad+
\frac{3}{2}\mathbb{E}[\langle X_i\rangle^2\langle X_j\rangle^2] -\frac{1}{2} \mathbb{E}\langle X_i\rangle^2 \langle X_j^2\rangle] - \frac{1}{2}\mathbb{E}[\langle X_i^2\rangle\langle X_j\rangle^2]
- \frac{3}{2} \mathbb{E}[\langle X_i\rangle \langle X_j\rangle \langle X_iX_j\rangle]  
\nonumber \\ &
\qquad+ \frac{1}{4{\tilde \epsilon}} \mathbb{E}[\langle X_i\rangle^2] \delta_{ij}.
\label{lengthy}
\end{align}
Then
\begin{align}
-\frac{1}{4\sqrt{\tilde \epsilon}}\mathbb{E}[\langle X_i\rangle \langle X_j^2\rangle \widehat Z_i]
& = - \frac{1}{4}\mathbb{E}[\langle X_i^2\rangle \langle X_j^2\rangle] 
+ \frac{1}{2}\mathbb{E}[\langle X_i\rangle^2 \langle X_j^2\rangle] 
- \frac{1}{4}\mathbb{E}[\langle X_i\rangle \langle X_j^2 X_i\rangle]
,\label{11}\\
-\frac{1}{4\sqrt{\tilde \epsilon}}\mathbb{E}[\langle X_i^2\rangle \langle X_j\rangle \widehat Z_j]
& = - \frac{1}{4}\mathbb{E}[\langle X_i^2\rangle \langle X_j^2\rangle] 
+ \frac{1}{2}\mathbb{E}[\langle X_i^2\rangle \langle X_j\rangle^2] 
- \frac{1}{4}\mathbb{E}[\langle X_j\rangle \langle X_i^2 X_j\rangle]
,
\label{11bis}\\
 \frac{1}{2\sqrt {\tilde \epsilon}} \mathbb{E}[\langle X_i\rangle \langle X_j\rangle S_i \widehat Z_j] & = 
 \frac{1}{2}\mathbb{E}[ S_i \langle X_i X_j\rangle \langle X_j\rangle] - \mathbb{E}[S_i\langle X_i\rangle \langle X_j\rangle^2] +\frac{1}{2}\mathbb{E}[S_i\langle X_i\rangle \langle X_j^2\rangle]
 \nonumber \\ &
 =
 \frac{1}{2}\mathbb{E}[\langle X_i X_j\rangle \langle X_i\rangle \langle X_j\rangle] - \mathbb{E}[\langle X_i\rangle^2 \langle X_j\rangle^2] 
 + \frac{1}{2}\mathbb{E}[\langle X_i\rangle^2 \langle X_j^2\rangle],
\end{align}
and finally
\begin{align}
 \frac{1}{2\sqrt {\tilde \epsilon}} \mathbb{E}[\langle X_i\rangle \langle X_j\rangle S_j \widehat Z_i] & = 
 \frac{1}{2}\mathbb{E}[ S_j \langle X_i X_j\rangle \langle X_i\rangle] - \mathbb{E}[S_j\langle X_j\rangle \langle X_i\rangle^2] +\frac{1}{2}\mathbb{E}[S_j\langle X_j\rangle \langle X_i^2\rangle]
 \nonumber \\ &
 =
 \frac{1}{2}\mathbb{E}[\langle X_i X_j\rangle \langle X_i\rangle \langle X_j\rangle] - \mathbb{E}[\langle X_i\rangle^2 \langle X_j\rangle^2] 
 + \frac{1}{2}\mathbb{E}[\langle X_j\rangle^2 \langle X_i^2\rangle]
 \label{stillone}
\end{align}
where in the last two identities we used \eqref{nishibasic} {\it after} the integration by parts.
Replacing the last five identities \eqref{lengthy}--\eqref{stillone} into \eqref{startbis} we get 
\begin{align}
\mathbb{E}[\langle \mathcal{L}\rangle^2] &= \frac{1}{n^2}\sum_{i,j=1}^n
\Big\{ \frac{1}{4} \mathbb{E}[\langle X_i X_j\rangle^2] +\frac{1}{2} \mathbb{E}[\langle X_i\rangle \langle X_j\rangle \langle X_i X_j\rangle]  
- \frac{1}{2} \mathbb{E}[\langle X_i\rangle^2 \langle X_j\rangle^2] \Big\}\nn
&\qquad+ \frac{1}{4{\tilde \epsilon} n^2} \sum_{i=1}^n \mathbb{E}[\langle X_i\rangle^2].
\label{atlast}
\end{align}
Subtracting \eqref{atlast} and \eqref{easy} we finally find \eqref{disorder}. \QEDA

\subsubsection*{Derivation of \eqref{thermal}}

Acting with $n^{-1}d/d{\tilde \epsilon}$ on both sides
of \eqref{bothsides} we find
\begin{align}
 -\mathbb{E}[\langle \mathcal{L}^2\rangle - \langle \mathcal{L}\rangle^2] 
 + \frac{1}{n}\mathbb{E}\Big[\Big\langle \frac{d\mathcal{L}}{d{\tilde \epsilon}} \Big\rangle \Big]
 =
   \frac{1}{n} \sum_{i=1}^n\mathbb{E}[\langle X_i\rangle (\langle X_i\mathcal{L}\rangle - \langle X_i\rangle \langle \mathcal{L}\rangle)].
  \label{intermediate_}
\end{align}
Computing the derivative of $\mathcal{L}$ and using \eqref{first-integration-by-parts} we find that \eqref{intermediate_} is equivalent to
\begin{align}\label{55}
\mathbb{E}[\langle \mathcal{L}^2\rangle - \langle \mathcal{L}\rangle^2] 
=
-\frac{1}{n} \sum_{i=1}^n\mathbb{E}[\langle X_i\rangle (\langle X_i\mathcal{L}\rangle - \langle X_i\rangle \langle \mathcal{L}\rangle)] + \frac{1}{4n^2{\tilde \epsilon}}\sum_{i=1}^n \mathbb{E}[\langle X_i^2\rangle - \langle X_i\rangle^2].
\end{align}
Now we compute the terms in the first sum. We have 
\begin{align}
\langle X_i\rangle (\langle X_i\mathcal{L}\rangle - \langle X_i\rangle \langle \mathcal{L}\rangle)
&= 
\frac{1}{n}\sum_{j=1}^n \Big\{ \frac{1}{2}\langle X_i\rangle \langle X_i X_j^2\rangle - \langle X_i\rangle \langle X_iX_j\rangle s_j - \frac{1}{2\sqrt{\tilde \epsilon}}\langle X_i\rangle \langle X_i X_j\rangle \widehat z_j 
\nonumber \\ &
\qquad- \frac{1}{2}\langle X_i\rangle^2 \langle X_j^2\rangle   + \langle X_i\rangle^2 \langle X_j\rangle s_j + 
\frac{1}{2\sqrt{\tilde \epsilon}} \langle X_i\rangle ^2 \langle X_j\rangle \widehat z_j\Big\}.
\end{align}
Then from \eqref{nishibasic},
\begin{align}
\mathbb{E}[\langle X_i\rangle (\langle X_i\mathcal{L}\rangle - \langle X_i\rangle \langle \mathcal{L}\rangle)]
&= 
\frac{1}{n}\sum_{j=1}^n \Big\{ \frac{1}{2}\mathbb{E}[\langle X_i\rangle \langle X_i X_j^2\rangle] - \mathbb{E}[\langle X_i\rangle \langle X_j\rangle\langle X_iX_j\rangle]  - \frac{1}{2\sqrt{\tilde \epsilon}}\mathbb{E}[\langle X_i\rangle \langle X_i X_j\rangle \widehat Z_j] 
\nonumber \\ &\qquad
- \frac{1}{2}\mathbb{E}[\langle X_i\rangle^2 \langle X_j^2\rangle]   + \mathbb{E}[\langle X_i\rangle^2 \langle X_j\rangle^2] + 
\frac{1}{2\sqrt{\tilde \epsilon}} \mathbb{E}[\langle X_i\rangle ^2 \langle X_j\rangle \widehat Z_j]\Big\}.
\end{align}
It remains to integrate by parts the two terms involving the explicit $\widehat Z_j$ dependence (these can be found in the previous integrations by parts).
This leads to
\begin{align}\label{44}
&\frac{1}{n}\sum_{i=1}^n\mathbb{E}[\langle X_i\rangle (\langle X_i\mathcal{L}\rangle - \langle X_i\rangle \langle \mathcal{L}\rangle)]\nn
&\qquad= 
-\frac{1}{2n^2}\sum_{i,j=1}^n \big\{ \mathbb{E}[\langle X_i X_j\rangle^2] -2\mathbb{E}[\langle X_iX_j\rangle \langle X_i\rangle\langle X_j\rangle] +\mathbb{E}[\langle X_i\rangle^2\langle X_j\rangle^2] \big\}.
\end{align}
The formula \eqref{thermal} then follows from \eqref{55} and \eqref{44}.
\section{Concentration of the free energy}\label{concentration-free-energy}

In this section we prove Proposition \ref{conc-free}. We will call $\mathbb{E}_{\bZ}$, $\mathbb{P}_{\bZ}$ the expectation and probability law over all Gaussian variables, $\mathbb{E}_{\bS}$, $\mathbb{P}_{\bS}$ the ones over the input signal variables,
and $\mathbb{E}$, $\mathbb{P}$ the ones over the joint law.
The proof is broken up in two lemmas.
We first show a lemma which expresses concentration w.r.t 
all Gaussian sources of disorder uniformly in the input signal. 
\begin{lemma}[Concentration w.r.t the Gaussian quenched disorder]\label{conc-free-lemma-1}
Take $P_0$ with bounded support in $[-M, M]$. For any signal realisation $\bs$ and all $k = 1, \ldots, K$, $t \in  [0,1]$ and $\epsilon   >   0$ we have
\begin{align}
 \mathbb{P}_{\bZ}\big[ \vert F_{k, t; \epsilon}(\boldsymbol{\Theta}) - \mathbb{E}_{\bZ}[F_{k, t; \epsilon}(\boldsymbol{\Theta})] \vert > u/2  \big] 
 \leq 
 2 \exp\Big(-\frac{n u^2}{16 (\frac{ 2M^4}{\Delta} + \frac{\epsilon M^2}{2})}\Big), \label{114}
\end{align}
where $u >  0$. 
\end{lemma}

\begin{proof}
The proof method is again based on an interpolation (of a different kind) that goes back to a beautiful work of 
Guerra and Toninelli \cite{generalised-mean-field}. We fix the input signal realisation $\bs$ and consider two i.i.d copies for 
the Gaussian quenched variables 
$\bz^{(k, 1)} = [z_{ij}^{(k, 1)}]_{i,j=1}^n, \widetilde{\bz}^{(k, 1)} = [\widetilde{z_{i}}^{(k, 1)}]_{i=1}^n$ and 
$\bz^{(k, 2)} = [z_{ij}^{(k, 2)}]_{i,j=1}^n, \widetilde{\bz}^{(k, 2)} = [\widetilde{z_{i}}^{(k, 2)}]_{i=1}^n$. We also need two copies of 
the extra Gaussian noise introduced in the perturbation term \eqref{perturb}, 
namely $\widehat{\bz}^{(1)}  =  [\widehat z_i^{(1)}]_{i=1}^n$ and $\widehat{\bz}^{(2)}  =  [\widehat z_i^{(2)}]_{i=1}^n$. We define an Hamiltonian interpolating 
between the two realizations of the
Gaussian disorder, with new interpolating parameter $\tau \in  [0,1]$:
\begin{align*}
{\cal H}_{k,t, \tau;\epsilon} \defeq 
&
\sum_{k'>k} h\Big(\bx, \bs,\sqrt{\tau}\,\bz^{(k', 1)} + \sqrt{1-\tau}\,\bz^{(k', 2)},K\Delta\Big)
+\sum_{k'<k} h_{\rm mf}\Big(\bx, \bs,\sqrt{\tau}\,\widetilde{\bz}^{(k', 1)} + \sqrt{1-\tau}\,\widetilde{\bz}^{(k', 2)} ,\frac{K\Delta}{m_{k'}}\Big)
\\ \nonumber &
\quad+ h\Big(\bx, \bs,\sqrt{\tau}\,\bz^{(k, 1)} + \sqrt{1-\tau}\,\bz^{(k, 2)},\frac{K\Delta}{1-t}\Big)
+h_{\rm mf}\Big(\bx, \bs,\sqrt{\tau}\,\widetilde{\bz}^{(k, 1)} + \sqrt{1-\tau}\,\widetilde{\bz}^{(k, 2)},\frac{K\Delta}{t\,m_{k}}\Big)\nonumber\\
&\quad+ \epsilon\sum_{i=1}^n\Big(\frac{x_i^2}{2} -x_is_i -\frac{1}{\sqrt\epsilon} x_i(\sqrt{\tau}\,\widehat z_i^{(1)} + \sqrt{1-\tau} \widehat z_i^{(2)})\Big).
\end{align*}
Let ${\cal Z}_{k, t;\epsilon}(\tau) \defeq \int\{\prod_{i=1}^n dx_i P_0(x_i)\}\exp(-{\cal H}_{k,t, \tau;\epsilon})$ the partition function associated to ${\cal H}_{k,t, \tau;\epsilon}$.
Let $s > 0$ be a trial parameter to be fixed later on and let
\begin{align}
 \varphi_{k, t;\epsilon}(\tau) \defeq \ln \mathbb{E}_1\Big[\exp \big(s\, \mathbb{E}_2[\ln {\cal Z}_{k, t;\epsilon}(\tau)]\big) \Big],
\end{align}
where $\mathbb{E}_{1}$ and $\mathbb{E}_2$ are the expectations w.r.t the two independent sets of Gaussian variables (note that $\varphi_{k,t;\epsilon}(\tau)$ depends 
on the fixed signal instance $\bs$). Using the union bound for the first inequality and Markov's inequality together with $\exp( \varphi_{k, t;\epsilon}(1)) = \EE_\bZ [\exp(-sn F_{k, t; \epsilon}(\boldsymbol{\Theta}))]$ and $\exp( \varphi_{k, t;\epsilon}(0)) =  \exp(-sn \EE_\bZ [F_{k, t; \epsilon}(\boldsymbol{\Theta})])$ for the second one, one deduces that 
\begin{align}
 \mathbb{P}_{\bZ}\big[ \vert F_{k, t; \epsilon}(\boldsymbol{\Theta}) - \mathbb{E}_{\bZ}[F_{k, t; \epsilon}&(\boldsymbol{\Theta})] \vert > u/2  \big] \nn
  \leq \,&\mathbb{P}_{\bZ}\big[ e^{ns(F_{k, t; \epsilon}(\boldsymbol{\Theta}) - \mathbb{E}_{\bZ}[F_{k, t; \epsilon}(\boldsymbol{\Theta})]-u/2)} > 1  \big]+\mathbb{P}_{\bZ}\big[ e^{ns( \mathbb{E}_{\bZ}[F_{k, t; \epsilon}(\boldsymbol{\Theta})]-F_{k, t; \epsilon}(\boldsymbol{\Theta}) -u/2)} >1  \big]\nonumber\\
 \leq \,&\exp\Big(\varphi_{k, t;\epsilon}(0) - \varphi_{k, t;\epsilon}(1) - snu/2\Big)+\exp\Big(\varphi_{k, t;\epsilon}(1) - \varphi_{k, t;\epsilon}(0) - snu/2\Big)
 \nonumber \\ 
 \leq \,&
 2 \exp\Big(\vert \varphi_{k, t;\epsilon}(1) - \varphi_{k, t;\epsilon}(0)\vert  - snu/2\Big) 
 \nonumber \\ 
 \leq\,&
 2 \exp\Big(\int_0^1d\tau\, \vert \varphi_{k, t;\epsilon}^{\prime}(\tau)\vert   - snu/2\Big).
 \label{markov}
\end{align}
Our essential task is now to prove an upper bound on $\vert \varphi_{k, t;\epsilon}^\prime(\tau)\vert$. We have
\begin{align}\label{derivative-varphi}
 \varphi_{k, t;\epsilon}^\prime(\tau) = \frac{\mathbb{E}_1 \big[ s\, \mathbb{E}_2 \big[\frac{{\cal Z}_{k, t;\epsilon}^\prime(\tau)}{{\cal Z}_{k, t;\epsilon}(\tau)}\big]
 \exp( s\, \mathbb{E}_2[\ln {\cal Z}_{k, t;\epsilon}(\tau)])\big]}{\mathbb{E}_1\big[\exp(s \,\mathbb{E}_2[\ln Z_{k, t; \epsilon}(\tau)])\big]}
\end{align}
where
\begin{align*}
 &\mathbb{E}_2 \Big[\frac{{\cal Z}_{k, t;\epsilon}^\prime(\tau)}{{\cal Z}_{k, t;\epsilon}(\tau)}\Big]
  = \frac{1}{2 \sqrt{K \Delta \tau n}} \sum_{k^\prime > k}\sum_{i\leq j} z_{ij}^{(k^\prime, 1)} \mathbb{E}_2[\langle X_i X_j\rangle_{k, t; \epsilon}]
 \nn
 &\qquad- \frac{1}{2 \sqrt{K \Delta (1-\tau) n}} \sum_{k^\prime > k}\sum_{i\leq j} \mathbb{E}_2[Z_{ij}^{(k^\prime, 2)} \langle X_i X_j\rangle_{k, t; \epsilon}]
 + \frac{1}{2\sqrt{K\Delta  \tau}}\sum_{k^\prime < k}\sqrt{m_{k^\prime}} \sum_i \widetilde z_i^{(k^\prime, 1)}\mathbb{E}_2[\langle X_i\rangle_{k,t;\epsilon}]
 \nonumber \\ &
 \qquad-  \frac{1}{2\sqrt{K\Delta (1-\tau)}}\sum_{k^\prime < k} \sqrt{m_{k^\prime}} \sum_i \mathbb{E}_2[\widetilde Z_i^{(k^\prime, 2)}\langle X_i\rangle_{k,t;\epsilon}]
 +   \frac{\sqrt{1-t}}{2 \sqrt{K \Delta \tau n}} \sum_{i\leq j} z_{ij}^{(k, 1)} \mathbb{E}_2[\langle X_i X_j\rangle_{k, t; \epsilon}] \nonumber \\ &
 \qquad -  \frac{\sqrt{1-t}}{2 \sqrt{K \Delta (1-\tau) n}} \sum_{i\leq j} \mathbb{E}_2[ Z_{ij}^{(k, 2)} \langle X_i X_j\rangle_{k, t; \epsilon}]
 +  \frac{\sqrt{t\, m_{k}}}{2\sqrt{K\Delta  \tau}}\sum_i\widetilde z_i^{(k, 1)}\mathbb{E}_2[\langle X_i\rangle_{k,t;\epsilon}] \nonumber \\ &
 \qquad-   \frac{\sqrt{t\,m_{k}}}{2\sqrt{K\Delta  (1-\tau)}}\sum_i\mathbb{E}_2[ \widetilde Z_i^{(k, 2)}\langle X_i\rangle_{k,t;\epsilon}]
 +  \frac{\sqrt{\epsilon}}{2\sqrt{\tau}}\sum_i \widehat z_i^{(1)}\mathbb{E}_2[\langle X_i\rangle_{k,t;\epsilon}]\nn
 &\qquad-   \frac{\sqrt\epsilon}{2\sqrt{1-\tau}}\sum_i \mathbb{E}_2[ \widehat Z_i^{(2)}\langle X_i\rangle_{k,t;\epsilon}] \,.
\end{align*}
We then replace this expression in the numerator of \eqref{derivative-varphi} and integrate by parts over all standard Gaussian 
variables of type $z^{(1)}$ and $z^{(2)}$. Doing so generates partial derivatives of the form 
$\mathbb{E}_2[\frac{\partial}{\partial z^{(1)}}\langle - \rangle]$ and $\mathbb{E}_2[\frac{\partial}{\partial z^{(2)}}\langle - \rangle]$ 
as well 
as derivatives of the form $\frac{\partial}{\partial z^{(1)}} \exp( s\, \mathbb{E}_2[\ln {\cal Z}_{k, t;\epsilon}(\tau)])$. A lengthy but straightforward calculation shows that only the later survive. The numerator of \eqref{derivative-varphi} becomes
\begin{align*}
& \hspace{0.05cm}\mathbb{E}_1\Big[\frac{s}{2\sqrt{K\Delta  \tau n}}\sum_{k^\prime > k}\sum_{i\leq j} \mathbb{E}_2[\langle X_i X_j\rangle_{k, t;\epsilon}]
\frac{\partial}{\partial Z_{ij}^{(k^\prime, 1)}} \exp(s\,\mathbb{E}_2[\ln {\cal Z}_{k, t;\epsilon}(\tau)])\Big]
\nonumber  \\  +  &
\, \mathbb{E}_1\Big[\frac{s}{2\sqrt{K\Delta  \tau}}\sum_{k^\prime < k}\sqrt{m_{k^\prime}} \sum_{i}  \mathbb{E}_2[\langle X_i\rangle_{k, t;\epsilon}]
\frac{\partial}{\partial \widetilde Z_{i}^{(k^\prime, 1)}} \exp(s\,\mathbb{E}_2[\ln {\cal Z}_{k, t;\epsilon}(\tau)])\Big]
\nonumber  \\  +   &
\, \mathbb{E}_1\Big[\frac{s\sqrt{1-t}}{2\sqrt{K\Delta  \tau n}}\sum_{i\leq j} \mathbb{E}_2[\langle X_i X_j\rangle_{k, t;\epsilon}]
\frac{\partial}{\partial Z_{ij}^{(k, 1)}} \exp(s\,\mathbb{E}_2[\ln {\cal Z}_{k, t;\epsilon}(\tau)])\Big]
\nonumber \\ +  &
\, \mathbb{E}_1\Big[\frac{s\sqrt{t\, m_k}}{2\sqrt{K\Delta \tau}} \sum_{i} \mathbb{E}_2[\langle X_i\rangle_{k, t;\epsilon}]
\frac{\partial}{\partial \widetilde Z_{i}^{(k, 1)}} \exp(s\,\mathbb{E}_2[\ln {\cal Z}_{k, t;\epsilon}(\tau)])\Big]
\nonumber \\ +  &
\, \mathbb{E}_1\Big[\frac{s\sqrt\epsilon}{2\sqrt{\tau}} \sum_{i}\mathbb{E}_2[\langle X_i\rangle_{k, t;\epsilon}]
\frac{\partial}{\partial \widehat Z_{i}^{(1)}} \exp(s\,\mathbb{E}_2[\ln {\cal Z}_{k, t;\epsilon}(\tau)])\Big].
\end{align*}
Working out the partial derivatives yields
\begin{align*}
 & \hspace{0.05cm}\frac{s^2}{2K\Delta } \sum_{k^\prime > k} \frac{1}{n}\sum_{i\leq j} \mathbb{E}_1\Big[ \mathbb{E}_2[\langle X_i X_j\rangle_{k,t;\epsilon}]^2  
\exp(s \,\mathbb{E}_2[\ln Z_{k, t; \epsilon}(\tau)])\Big]
\nonumber  \\  +    &
\,
\frac{s^2}{2K\Delta} \sum_{k^\prime < k} m_{k^\prime} \sum_{i} \mathbb{E}_1\Big[ \mathbb{E}_2[\langle X_i\rangle_{k,t;\epsilon}]^2  
\exp(s\, \mathbb{E}_2[\ln Z_{k, t; \epsilon}(\tau)])\Big]
\nonumber  \\  +    &
\,
\frac{s^2 (1-t) }{2K\Delta } \frac{1}{n}\sum_{i\leq j} \mathbb{E}_1\Big[ \mathbb{E}_2[\langle X_i X_j\rangle_{k,t;\epsilon}]^2  
\exp(s \,\mathbb{E}_2[\ln Z_{k, t; \epsilon}(\tau)])\Big]
\nonumber  \\  +    &
\,
\frac{s^2t\,m_k}{2K\Delta } \sum_{i} \mathbb{E}_1\Big[ \mathbb{E}_2[\langle X_i\rangle_{k,t;\epsilon}]^2  
\exp(s\, \mathbb{E}_2[\ln Z_{k, t; \epsilon}(\tau)])\Big] 
\nonumber  \\  +    &
\,
\frac{s^2\epsilon}{2} \sum_{i} \mathbb{E}_1\Big[ \mathbb{E}_2[\langle X_i\rangle_{k,t;\epsilon}]^2  
\exp(s\,\mathbb{E}_2[\ln Z_{k, t; \epsilon}(\tau)])\Big] .
\end{align*}
For bounded signals we have $\vert x_i\vert <M$ as well as $m_k\leq M^2$. Thus the sum of these four terms is bounded 
by
\begin{align}
s^2 n \Big(\frac{ 2M^4}{\Delta} + \frac{\epsilon M^2}{2}\Big)\mathbb{E}_1\big[
\exp(s \,\mathbb{E}_2[\ln Z_{k, t; \epsilon}(\tau)])\big]
\end{align}
for all $k=1,\dots, K$. This is an upper bound for the numerator of \eqref{derivative-varphi}, which implies
$\vert \varphi_{k, t;\epsilon}'(\tau)\vert  \leq  s^2 n (2M^4/\Delta  +  \epsilon M^2/2)$. From \eqref{markov}
\begin{align}
\mathbb{P}_{\bZ}\big[ \vert F_{k, t; \epsilon}(\boldsymbol{\Theta}) - \mathbb{E}_{\bZ}[F_{k, t; \epsilon}(\boldsymbol{\Theta})] \vert > u/2  \big]
\leq
2 \exp\Big(s^2 n \Big(\frac{2M^4}{\Delta} + \frac{\epsilon M^2}{2}\Big) - snu/2\Big)
\end{align}
and the best possible value $s =  u (M^4/\Delta  +  \epsilon M^2/2)^{-1}$ yields \eqref{114} and ends the proof.
\end{proof}

The second lemma expresses concentration w.r.t the input signal of the free 
energy averaged over the Gaussian disorder. Recall that $\mathbb{P}_{\bS}$ is the probability law w.r.t the signal realisation.
\begin{lemma}[Concentration w.r.t the signal realisation]\label{conc-free-lemma-2}
Take $P_0$ with bounded support in $[-M, M]$. For all $k = 1, \ldots, K$, $t \in  [0,1]$, and $\epsilon \in  [0,1]$ we have
\begin{align}
 \mathbb{P}_{\bS}\big[ \vert \mathbb{E}_{\bZ}[F_{k, t; \epsilon}(\boldsymbol{\Theta})] - \mathbb{E}[F_{k, t; \epsilon}(\boldsymbol{\Theta})] \vert > u/2  \big] \leq \exp\Big(- \frac{n u^2}{32(\frac{M^4}{\Delta} +  \epsilon M^2)^2}\Big), \label{120}
\end{align}
where $u>0$.
\end{lemma}

\begin{proof}
We first prove a bounded difference property on $\mathbb{E}_{\bZ}[F_{k, t; \epsilon}(\boldsymbol{\Theta})]$ and then apply the McDiarmid inequality \cite{McDiarmid,boucheron2004concentration}.
Let $\bs$ and $\bs^\prime$
two signal realisations that differ at the component $i$ only, i.e. $s_j = s_j^\prime$ for $j\neq i$. We first consider the difference of Hamiltonians corresponding to these two 
realisations. From \eqref{intH}--\eqref{perturb} we have 
\begin{align}
&{\cal H}_{k,t;\epsilon}(\bx;\bz, \widetilde{\bz}, \widehat{\bz}, \bs)  -  {\cal H}_{k,t;\epsilon}(\bx;\bz, \widetilde{\bz}, \widehat{\bz}, \bs^\prime) 
 = 
- \frac{1}{K\Delta n}\sum_{k^\prime>k}\sum_{j=1, j\neq i}^n x_ix_j (s_i-s_i^\prime) s_j 
- \frac{1}{K\Delta n}\sum_{k^\prime>k} x_i^2 (s_i^2-s_i^{\prime 2})
\nonumber \\ &
\qquad\qquad- \frac{1-t}{K\Delta n}\sum_{j=1, j\neq i}^n x_ix_j (s_i-s_i^\prime) s_j 
- \frac{1-t}{K\Delta n} x_i^2 (s_i^2-s_i^{\prime 2}) \nn
&\qquad\qquad- \frac{1}{K\Delta}\sum_{k^\prime<k} m_{k^\prime} x_i(s_i  -s_i^\prime) - \frac{t\,m_k}{K\Delta} x_i(s_i  -s_i^\prime)
- \epsilon x_i(s_i - s_i^\prime).
\end{align}
For a signal distribution with bounded support $[-M, M]$ we get (recall $\vert m_k\vert \leq M^2$)
\begin{align}\label{hamdiff}
\vert {\cal H}_{k,t;\epsilon}(\bx;\bz, \widetilde{\bz}, \widehat{\bz}, \bs)  -  {\cal H}_{k,t;\epsilon}(\bx;\bz, \widetilde{\bz}, \widehat{\bz}, \bs^\prime)\vert \leq  2\Big(\frac{M^4}{\Delta} + \epsilon M^2\Big).
\end{align}
Now set
$g(s_1, \ldots, s_n)  \defeq  \mathbb{E}_{\bZ}[F_{k, t; \epsilon}(\boldsymbol{\Theta})]$.
We have (here $\bX \sim P_0$) 
\begin{align}
&g(s_1, \ldots, s_i, \ldots, s_n) - g(s_1, \ldots, s_i^\prime, \ldots, s_n) 
  = \frac{1}{n} \mathbb{E}_{\bZ}\Big[\ln  \frac{\mathbb{E}_\bX[e^{-{\cal H}_{k,t;\epsilon}(\bX;\bZ, \widetilde{\bZ}, \widehat{\bZ}, 
\bs^\prime)}  ]}{\mathbb{E}_\bX[e^{-{\cal H}_{k,t;\epsilon}(\bX;\bZ, \widetilde{\bZ}, \widehat{\bZ}, \bs)}]}\Big]
\nonumber \\ &
\qquad\qquad = \frac{1}{n} \mathbb{E}_{\bZ}\Big[\ln  \frac{\mathbb{E}_\bX[e^{-{\cal H}_{k,t;\epsilon}(\bX;\bZ, \widetilde{\bZ}, \widehat{\bZ}, \bs)} 
e^{{\cal H}_{k,t;\epsilon}(\bX;\bZ, \widetilde{\bZ}, \widehat{\bZ}, \bs) 
-  {\cal H}_{k,t;\epsilon}(\bX;\bZ, \widetilde{\bZ}, \widehat{\bZ}, \bs^\prime)} ]}{\mathbb{E}_\bX[e^{-{\cal H}_{k,t;\epsilon}(\bX;\bZ, \widetilde{\bZ}, \widehat{\bZ}, \bs)}]}\Big]
\end{align}
and since from \eqref{hamdiff} 
\begin{align}
e^{- 2(\frac{M^4}{\Delta} + \epsilon M^2)}
\leq 
e^{{\cal H}_{k,t;\epsilon}(\bx;\bz, \widetilde{\bz}, \widehat{\bz}, \bs) 
-  {\cal H}_{k,t;\epsilon}(\bx;\bz, \widetilde{\bz}, \widehat{\bz}, \bs^\prime)}
\leq 
e^{2 (\frac{M^4}{\Delta} + \epsilon M^2)}
\end{align}
we readily obtain 
\begin{align}
\vert g(s_1, \ldots, s_i, \ldots, s_n) - g(s_1, \ldots, s_i^\prime, \ldots, s_n)\vert \leq c_i 
\end{align}
with $c_i  =  2(M^4/\Delta + \epsilon M^2)/n$, $i =  1,\ldots, n$. McDiarmid's inequality states that
\begin{align}
\mathbb{P}_\bS\big[\vert g(\bS)  -  \mathbb{E}_\bS[g(\bS)]\vert \geq u/2\big] 
\leq \exp\Big(- \frac{u^2}{8\sum_{i=1}^n c_i^2}\Big)
\end{align}
which here reads \eqref{120}
and ends the proof of the lemma.
\end{proof}

\subsubsection*{Proof of Proposition \ref{conc-free}}
From the triangle inequality and the union bound
\begin{align}
 \mathbb{P}\big[ \vert F_{k, t; \epsilon}(\boldsymbol{\Theta})& - f_{k, t; \epsilon} \vert > u  \big] 
    = 
 \mathbb{P}\big[ \vert F_{k, t; \epsilon}(\boldsymbol{\Theta}) - \mathbb{E}_{\bZ}[F_{k, t; \epsilon}(\boldsymbol{\Theta})]
 + \mathbb{E}_{\bZ}[F_{k, t; \epsilon}(\boldsymbol{\Theta})] - \mathbb{E}[F_{k, t; \epsilon}(\boldsymbol{\Theta})] \vert > u  \big] 
\nonumber \\ & \quad
\leq 
\mathbb{P}\big[ \vert F_{k, t; \epsilon}(\boldsymbol{\Theta}) - \mathbb{E}_{\bZ}[F_{k, t; \epsilon}(\boldsymbol{\Theta})] \vert
+ \vert \mathbb{E}_{\bZ}[F_{k, t; \epsilon}(\boldsymbol{\Theta})] - \mathbb{E}[F_{k, t; \epsilon}(\boldsymbol{\Theta})] \vert > u  \big]
\nonumber \\ & \quad
\leq 
\mathbb{P}\big[ \vert F_{k, t; \epsilon}(\boldsymbol{\Theta}) - \mathbb{E}_{\bZ}[F_{k, t; \epsilon}(\boldsymbol{\Theta})] \vert > u/2\big]
+ \mathbb{P}\big[\vert \mathbb{E}_{\bZ}[F_{k, t; \epsilon}(\boldsymbol{\Theta})] - \mathbb{E}[F_{k, t; \epsilon}(\boldsymbol{\Theta})] \vert > u/2  \big]
\nonumber \\ &\quad
=  \mathbb{E}_{\bS}\mathbb{P}_{\bZ}\big[\vert F_{k, t; \epsilon}(\boldsymbol{\Theta}) - \mathbb{E}_{\bZ}[F_{k, t; \epsilon}(\boldsymbol{\Theta}\big)] \vert > u/2\big]
+ \mathbb{P}_{\bS}\big[\vert \mathbb{E}_{\bZ}[F_{k, t; \epsilon}(\boldsymbol{\Theta})] - \mathbb{E}[F_{k, t; \epsilon}(\boldsymbol{\Theta})] \vert > u/2  \big]
\nonumber \\ & \quad
\leq
2 \exp\Big(- \frac{n u^2}{16(\frac{2M^4}{\Delta} + \frac{\epsilon M^2}{2})}\Big)  +   \exp\Big(- \frac{n u^2}{32(\frac{M^4}{\Delta} + \epsilon M^2)^2}\Big)
\label{explicit}
\end{align}
where the last inequality comes from Lemmas \ref{conc-free-lemma-1} and \ref{conc-free-lemma-2}.
\QEDA

\appendix
\section{Linking the perturbed and plain free energies}\label{appendix-exchg-lim}
The purpose of this appendix is to prove Lemma \ref{thermolimit}. We first note that differentiating the function $\epsilon\mapsto f_{k=1,t=0;\epsilon}$ in \eqref{intf} 
\begin{align}
\frac{d f_{1, 0;\epsilon}}{d\epsilon} =  \frac{1}{n}\sum_{i=1}^n \mathbb{E}\Big[\frac{1}{2}\langle X_i^2\rangle_{1, 0;\epsilon} - \langle X_i \rangle_{1, 0;\epsilon} S_i 
- \frac{1}{2\sqrt\epsilon}\langle X_i\rangle_{1, 0;\epsilon}  \hat Z_i \Big].
\end{align}
By a Gaussian integration by parts the last term becomes
\begin{align}
-\frac{1}{2\sqrt\epsilon}\mathbb{E}[\langle X_i\rangle_{1, 0; \epsilon}  \hat Z_i ] = -\frac{1}{2\sqrt\epsilon}\mathbb{E}[\frac{\partial}{\partial \hat Z_i}\langle X_i\rangle_{1, 0;\epsilon}] 
= -\frac{1}{2}\mathbb{E}[\langle X_i^2\rangle_{1, 0;\epsilon} - \langle X_i\rangle_{1, 0;\epsilon}^2].
\end{align}
By an application of the identity \eqref{nishibasic} we have $\mathbb{E}[\langle X_i\rangle_{1, 0;\epsilon} S_i] = \mathbb{E}[\langle X_i\rangle_{1, 0;\epsilon}^2]$. Therefore we find
\begin{align}
\frac{d f_{1, 0;\epsilon}}{d\epsilon} = - \frac{1}{2n}\sum_{i=1}^n \mathbb{E}[\langle X_i\rangle_{1, 0;\epsilon}^2].
\end{align}
Now by convexity and \eqref{nishibasic} we have $\mathbb{E}[\langle X_i\rangle_{1, 0;\epsilon}^2] \leq \mathbb{E}[\langle X_i^2\rangle_{1, 0;\epsilon}] = \mathbb{E}[S^2]$. Therefore
\begin{align}
\Big\vert \frac{d f_{1, 0;\epsilon}}{d\epsilon} \Big\vert \leq \frac{\mathbb{E}[S^2]}{2}
\end{align}
and the first inequality of the Lemma follows from an application of the mean value theorem.

The second inequality follows from the Lipschitz continuity of the free energy $f_{k=K,t=1;\epsilon}$ of the decoupled scalar system. We refer to \cite{guo2011estimation} for the proof of this standard fact.
\section{Proof of Lemma \ref{tinvar}}\label{appendix-proof-of-lemma3}
The proof of this lemma uses another interpolation:
\begin{align}
\mathbb{E}[\langle q_{\bX,\bS}^{}\rangle_{k,t;\epsilon}]- & \mathbb{E}[\langle q_{\bX,\bS}^{}\rangle_{k,0;\epsilon}] = \int_0^t ds 
\frac{d\,\mathbb{E}[\langle q_{\bX,\bS}^{}\rangle_{k,s;\epsilon}]}{ds} 
\nn
&
= 
\int_0^t ds \mathbb{E}\Big[\big\langle q_{\bX,\bS}^{}\big\rangle_{k,s;\epsilon}\Big\langle\frac{d{\cal H}_{k,s;\epsilon}(\bX;\boldsymbol{\Theta})}{ds}\Big\rangle_{k,s;\epsilon}
- \Big\langle q_{\bX,\bS}^{} \frac{d{\cal H}_{k,s;\epsilon}(\bX;\boldsymbol{\Theta})}{ds}\Big\rangle_{k,s;\epsilon}\Big],
\nn
&=  \int_0^t ds \mathbb{E}\Big[\Big\langle q_{\bX,\bS}^{}\Big(\frac{d{\cal H}_{k,s;\epsilon}(\bX';\boldsymbol{\Theta})}{ds} 
- \frac{d{\cal H}_{k,s;\epsilon}(\bX;\boldsymbol{\Theta})}{ds}\Big)\Big\rangle_{k,s;\epsilon}\Big],\label{E0Es}
\end{align}
where $\bX,\bX', \bX''$ etc are i.i.d replicas distributed according to \eqref{post}. Computations similar to those
in sec.~\ref{proofdH} lead to
%
%
\begin{align}
\mathbb{E}[\langle q_{\bX,\bS}^{}\rangle_{k,t;\epsilon}] - \mathbb{E}[\langle q_{\bX,\bS}^{}\rangle_{k,0;\epsilon}] =  \frac{1}{K} \int_0^t ds \mathbb{E}[\langle q_{\bX,\bS}^{}(g(\bX', \bX'';\bS) - g(\bX,\bX';\bS))\rangle_{k,s;\epsilon}]\label{E0Es_2}
\end{align}
where we define
\begin{align}
g(\bx, \bx';\bs)\defeq \frac{m_{k}}{\Delta}\sum_{i=1}^n\Big(\frac{x_ix_i'}{2} - x_is_i \Big) - \frac{1}{\Delta}\sum_{i\le j=1}^n\Big(\frac{x_ix_j x_i'x_j'}{2n}  - \frac{x_ix_js_is_j}{n} \Big).
\end{align}
Finally from \eqref{E0Es_2} and Cauchy-Schwarz, one obtains
\begin{align}
\big|\mathbb{E}[\langle q_{\bX,\bS}^{}\rangle_{k,t;\epsilon}] - \mathbb{E}[\langle q_{\bX,\bS}^{}\rangle_{k,0;\epsilon}]\big| = {\cal O}\Big(\frac{1}{K} \sqrt{\mathbb{E}[\langle q_{\bX,\bS}^2 \rangle_{k,s;\epsilon}]  \mathbb{E}[\langle g(\bX, \bX';\bS)^2 \rangle_{k,s;\epsilon}]}\Big) = {\cal O}\Big(\frac{n}{K}\Big).\label{E0Es_4}	
\end{align}
The last equality is true as long as the prior $P_0$ has bounded first four moments. We prove this claim now. Let us start by studying $\mathbb{E}[\langle q_{\bX,\bS}^2 \rangle_{k,s;\epsilon}]$. Using Cauchy-Schwarz for the inequality 
and \eqref{nishibasic} for the subsequent equality,
\begin{align}
\mathbb{E}[\langle q_{\bX,\bS}^2\rangle_{k,s;\epsilon}] &= \frac{1}{n^2}\sum_{i,j=1}^n\mathbb{E}[\langle X_iX_jS_iS_j \rangle_{k,s;\epsilon}] \nn
&\le \frac{1}{n^2}\sum_{i,j=1}^n \sqrt{\mathbb{E}[\langle X_i^2X_j^2 \rangle_{k,s;\epsilon}]\mathbb{E}[S_i^2S_j^2 ] } = \frac{1}{n^2}\sum_{i,j=1}^n \mathbb{E}[S_i^2S_j^2 ]={\cal O}(1), 
\label{qsmall}
\end{align}
where the last equality is valid for $P_0$ with bounded second and fourth moments. 
For $\mathbb{E}[\langle g(\bX, \bX';\bS)^2 \rangle_{k,s;\epsilon}]$ we proceed similarly by 
decoupling the expectations using Cauchy-Schwarz and then using \eqref{nishibasic} to make appear 
only terms depending on the signal $\bs$. One finds that under the same conditions on the 
moments of $P_0$ we have $\mathbb{E}[\langle g(\bX, \bX';\bS)^2 \rangle_{k,s;\epsilon}] = {\cal O}(n^2)$. 
Combined with \eqref{qsmall} leads to the last equality of \eqref{E0Es_4} and ends the proof.
\section{Alternative argument for the lower bound}\label{appendix-alternative-route}
We present an alternative useful, albeit not completely rigorous, argument to obtain the lower bound \eqref{llbb}. With enough work the argument can be made rigorous. 
Note that defining
\begin{align}
\widetilde f_{\rm RS}(\{m_k\}_{k=1}^K;\Delta) \defeq \frac{1}{4\Delta K}\sum_{k=1}^K m_{k}^2+f_{\rm den}\big(\Sigma(m_{\rm mf}^{(K)};\Delta)\big) =f_{\rm RS}(m_{\rm mf}^{(K)};\Delta) + \frac{V(\{m_k\})}{4\Delta} , 
\label{tildefrs}
\end{align}
the identity \eqref{fundam_afterConc_aftertinvar_2} is equivalent to
\begin{align}
\int_{a_n}^{b_n}d\epsilon\, f_{1,0;\epsilon} &=\int_{a_n}^{b_n}d\epsilon\,\biggl\{(f_{K_n,1;\epsilon} - f_{K_n,1;0}) + \widetilde f_{\rm RS}(\{m_k^{(n)}\}_{k=1}^{K_n};\Delta)\biggr\} + \mathcal{O}(a_n^{-2}n^{-\alpha}) \nonumber \\
&\ge \int_{a_n}^{b_n}d\epsilon\,(f_{K_n,1;\epsilon} - f_{K_n,1;0}) + \min_{\{m_k\ge0\}_{k=1}^{K_n}} \widetilde f_{\rm RS}(\{m_k\}_{k=1}^{K_n};\Delta) + \mathcal{O}(a_n^{-2}n^{-\alpha}) . \label{similar}
\end{align}
Setting $b_n=2a_n$, taking a sequence $a_n\to 0$ slowly enough as $n\to +\infty$, using Lemma~\ref{thermolimit} and \eqref{f10isf}, we obtain 
\begin{align}\label{otherbound}
 \liminf_{n\to +\infty}f_n \geq \min_{\{m_k\ge0\}_{k=1}^{K_n}}\widetilde f_{\rm RS}(\{m_k\}_{k=1}^{K_n};\Delta).
\end{align}
Simple algebra starting from $\partial_{m_k}\widetilde f_{\rm RS}(\{m_k\}_{k=1}^{K_n};\Delta) = 0$ implies, under the assumption that the extrema are attained at interior points of $\mathbb{R}^{K_n}_+$ (the point to work out to make the argument rigorous), that the 
minimizer of $\widetilde f_{\rm RS}(\{m_k\}_{k=1}^{K_n};\Delta)$ satisfies 
\begin{align}
 m_k = -2\,\partial_{\Sigma^{-2}} f_{\rm den}(\Sigma)|_{\Sigma(m_{\rm mf}^{(K_n)};\Delta)}\,, \qquad k=1, \ldots, K_n\,.
\end{align}
The right hand side is independent of $k$, thus the minimizer is $m_k=m_*$ for $k=1, \cdots, K_n$ where 
\begin{align}
 m_* = -2\, \partial_{\Sigma^{-2}} f_{\rm den}(\Sigma)|_{\Sigma = \sqrt{\frac{\Delta}{m_*}}}\,, \qquad k=1, \ldots, K_n\,.
\end{align}
Thus 
\begin{align}
\min_{\{m_k\ge0\}_{k=1}^{K_n}} \widetilde f_{\rm RS}(\{m_k\}_{k=1}^{K_n};\Delta) = f_{\rm RS}( m_*; \Delta) \geq \min_{m\geq 0} f_{\rm RS}(m;\Delta).
\end{align}
From \eqref{otherbound} we get 
\begin{align}
\liminf_{n\to +\infty}f_n &\ge \min_{m \geq 0 } f_{\rm RS}(m;\Delta)\label{131_}
\end{align}
which is the inequality \eqref{llbb}.

\section{A consequence of Bayes rule}\label{appendix-Nishimori-Bayes}
The purpose of this appendix is to prove the identity \eqref{nishibasic}. Recall that the Gibbs 
bracket $\langle - \rangle_{k,t;\epsilon}$ is the average with respect to the posterior $P_{k,t;\epsilon}(\bx\vert \boldsymbol{\theta})$ where
$\boldsymbol{\theta} \defeq \{\bs, \{\bz^{(k)}, \widetilde{\bz}^{(k)}\}_{k=1}^K, \widehat \bz\}$.
 Using Bayes law we have:
\begin{align}
\mathbb{E}_{\boldsymbol{\Theta}}[\langle g(\bX, \bS)\rangle_{k,t;\epsilon}] =  \mathbb{E}_{\bS} \mathbb{E}_{\boldsymbol{\Theta}\vert \bS} [\langle g(\bX, \bS)\rangle_{k,t;\epsilon}]
= 
\mathbb{E}_{\boldsymbol{\Theta}} \mathbb{E}_{\bS \vert \boldsymbol{\Theta}} [\langle g(\bX, \bS)\rangle_{k, t,\epsilon}].
\end{align}
It remains to notice that 
\begin{align}
\mathbb{E}_{\boldsymbol{\Theta}}\mathbb{E}_{\bS \vert \boldsymbol{\Theta}} [\langle g(\bX, \bS)\rangle_{k,t;\epsilon}] 
= \mathbb{E}_{\boldsymbol{\Theta}}[\langle g(\bX, \bX')\rangle_{k,t;\epsilon}]
\end{align}
where the Gibbs bracket on the right hand side is an average with respect to the product measure of two posteriors 
$P_{k,t;\epsilon}(\bx\vert \boldsymbol{\theta}) P_{k,t;\epsilon}(\bx'\vert \boldsymbol{\theta})$.
\section{A stochastic calculus interpretation}\label{interpretation}

We note that the proofs do not require any upper limit on $K$. This suggests that it is possible to formulate 
the adaptive interpolation method entirely in a continuum language. Here we informally show this for the simplest problem, namely symmetric rank-one matrix 
factorisation, and plan to come back to a rigorous formulation of the continuum formulation in future work.

It is helpful to first write down explicitly the $(k,t)$--interpolating Hamiltonian \eqref{intH} (leaving out the perturbation in \eqref{perturb} which is irrelevant for the argument here)
\begin{align}
\mathcal{H}_{k,t}(\bx;\boldsymbol{\theta})  = \frac{1}{K\Delta}&\sum_{k^\prime = k+1}^K \sum_{i\leq j=1}^n 
\Big( \frac{x_i^2 x_j^2}{2n} - \frac{x_ix_js_is_j}{n} - \sqrt{\frac{K\Delta}{n}} x_i x_j z_{ij}^{(k^\prime)}\Big)
\label{A} \\ &
+ \frac{1}{K\Delta}\sum_{k^\prime =1}^{k-1} m_{k^\prime} \sum_{i=1}^n \Big(  \frac{x_i^2}{2} 
- x_i s_i - \sqrt{\frac{K\Delta}{m_{k^\prime}}} x_i \widetilde z_i^{(k^\prime)}\Big)
\label{B} \\ &
+ \frac{1-t}{K\Delta} \sum_{i\leq j=1}^n \Big( \frac{x_i^2 x_j^2}{2n} - \frac{x_ix_js_is_j}{n} - \sqrt{\frac{K\Delta}{(1-t)n}} x_i x_j z_{ij}^{(k)} \Big)
\label{C} \\ &
+ \frac{t\,m_k}{K\Delta}\sum_{i=1}^n \Big(  \frac{x_i^2}{2} 
- x_i s_i - \sqrt{\frac{K\Delta}{t \,m_{k}}} x_i \widetilde z_i^{(k)}\Big),
\label{D}
\end{align}
and to define the step-wise function 
$m(u)  =  m_{k^\prime}$ for $k'/K \leq u< (k^\prime+1)/K$, $k^\prime  =  1, \dots, K$.

Let us first look at the terms that do {\it not} involve Gaussian noise and become simple Riemann integrals. We have
for the contribution coming from  \eqref{A} and \eqref{C}, 
\begin{align}
& \frac{1}{\Delta}\sum_{i\leq j=1}^n\Big\{
\frac{1}{K}\sum_{k^\prime = k+1}^K  
\Big(\frac{x_i^2 x_j^2}{2n} - \frac{x_ix_js_is_j}{n}\Big) 
+ \frac{1-t}{K} \Big(\frac{x_i^2 x_j^2}{2n} - \frac{x_ix_js_is_j}{n}\Big) \Big\}
\nonumber \\ 
= 
~&\frac{1}{\Delta}\sum_{i\leq j=1}^n\Big\{\int_{\frac{k+1}{K}}^{\frac{K+1}{K}}du\, \Big(\frac{x_i^2 x_j^2}{2n} - \frac{x_ix_js_is_j}{n}\Big) 
+ \int_{\frac{k+t}{K}}^{\frac{k+1}{K}} du \Big(\frac{x_i^2 x_j^2}{2n} - \frac{x_ix_js_is_j}{n}\Big)
\Big\}
\nonumber \\
=
~&\frac{1}{\Delta}\sum_{i\leq j=1}^n \int_{\frac{k+t}{K}}^{\frac{K+1}{K}} du\,  \Big(\frac{x_i^2 x_j^2}{2n} - \frac{x_ix_js_is_j}{n}\Big) .
\label{AC1}
\end{align}
Similarly,  we have for the terms coming from \eqref{B} and \eqref{D},  
\begin{align}
&\frac{1}{\Delta}\sum_{i=1}^n\Big\{\frac{1}{K}\sum_{k^\prime =1}^{k-1} m_{k^\prime}\Big(\frac{x_i^2}{2} 
- x_i s_i\Big) 
+
\frac{t\,m_k}{K} \Big(\frac{x_i^2}{2} 
- x_i s_i \Big) \Big\}\nn
=\,
&\frac{1}{\Delta}\sum_{i=1}^n\Big\{\int_{\frac{1}{K}}^{\frac{k}{K}} du\, m(u) \Big(\frac{x_i^2}{2} 
- x_i s_i\Big) 
+
\int_{\frac{k}{K}}^{\frac{k+t}{K}} du\,m(u) \Big(\frac{x_i^2}{2} 
- x_i s_i \Big) \Big\}
\nonumber \\ 
=\,& 
\frac{1}{\Delta}\sum_{i=1}^n\Big\{\int_{\frac{1}{K}}^{\frac{k+t}{K}} du\, m(u) \Big(\frac{x_i^2}{2} 
- x_i s_i\Big) \Big\}.
\label{BD1}
\end{align}

Now we treat the more interesting contributions {\it involving} the Gaussian noise.
Let $B(u)$ be the Wiener process defined by $B(0)=0$, $\mathbb{E}[B(u)] =0$, 
$\mathbb{E}[B(u)B(v)] = \min(u,v)$ for $u, v \in \mathbb{R}_+$. We introduce independent copies $B_{ij}(u)$,
$i, j = 1,\dots, n$ and consider the sum of increments (also written as an Ito integral)
\begin{align}\label{sum-of-increments}
 \Big\{B_{ij}\Big(\frac{k+1}{K}\Big) - B_{ij}\Big(\frac{k+t}{K}\Big) \Big\} + \sum_{k^\prime=k+1}^K \Big\{ B_{ij}\Big(\frac{k^\prime+1}{K}\Big) 
 - B_{ij}\Big(\frac{k^\prime}{K}\Big) \Big\} 
 =
 \int_{\frac{k+t}{K}}^{\frac{K+1}{K}} dB_{ij}(u).
\end{align}
Since the increments are independent and  $\mathbb{E}[(B(u)  -  B(v))^2]  =  |u - v|$, this is a Gaussian random variable with zero mean and variance $(K + 1  - k  - t)/K$. 
It is therefore equal in distribution to 
\begin{align}
\frac{1}{\sqrt K}\sum_{k^\prime = k+1}^K Z_{ij}^{(k^\prime)} + \sqrt{\frac{1-t}{K}} Z_{ij}^{(k)},
\end{align} 
and the contribution of the (random) Gaussian noise in \eqref{A} and \eqref{C} becomes
\begin{align}\label{AC2}
\frac{1}{\sqrt{\Delta n}}\sum_{i\leq j=1}^n x_i x_j \Big\{\frac{1}{\sqrt K}\sum_{k^\prime = k+1}^K 
  Z_{ij}^{(k^\prime)} + 
 \sqrt{\frac{1-t}{K}} Z_{ij}^{(k)} \Big\}
=
\frac{1}{\sqrt {\Delta n}}\sum_{i\leq j=1}^n \int_{\frac{k+t}{K}}^{\frac{K+1}{K}} dB_{ij}(u) x_i x_j .
\end{align}
To represent the contributions of \eqref{B}, \eqref{D} we introduce independent copies of the Wiener process
$\widetilde{B}_i(u)$, $i=1, \dots, n$ and form the Ito integral
\begin{align}
\sum_{k^\prime=1}^{k-1} \sqrt{m_{k^\prime}} \Big\{\widetilde B_i\Big(\frac{k^\prime +1}{K}\Big) 
- \widetilde B_i\Big(\frac{k^\prime}{K}\Big) \Big\}  +  \sqrt{m_k} \Big\{\widetilde B_i\Big(\frac{k +t}{K}\Big) 
- \widetilde B_i\Big(\frac{k}{K}\Big) \Big\} = \int_{\frac{1}{K}}^{\frac{k+t}{K}}  \sqrt{m(u)} d\widetilde B_i(u)
\label{the-quantity}
\end{align}
which has the same variance than
\begin{align}
\frac{1}{\sqrt K} \sum_{k^\prime=1}^{k-1} \sqrt{m_{k^\prime}} \, \widetilde Z_i^{(k^\prime)} + \sqrt{\frac{t \,m_k}{K}} \,\widetilde Z_i
^{(k)}.
\end{align}
Indeed 
\begin{align}
\frac{1}{K}\sum_{k^\prime =1}^{k-1} m_{k^\prime} + \frac{t \,m_k}{K} = 
\sum_{k^\prime=1}^{k-1} m_{k^\prime} \Big(\frac{k^\prime+1}{K} - \frac{k^\prime}{K}\Big) + m_k \Big(\frac{k+t}{K} - \frac{k}{K}\Big) = \frac{1}{K}
\int_{\frac{1}{K}}^{\frac{k+t}{K}} du\, m(u).
\end{align}
Therefore the contribution of \eqref{B} and \eqref{D} can be represented as
\begin{align}\label{BD2}
\frac{1}{\sqrt\Delta} \sum_{i=1}^n x_i \int_{\frac{1}{K}}^{\frac{k+t}{K}}  \sqrt{m(u)} d\widetilde B_i(u).
\end{align}

Finally, collecting \eqref{AC1}, \eqref{BD1}, \eqref{AC2}, \eqref{BD2}, setting $\tau  \defeq  (t + k)/K$ and  $K \to  \infty$, we obtain a continuous form of the random $(k,t)$--interpolating Hamiltonian,
\begin{align}
\mathcal{H}_\tau(\bx; \bs, \mathbf{B}) = \, &\frac{1}{\Delta} \sum_{i\leq j =1}^n \int_\tau^1 \Big\{
\Big(\frac{x_i^2 x_j^2}{2n} - \frac{x_ix_js_is_j}{n}\Big) du - \sqrt{\frac{\Delta}{n}} x_i x_j dB_{ij}(u) \Big\}
\nonumber \\ 
&\qquad+ \frac{1}{\Delta} \sum_{i=1}^n \int_0^\tau \Big\{\Big(\frac{x_i^2}{2} 
- x_i s_i\Big) m(u) du  -  \sqrt{\Delta m(u)} x_i d\widetilde B_i(u)\Big\}
\label{continuous-interpolating-Ham}
\end{align}
where $m(u)$ is an arbitrary trial function and $\mathbf{B}$ denotes the collection of all Wiener processes. 
Note that $\int_\tau^1du \, B_{ij}(u) = B_{ij}(1) - B_{ij}(\tau)$ which is distributed as $\sqrt{1-\tau}Z_{ij}$ for $Z_{ij}\sim 
\mathcal{N}(0, 1)$, and $\int_0^\tau \sqrt{m(u)}d\widetilde B_i(u)$ is distributed as
$\sqrt{\int_0^\tau m(u)}\widetilde Z_i$ for $\widetilde Z_i\sim\mathcal{N}(0, 1)$. Therefore \eqref{continuous-interpolating-Ham} is equal 
in distribution to 
\begin{align}
&\frac{1}{\Delta}  \sum_{i\leq j =1}^n \Big\{
\Big(\frac{x_i^2 x_j^2}{2n} - \frac{x_ix_js_is_j}{n}\Big)(1-\tau) -  x_i x_j Z_{ij}\sqrt{\frac{\Delta(1-\tau)}{n}} \Big\}
\nonumber \\ 
+ ~  &\frac{1}{\Delta} \sum_{i=1}^n  \Big\{\Big(\frac{x_i^2}{2} 
- x_i s_i\Big) \int_0^\tau m(u) du  -  x_i \widetilde Z_i \sqrt{\Delta \int_0^\tau m(u)du} \Big\}.
\label{continuous-interpolating-Ham-II}
\end{align}
Clearly, the usual Guerra-Toninelli interpolation appears as a special case where one chooses a constant trial function $m(u)=m$ constant.
When we go from \eqref{continuous-interpolating-Ham} to \eqref{continuous-interpolating-Ham-II} we eliminate completely
the Wiener process, however we believe it is useful to keep in mind the point of view expressed by \eqref{continuous-interpolating-Ham} which may turn out to be important for more complicated problems.


Starting from \eqref{continuous-interpolating-Ham} or \eqref{continuous-interpolating-Ham-II} it is possible to evaluate the free energy change along the interpolation path. We define the free energy
\begin{align}
f(\tau) = - \frac{1}{n} \mathbb{E}_{\bS, \mathbf{B}}\big[\ln \mathbb{E}_\bX\big[e^{-\mathcal{H}_\tau(\bX; \bS, \mathbf{B})}\big]\big].
\end{align}
For $\tau = 0$ using we recover the original Hamiltonian $\mathcal{H}_{k=1, t=0}$ (see \eqref{initialmodel}) 
and $f(0)  =  f$ given in \eqref{foriginal}. For $\tau  = 1$ setting $\int_0^1 du \,m(u)  =  m_{\rm mf}$
 we recover the mean-field Hamiltonian 
$\mathcal{H}_{k=K, t=1}$ (see \eqref{HK1_after}) and 
$f(1)  =  f_{\rm den}(\Sigma(\int_0^1 du\, m(u)); \Delta)$. Then proceeding similarly to sec. \ref{proofdH} one finds
the identity
\begin{align}
f &= f_{\rm RS}\Big(\int_0^1d\tau\, m(\tau); \Delta\Big) + \Big\{\int_0^1d\tau\, m(\tau)^2 - \Big(\int_0^1 d\tau\, m(\tau)\Big)^2\Big\}\nn
&\qquad- \frac{1}{4\Delta}\int_0^1 d\tau\, \mathbb{E}_{\bS, \mathbf{B}}\big[\big\langle (q_{\bX, \bS}^{} - m(\tau))^2\big\rangle_\tau\big] + {\cal O}(n^{-1})
\end{align}
where $\langle -\rangle_\tau$ is the Gibbs average w.r.t \eqref{continuous-interpolating-Ham}. 

Of course this immediately gives the upper bound in Proposition \ref{UpperBound}. The matching lower bound is obtained by the same ideas used in the discrete version. 
We briefly review them informally in the continuous language.
One first introduces the $\epsilon$-perturbation term \eqref{perturb} and proves a concentration property for the overlap analogous to Lemma \ref{concentration}. Starting with 
the continuous version of the interpolating Hamiltonian the proof of the free energy concentration is essentially identical (even simpler) than in sec. \ref{concentration-free-energy}, which implies the overlap concentration through sec. \ref{proofConc} that is unchanged. Then, the 
square in the remainder term is approximately equal to $(\mathbb{E}_{\bS, \mathbf{B}}[\langle q_{\bX, \bS}^{}\rangle_{\tau, \epsilon}] - m(\tau))^2$ and we make it vanish by choosing
\begin{align}\label{diff}
m(\tau) = \mathbb{E}_{\bS, \mathbf{B}}[\langle q_{\bX, \bS}^{}\rangle_{\tau, \epsilon}].
\end{align}
This continuous setting thus allows to avoid proving Lemma \ref{tinvar}. This then easily yields the lower bound in Proposition \ref{UpperBound}.
One must still check that \eqref{diff} has a solution. The right hand side is a function $G_{n, \epsilon}(\tau; \int_0^\tau du\, m(u))$ so 
setting $x(\tau) = \int_0^\tau du\, m(u)$, $dx/d\tau = m(\tau)$,  we recognize that \eqref{diff} is a first order differential equation with initial condition $x(0)=0$. The existence of a unique global solution on $\tau\in [0,1]$ is then proved using the Cauchy-Lipschitz theorem. Moreover this solution is differentiable and monotone increasing with respect to $\epsilon$. This last step of the analysis replaces Lemma \ref{freedom}.

\section*{Acknowledgments}
Jean Barbier acknowledges funding by the Swiss National Science Foundation grant no. 200021-156672. 
We thank Thibault Lesieur for providing us the expression of the RS potential for tensor estimation. 
We also acknowledge helpful discussions with Olivier L\'ev\^eque and L\'eo Miolane on the stochastic calculus interpretation and continuous version of Appendix \ref{interpretation}.

%
{
	\singlespacing
	\bibliographystyle{unsrt_abbvr}
	\bibliography{refs}
}
\end{document}